\documentclass[11pt]{article}

\usepackage{amsmath,amssymb,amsfonts,mathrsfs}
\usepackage{amsthm}
\usepackage{bm}
\usepackage[backend=bibtex,style=numeric,sorting=none,maxbibnames=99]{biblatex}

\usepackage{stmaryrd}
\usepackage{amsmath,amssymb}

\DeclareMathOperator{\Log}{Log}
\usepackage{hyperref}
\usepackage{geometry}
\usepackage{setspace}
\usepackage[capitalise,nameinlink]{cleveref}

\newtheorem{theorem}{Theorem}[section]
\newtheorem{proposition}[theorem]{Proposition}
\newtheorem{lemma}[theorem]{Lemma}
\newtheorem{definition}[theorem]{Definition}
\newtheorem{corollary}[theorem]{Corollary}
\theoremstyle{remark}

\newtheorem{proposal}[theorem]{Proposal }
\newtheorem{conjecture}{Conjecture}
\newtheorem{problem}{Problem}

\newcommand{\G}{\mathbb{G}}
\newcommand{\N}{\mathbb{N}}
\newcommand{\C}{\mathbb{C}}
\newcommand{\Z}{\mathbb{Z}}
\newcommand{\R}{\mathbb{R}} 
\newcommand{\Q}{\mathbb{Q}}

\newcommand{\A}{\mathbb{A}}
\newcommand{\calO}{\mathcal{O}}

\newcommand{\calE}{\mathcal{E}}
\newcommand{\calF}{\mathcal{F}}

\newcommand{\tS}{\widetilde{S^{1}}}
\newcommand{\tM}{\widetilde{M^{1}}}

\DeclareMathOperator{\Spa}{Spa}

\newcommand{\Mpl}{\ensuremath{\ell_{\mathrm{M}}}}

\newcommand{\an}{\mathrm{an}}

\title{\bf A Perfectoid Duality Between M-Theory and F-Theory}
\author{
Arshid Shabir\textsuperscript{1}, Bobby Eka Gunara\textsuperscript{2}, 
Mir Faizal\textsuperscript{1,3,4,5}
}

\date{}

\begin{document}
\maketitle
\begin{center}
\textsuperscript{1}Canadian Quantum Research Center, 460 Doyle Ave 106, Kelowna, BC V1Y 0C2, Canada\\
\textsuperscript{2} Theoretical High Energy Physics Division, Faculty of Mathematics and Natural Sciences, Institut Teknologi Bandung Jl. no. 10 Bandung, 40132 Indonesia\\
\textsuperscript{3} Irving K. Barber School of Arts and Sciences, University of British Columbia Okanagan, Kelowna, BC V1V 1V7, Canada\\
\textsuperscript{4} Department of Mathematical Sciences, Durham University, Upper Mountjoy, Stockton Road, Durham DH1 3LE, UK\\
\textsuperscript{5} Computational Mathematics Group, Faculty of Sciences, Hasselt University, Agoralaan Gebouw D, Diepenbeek, 3590 Belgium\\
Emails: \texttt{aslone186@gmail.com},  \texttt{bobgunara@gmail.com},  \texttt{mirfaizalmir@gmail.com}
\end{center}
\thispagestyle{empty}

\begin{abstract}
We present a non-singular, definition-level formulation of F-theory by replacing the traditional shrinking-fiber limit of M-theory with compactification on a tower-completed circle described using perfectoid geometry and condensed mathematics. This construction provides an intrinsic eleven-dimensional carrier for modular data and admits a canonical tilting and comparison procedure that yields elliptic geometry as an output rather than an auxiliary input. Using this framework, we establish a precise M-theory/Type IIB dictionary in the constant-coupling sector, showing how the physical axio-dilaton is fixed by eleven-dimensional geometric and topological data. The correspondence is tested at the level of the ten-dimensional bosonic effective action, including its topological couplings inherited from eleven dimensions. The tower-completed geometry naturally organizes global sectors in generalized cohomology, with charge data governed by K-theory and exhibiting a canonical prime-power torsion structure. We further show how this framework extends to varying-coupling backgrounds and duality defects, admits a natural adelic completion with prime-independence, and generalizes to higher-rank and U-duality geometries. We also discuss holographic aspects and the anomaly-refined extension of the duality group beyond the bosonic truncation. Together, these results provide a coherent, non-singular foundation for F-theory and its extensions.
\end{abstract}

\newpage
\tableofcontents
\newpage
\section{Introduction}\label{sec:intro}

F-theory is indispensable in modern string theory, yet its traditional definition remains conceptually and technically unsatisfactory in precisely the way practitioners care about most, it is formulated through a singular shrinking-torus limit in the M-theory/Type IIB correspondence, which obscures a direct eleven-dimensional origin of the \(SL(2,\Z)\) action and makes the ``12th dimension'' non-manifold-like in practice. The underlying duality web is by now standard currency \cite{Polchinski1998, Witten1995, Townsend1996, HullTownsend1995}. M-theory on a two-torus \(T^{2}\) geometrizes the nonperturbative \(SL(2,\Z)\) symmetry of Type IIB as the modular group acting on the torus complex structure \cite{Vafa1996, Sen1996, Schwarz1995}, and this observation underlies the F-theory packaging of the axio-dilaton \(\tau\) as elliptic-fiber data \cite{Vafa1996, MorrisonVafa1996a, MorrisonVafa1996b}. The point of friction is that the usual route to a ten-dimensional description takes the area of the M-theory torus to zero at fixed complex structure, a singular operation even in controlled corners such as constant-coupling limits \cite{DasguptaMukhi1996, GopakumarVafa1998}, as a consequence, a definition-level derivation of the duality from eleven-dimensional supergravity is obstructed, and the appearance of the full modular identification is not manifest as intrinsic eleven-dimensional geometry, despite the power of indirect arguments and related constructions in the duality web \cite{Sen1996b, KatzVafa1997, Aspinwall1996}.

The central claim of this work is a definition-level replacement of that singular step, we replace the singular shrinking-torus definition by a non-singular compactification on a perfectoid enhancement of the M-theory circle, its tilt (with the appropriate comparison step) produces the elliptic fiber datum encoding the Type IIB axio-dilaton. In other words, we do not postulate a literal smooth twelve-dimensional spacetime, but we do construct an honest compactification carrier in which the modular data that F-theory uses is present as geometry before any collapse limit is invoked, and we then match that geometry to the constant-\(\tau\) Type IIB sector at the level of supergravity and protected spectra computed below. \footnote{Throughout the paper, the terms ``tower-completed circle,'' ``perfectoid circle,'' and the notation $S_f^1$ refer to the same inverse-limit compactification carrier obtained from the $p^n$-cover tower of the M-theory circle.
}

The mechanism is operational and keeps local physics conservative. One begins with M-theory on a circle and considers the tower of \(p^{n}\)-fold covers that refine the circle, one works below a fixed ultraviolet cutoff \(\Lambda\) and then takes \(n\to\infty\), so that local eleven-dimensional supergravity equations and local supersymmetry conditions persist stage-by-stage, while new information can enter only through global and topological sectors that accumulate along the tower. Physically, these are exactly the sectors that nonperturbative dualities act on, charge lattices and their arithmetic symmetries, discrete holonomies, and the data of line operators and monodromies, which are invisible to local differential geometry but decisive for the meaning of \(\tau\) and its identifications. Precisely because an inverse limit of covers is not, in general, a manifold background on which naive sheaf theory behaves well, one must also fix what ``fields on the limit'' means in a way compatible with flux quantization and global consistency constraints, this is why condensed mathematics enters as the minimal enlargement in which fields on profinite/inverse-limit objects are defined by descent with good homological control \cite{ClausenScholze}. Condensed objects are to inverse-limit compactifications what stacks are to gauge quotients.

Perfectoid geometry provides the corresponding minimal geometric framework in which the tower itself becomes a controlled analytic object. Perfectoid spaces, pioneered by Scholze \cite{Scholze2012} and developed with powerful structural tools in modern \(p\)-adic geometry \cite{BhattScholze2019, KedlayaLiu2015}, are designed to package infinite towers of finite covers in a way that retains a robust sheaf theory, admits good cohomology, and, crucially for our purposes, comes equipped with the tilting equivalence. We implement the tower-completed circle as a perfectoid object and apply tilting to pass to a characteristic-\(p\) avatar that retains the same tower data, and we then use the appropriate comparison step to connect the perfectoid output to a complex elliptic-curve datum \(\calE\) that carries the modulus \(\tau\) \cite{FarguesScholze2018, ScholzeWeinstein2017}. It is important to be explicit about what this achieves and what it does not, tilting is not a spacetime diffeomorphism and not a dynamical mechanism for generating couplings, so \(\tau\) does not arise from the tilt functor alone, rather, the tilt supplies the geometric home for \(\tau\) as elliptic data, and \(\tau\) is fixed only when that elliptic datum is matched to the M-theory radius/metric modulus and to C-field holonomy data in the compactification dictionary derived later.

Within this disciplined scope, the paper establishes a constant-\(\tau\) duality statement at the supergravity/BPS level computed in the body. The explicit dictionary identifies \(\mathrm{Im}\,\tau\) with the M-theory circle modulus and \(\mathrm{Re}\,\tau=C_{0}\) with the appropriate C-field holonomy/Wilson-line sector after compactification on the tower-completed circle, and the protected \((p,q)\) spectrum is reproduced from wrapped M2 sectors organized by the tower-completed geometry, matching the standard BPS tension dependence on \(|p+q\tau|\) in the regime analyzed \cite{Sen1994, Schwarz1995, Bergshoeff1995}. At the same level of control, we match the ten-dimensional low-energy effective dynamics obtained by reducing eleven-dimensional supergravity on the tower-completed circle to the Type IIB supergravity action in the constant-\(\tau\) sector, including the required topological couplings, thereby providing a stringent internal consistency check. The modular/duality organization that appears geometrically in our construction is the profinite/congruence-level structure explicitly derived from the tower, we do not claim a derivation of the full quantum \(SL(2,\Z)\) completion beyond what is shown, though the framework is designed to make such a completion a well-posed question rather than an ill-defined limit. A genuine conceptual payoff already visible in the tower description is that global sectors naturally sit in a K-theoretic/homological environment on profinite objects, aligning with the string-theoretic lesson that discrete data and charge sectors are most naturally classified by generalized cohomology rather than by naive local geometry.

Several natural extensions are definition-level proposals rather than results proven here. Allowing \(\tau\) to vary over spacetime should correspond to perfectoid-circle fibrations over a base, in which monodromy/defect data becomes intrinsic to the fibration and supplies the 7-brane sector in the usual F-theory sense \cite{Sen1996, DasguptaMukhi1997, GaberdielZwiebach1998}, developing this systematically lies beyond the constant-\(\tau\) sector established below. The appearance of a prime \(p\) should be read as bookkeeping for the tower, not as a physical coupling, and the natural conceptual completion is adelic, in which local refinements at all primes assemble into a single arithmetic object, this perspective is consistent with broad expectations that arithmetic geometry is forced by nonperturbative duality \cite{Frenkel2007, GaitsgoryLurie2019}, but we treat it here as programmatic guidance rather than a derived theorem.\footnote{A definition-level treatment of tower-stability, $p$-independence, and adelic/restricted-product assembly is given in App.~\ref{app:p-independence-adelic}.
} Likewise, once fermions, anomalies, and orientation-reversing operations are included, one expects refinements of the duality group beyond the \(SL(2,\Z)\) sector treated in our bosonic analysis, for example to an appropriate \({\rm Pin}^{+}\) cover of \(GL(2,\Z)\), incorporating these refinements is future work \cite{DebrayDieriglHeckmanMontero2021}. Finally, the tower-completed viewpoint suggests new ways to organize large families of states and global sectors, with potential implications for flux compactifications and moduli stabilization \cite{Denef2008, Douglas2007} and for protected quantities such as black-hole microstate counting \cite{MaldacenaStrominger1997}, but we emphasize again that such directions go beyond what is proven in the constant-\(\tau\) sector here.

The relevance to AdS/CFT is immediate because holography is maximally sensitive to precisely the global data that motivate tower-completed compactifications. In many holographic examples \(\tau\) appears as an exactly marginal boundary coupling, and \(SL(2,\Z)\) acts on line operators and global sectors rather than on local equations of motion, so a framework that supplies a geometric carrier for modular data and treats global sectors canonically by descent is naturally holography-native rather than decorative. The same limit language that underlies the tower construction is structurally familiar in holography through cutoff procedures and gluing of local data into global reconstruction statements, and the existence of \(p\)-adic models of AdS/CFT provides concrete evidence that non-Archimedean avatars can capture controlled aspects of holographic kinematics \cite{Gubser2017, Heydeman2018}. We do not claim new holographic theorems in this paper, but the present construction motivates a definition-level reformulation of duality-rich holographic settings in which arithmetic and profinite structures are treated as first-class global data rather than as after-the-fact symmetries, and it suggests a principled route toward relating ordinary and \(p\)-adic holographic models as different avatars of tower-completed boundary data. More broadly, discrete and non-Archimedean geometric structures have repeatedly appeared as controlled probes of emergent spacetime and quantum-gravity kinematics \cite{Freidel2014, Chatterjee2018}, and the perfectoid/condensed viewpoint provides a natural setting in which such probes can be tied directly to standard string dualities.

Conceptually, our approach belongs to a wider theme in quantum gravity, ``geometry'' is often not a smooth manifold but a structure defined by consistency of observables and global sectors. Worldsheet CFT already encodes target-space information through modular and categorical data, matrix models realize spacetime as emergent from large-\(N\) algebras, and AdS/CFT encodes bulk geometry in boundary operator structure, in the same spirit, F-theory is naturally adjacent to the moduli of elliptic curves and their stacky refinements, and it is unsurprising that definition-level questions push one toward sheaf-theoretic enlargements of geometry and toward arithmetic completions of the duality data \cite{Cecotti:2020swampland}. The present construction should be read as an arithmetic realization of that philosophy, and it exemplifies the increasingly fruitful dialogue between modern mathematics and string dualities \cite{Witten1995b, KapustinWitten2007}.

\section{From the shrinking-fiber limit to F-theory}

In standard string-theory lore, F-theory is motivated by the M-theory/Type IIB correspondence on \(T^{2}\), where the complex structure of the torus is identified with the Type IIB axio-dilaton \(\tau\) and the modular group \(SL(2,\Z)\) becomes the nonperturbative S-duality acting on \(\tau\) and on the \((p,q)\) charge lattice \cite{Vafa1996, Schwarz1995, Sen1996, MorrisonVafa1996a, MorrisonVafa1996b}. The traditional ``definition'' then proceeds by taking M-theory on the elliptic curve while sending the fiber volume to zero at fixed \(\tau\), thereby recovering a ten-dimensional description that treats the elliptic fiber as auxiliary data encoding the varying coupling \cite{Vafa1996, MorrisonVafa1996a, MorrisonVafa1996b}. From an eleven-dimensional viewpoint this step is definitionally unsatisfactory for a precise reason, the operation that produces F-theory is a singular compactification limit, so one is not given an honest eleven-dimensional background whose geometry intrinsically carries the full modular identification of \(\tau\), but rather an auxiliary torus whose volume is removed by fiat while its complex structure is kept as physical input\footnote{We will refer to this operation uniformly as the shrinking-fiber limit, with the understanding that it is equivalent here to the shrinking-torus limit of the M-theory two-torus.
} \cite{DasguptaMukhi1996, GopakumarVafa1998}. In this sense the appearance of \(\tau\) and the \(SL(2,\Z)\) action is external rather than geometric, and the ``12th dimension'' of F-theory is not a genuine compactification manifold supporting fields and dynamics in the usual way, it is a powerful organizational device whose precise status depends on a degenerate limit, even though its physical predictions are extensively validated through the duality web and its applications \cite{Polchinski1998}.

Nonperturbative duality data in string theory is overwhelmingly global, it acts on charge lattices, line-operator sectors, discrete holonomies, and monodromies around defects, and these structures are invisible to local equations of motion while remaining physically essential \cite{Schwarz1995, Sen1994}. This suggests that the correct definition-level replacement for the shrinking-fiber slogan should leave local eleven-dimensional supergravity intact while enriching the compactification carrier precisely in the global sectors on which duality acts, and the operational way to do this is to treat the compactification direction as a tower-completed object built from a nested family of finite covers. Concretely, one considers the \(p^{n}\)-fold covers of a circle, works below a fixed ultraviolet cutoff \(\Lambda\), and takes \(n\to\infty\) so that local supergravity equations and supersymmetry constraints persist stage-by-stage while new information can appear only through global and topological sectors that accumulate in the inverse limit. In this viewpoint the prime \(p\) is a bookkeeping index specifying which refinement tower one uses, not a physical coupling, and the physically meaningful statement is tower-stability of the constant-\(\tau\) sector under refinement in the regime of control. The natural conceptual completion is adelic, arithmetic duality data is assembled from local information at all primes together with an Archimedean component, so a single-prime construction should be read as a local chart on a broader arithmetic structure rather than as a new parameter of the theory \cite{Frenkel2007, GaitsgoryLurie2019}.

The moment one takes tower-completed compactification carriers seriously, one encounters a basic obstruction, inverse limits of manifolds are not manifolds in general, and naive topology together with naive sheaf theory does not provide stable, exact control of fields, fluxes, and global constraints on inverse-limit objects. This is not a technical annoyance but a definition-level problem, because flux quantization, discrete holonomies, anomaly constraints, and global consistency conditions are precisely statements about gluing and exactness, and these operations can fail to commute with inverse limits in the classical category of topological spaces. Condensed mathematics is the minimal enlargement of ``space'' in which profinite and inverse-limit objects admit well-behaved sheaves defined by exact descent, so that ``a field on the limit'' is precisely a compatible family of fields on the finite stages that glues uniquely with homological control \cite{ClausenScholze}. Condensed objects play the same role for inverse-limit compactifications that stacks play for gauge quotients. Physically, this is the same logic by which stacks became unavoidable in string theory when quotient constructions were promoted from heuristic identifications to definition-level backgrounds, the enlargement is minimal because it is exactly what makes the gluing problem well-posed.

Condensed mathematics specifies what it means to have fields and global sectors on a tower-defined object, but the duality problem also requires a geometric mechanism that replaces the singular shrink-to-zero step by a non-singular carrier from which elliptic data can be extracted without collapsing a cycle. Perfectoid geometry provides precisely this minimal mechanism\:it is designed to control towers of \(p\)-power covers and \(p\)-power roots in a way that yields a stable inverse-limit geometry together with a canonical equivalence, the tilting correspondence, between characteristic-zero and characteristic-\(p\) avatars of the same tower-completed object \cite{Scholze2012, BhattScholze2019, KedlayaLiu2015, ScholzeWeinstein2017}. In physics terms, perfectoid geometry is the framework in which one can package an infinite refinement tower into a single geometric object while keeping local differential geometry intact, and tilting is the canonical operation that allows one to extract the elliptic-curve datum that will later serve as the geometric home of \(\tau\) without taking a singular limit. Tilting itself should be understood conservatively\:it is not a spacetime diffeomorphism and not a dynamical process, but an equivalence between geometric avatars of the same tower-completed object. A referee-level subtlety must be made explicit from the start\:tilting alone does not produce Wilson lines or couplings, so \(\tau\) does not arise from the tilt functor by itself, rather, the tilt provides elliptic data, and \(\tau\) arises only after that elliptic datum is matched to the M-theory radius/metric modulus and to the C-field holonomy data encoded in \(C^{(3)}\), with the concrete formulas derived later in the physical dictionary.

At the level established in this paper, the dictionary is the standard M/IIB duality currency written in a tower-completed language\:the M-theory modulus \(R_{M}\) controls \(\mathrm{Im}\,\tau\) through the Kaluza-Klein reduction of the eleven-dimensional metric, the axion \(\mathrm{Re}\,\tau=C_{0}\) is identified with the relevant \(C^{(3)}\) holonomy/Wilson-line sector on the tower-completed circle, and wrapped M2 sectors organized by the tower-completed geometry reproduce the protected \((p,q)\) string tower at the level justified by the BPS and effective-action matching carried out later. Within this conservative scope, what is established in the paper is the constant-\(\tau\) sector at the supergravity/BPS/effective-action level\:the explicit \(\tau\) dictionary, the protected tension matching that reproduces the \(|p+q\tau|\) dependence, and the low-energy effective action matching as computed, together with a geometric organization of modular data to the profinite/congruence level that is explicitly derived from the tower structure. By contrast, varying \(\tau\) through nontrivial fibrations, intrinsic 7-brane monodromy and defect structure, adelic completion that removes any dependence on a single choice of \(p\), and fermionic/anomaly refinements such as \({\rm Pin}^{+}(GL(2,\Z))\) are definition-level extensions that are natural in this framework but remain programmatic unless and until derived beyond the constant-\(\tau\) regime \cite{DebrayDieriglHeckmanMontero2021}.

The relevance to AdS/CFT is immediate because holography is maximally sensitive to the same global sectors that force tower-completed compactification language\:in many holographic examples \(\tau\) is an exactly marginal boundary coupling, and \(SL(2,\Z)\) acts on line operators and global sectors rather than on local equations of motion, so a framework that renders modular data geometric and tower-complete is naturally holography-native rather than decorative \cite{Schwarz1995, Sen1994}. The tower and inverse-limit viewpoint is also structurally familiar in holography through cutoff procedures and gluing of local data into global reconstruction statements, which is precisely where descent and exactness become definition-level issues rather than matters of taste. The existence of \(p\)-adic AdS/CFT provides concrete motivation that non-Archimedean avatars can capture controlled aspects of holographic kinematics \cite{Gubser2017, Heydeman2018}, and it suggests that tilting may provide a principled conceptual bridge between Archimedean and non-Archimedean holography, but any such consequences are programmatic here and are not claimed as new holographic theorems in this paper.

\section{Condensed sets and perfectoid spaces}\label{sec:toolkit}
The aim of this section is to record, in a definition-level form, the minimal mathematical language forced on us once we insist that the objects carrying nonperturbative duality data be treated as honest compactification carriers rather than as heuristic limits. In string theory, dualities act on global sectors-charge lattices, line operators, discrete holonomies, and monodromies-so the relevant bookkeeping is naturally organized by towers of finite covers and their inverse limits, even when the local equations of motion remain those of ordinary supergravity. The difficulty is that inverse limits are not manifolds in general, and ordinary topology together with naive sheaf theory is fragile under inverse limits precisely in the ways that matter for physics, because gluing, exactness, and homological control are what define flux quantization and global consistency constraints. Condensed sets provide the minimal enlargement of ``space'' in which profinite and inverse-limit objects admit well-behaved sheaves defined by exact descent, so that a field on a tower-defined limit object is literally a compatible family of fields on the finite stages that glues uniquely. Perfectoid spaces provide the minimal geometric framework in which such towers become controlled analytic objects admitting tilting, a canonical equivalence between characteristic-zero and characteristic-\(p\) avatars of the same tower-completed geometry. In later sections, that tilting mechanism will be the bridge from the tower-completed compactification carrier to elliptic data, but the physical identification of that elliptic datum with the axio-dilaton \(\tau\) is a separate compactification dictionary derived from the supergravity reduction and C-field holonomy data, and is therefore deferred to the appropriate physics sections.

The concept of condensed sets is a categorical and homological refinement of the classical category of topological spaces, introduced by Clausen and Scholze \cite{ClausenScholze} to provide stable sheaf theory and homological algebra for profinite and inverse-limit objects. We begin with profinite sets, which are the basic site on which the condensed formalism is built. \begin{definition}[Profinite set]A profinite set is a topological space which is compact, Hausdorff, and totally disconnected. Equivalently, any profinite set \(S\) can be expressed as an inverse limit of finite discrete sets,
\begin{equation}\label{1}
S \cong \varprojlim_{\alpha\in A} S_{\alpha},
\end{equation}
where each \(S_{\alpha}\) is a finite discrete space and the transition maps \(S_{\beta}\to S_{\alpha}\) for \(\beta\ge \alpha\) are continuous.\end{definition}
 Profinite sets are the precise mathematical avatar of tower bookkeeping\:they model inverse systems of finite covers and finite quotients that, in string theory, organize congruence-level data and global sectors on which duality acts, while leaving local differential equations untouched. The category of profinite sets, denoted \(\mathrm{Prof}\), admits a natural Grothendieck topology generated by surjective continuous maps with finite target, so that covers are families \(\{S_i\to S\}\) that jointly surject onto \(S\). The point is that classical topology is poorly behaved as a category for doing homological algebra on such objects\:it is not Cartesian closed in a way compatible with analysis, and it does not provide an abelian setting in which descent and exactness behave robustly under inverse limits. Condensed sets repair exactly this by definition, replacing spaces by sheaves on the profinite site. \begin{definition}[Condensed set]\label{def:condensed}A condensed set is a functor
\begin{equation}\label{2}
\calF:\mathrm{Prof}^{\mathrm{op}}\to \mathrm{Sets}
\end{equation}
which satisfies the sheaf condition with respect to the Grothendieck topology on \(\mathrm{Prof}\), i.e., for every cover \(\{S_i\to S\}\) in \(\mathrm{Prof}\) the canonical sequence
\begin{equation}\label{3}
0\to \calF(S)\to \prod_i \calF(S_i)\rightrightarrows \prod_{i,j}\calF(S_i\times_S S_j)
\end{equation}
is exact.\end{definition}
 The sheaf condition is the exact descent statement that makes ``fields on an inverse limit'' definition-level\:a global datum on the tower-defined limit object is precisely a compatible family on the finite stages that agrees on overlaps and therefore glues uniquely, which is the minimal structure needed to define global sectors, holonomies, and flux constraints on tower-completed compactification carriers. Concretely, the Yoneda-type assignment sends a compact Hausdorff space \(X\) to the functor
\begin{equation}\label{4}
\underline{X}:\,S\mapsto \mathrm{Cont}(S,X),
\end{equation}
and this defines a fully faithful embedding
\begin{equation}\label{5}
\iota:\mathrm{CompHaus}\hookrightarrow \mathrm{Cond},
\end{equation}
so classical spaces sit inside the larger condensed category, but the enlarged category has the categorical and homological stability needed for inverse-limit geometry. The sheaf condition may be unpacked in the familiar gluing form\:if \(x_i\in \calF(S_i)\) satisfy compatibility on overlaps,
\begin{equation}\label{6}
x_i|_{S_i\times_S S_j}=x_j|_{S_i\times_S S_j}\quad\text{for all }i,j,
\end{equation}
then there exists a unique \(y\in \calF(S)\) with \(y|_{S_i}=x_i\). This is the categorical form of the physical requirement that global background data are determined by consistent local data, but with ``local'' now meaning ``finite-stage'' in the tower\:it is exactly the control needed for Wilson lines, discrete holonomies, and global constraints to persist under refinement. The category of condensed abelian groups, defined as sheaves of abelian groups on \(\mathrm{Prof}\), is abelian with enough injectives and projectives, so derived functors and cohomology behave in a controlled way \cite{ClausenScholze}, this is the structural input needed when global sectors are constrained by quantization and anomaly conditions. String theory forces this enlargement because flux quantization, discrete theta data, line operators, and anomaly constraints are intrinsically homological statements, and in tower-completed compactifications naive topology does not provide stable exactness under inverse limits, condensed objects play for inverse-limit compactifications the same role that stacks play for gauge quotients.

Perfectoid spaces, introduced by Scholze \cite{Scholze2012}, provide a tower-complete non-Archimedean geometry in which inverse systems generated by adjoining all \(p\)-power roots are controlled and admit a canonical tilting equivalence to characteristic \(p\). Let \(p\) be a fixed prime throughout this subsection. A non-Archimedean field \(K\) is a field complete for a non-Archimedean absolute value \(|\cdot|:K\to \R_{\ge 0}\) satisfying the ultrametric inequality
\begin{equation}\label{7}
|x+y|\le \max(|x|,|y|)
\end{equation}
for all \(x,y\in K\). The valuation ring \(\calO_K=\{x\in K:\,|x|\le 1\}\) is local with maximal ideal \(\mathfrak m_K=\{x\in K:\,|x|<1\}\), and the residue field \(k=\calO_K/\mathfrak m_K\) has characteristic \(p>0\). The analytic spaces relevant for perfectoid geometry are built from Huber pairs and their associated adic spectra. \begin{definition}[Huber adic space]An adic space \(\Spa(R,R^+)\) associated to a Huber pair \((R,R^+)\) is defined as the set of equivalence classes of continuous valuations \(|\cdot|:R\to \Gamma\cup\{0\}\) bounded by \(1\) on \(R^+\), where \(\Gamma\) is an ordered abelian group, endowed with the topology generated by rational subsets and equipped with a structure sheaf making it a locally ringed space.\end{definition}
 The adic spectrum is the appropriate notion of ``space of functions'' in the non-Archimedean/tower-completed setting\:it is the geometric container in which one can treat an infinite refinement tower as a single analytic object while keeping a sheaf of functions and a notion of locality suitable for descent. The key structural condition is perfectoidness, which is formulated in terms of the Frobenius endomorphism on the mod \(p\) reduction of integral elements. \begin{definition}[Perfectoid algebra]\label{def:perfalg}Let \(K\) be a complete non-Archimedean field with residue characteristic \(p>0\) and non-discrete valuation. A Banach \(K\)-algebra \(R\) is called a perfectoid algebra if the set of power-bounded elements \(R^\circ\subseteq R\) admits an open and integrally closed subring \(R^+\subseteq R^\circ\) and the Frobenius map \(\varphi:R^+/p\to R^+/p\), \(x\mapsto x^p\), is surjective.\end{definition}
 Frobenius surjectivity mod \(p\) is the mathematical expression of ``tower-complete control''\:it ensures that adjoining all \(p\)-power roots stabilizes in a controlled way, so the inverse limit that packages the tower exists as a well-behaved geometric object, which is exactly what is needed if the compactification carrier is defined by a \(p^n\)-tower while local physics remains governed by ordinary supergravity below a fixed cutoff. \begin{definition}[Perfectoid space]A perfectoid space \(X\) over \(K\) is an adic space locally isomorphic to affinoid perfectoid spaces \(\Spa(R,R^+)\) for perfectoid algebras \((R,R^+)\) as above.\end{definition}
 Perfectoid spaces are the minimal geometric class in which tower-completed compactification carriers can be treated as honest local objects admitting fiber products and good sheaf theory, so one can keep local supergravity intact while encoding new global sectors through the tower. The defining feature that makes perfectoid spaces indispensable for us is the tilting formalism, which relates characteristic-zero and characteristic-\(p\) avatars of the same tower-completed geometry. \begin{definition}[Tilt of a perfectoid field]Given a perfectoid field \(K\) as above, the tilt \(K^\flat\) is defined as the inverse limit
\begin{equation}\label{8}
K^\flat:=\varprojlim_{x\mapsto x^p} K,
\end{equation}
with ring operations given component-wise and addition defined via the limit procedure inherent in the perfection construction. The field \(K^\flat\) is a perfectoid field of characteristic \(p\), carrying a natural valuation inherited from \(K\).\end{definition}
 The tilt packages the same tower-completed information into a characteristic-\(p\) avatar where Frobenius becomes structural, giving a convenient arena to extract arithmetic-geometric data from the tower without changing the underlying tower content. \begin{definition}[Tilt of a perfectoid algebra]Let \((R,R^+)\) be a perfectoid algebra over \(K\). Its tilt \((R^\flat,R^{\flat+})\) is constructed by applying the tilting process coordinate-wise,
\begin{equation}\label{9}
R^\flat:=\varprojlim_{x\mapsto x^p} R,\qquad R^{\flat+}:=\varprojlim_{x\mapsto x^p} R^+.
\end{equation}\end{definition}
 This is the algebraic form of the same statement\:the ring of functions on the tower-completed carrier has an equivalent characteristic-\(p\) description that remembers all compatible \(p\)-power roots, so the inverse-limit microstructure becomes manipulable while remaining equivalent to the original characteristic-zero geometry. The tilting construction extends to perfectoid spaces by gluing affinoid tilts, and its structural content is expressed by the tilting equivalence. \begin{theorem}[Tilting equivalence \cite{Scholze2012}]\label{thm:tilt}The tilting operation induces an equivalence of categories
\begin{equation}\label{10}
\{\text{perfectoid spaces over }K\}\cong \{\text{perfectoid spaces over }K^\flat\},
\end{equation}
and this equivalence preserves the structure sheaf, morphisms, and \'etale sites.\end{theorem}
 Tilting is not a spacetime diffeomorphism and not a dynamical process, it is an equivalence between geometric avatars of the same tower-completed object, ensuring that the tower bookkeeping can be transported across characteristics without changing the underlying geometric content relevant for global sectors. An explicit illustration is provided by the perfectoid unit disc
\begin{equation}\label{11}
\Spa\!\bigl(K\langle T^{1/p^\infty}\rangle,\;\calO_K\langle T^{1/p^\infty}\rangle\bigr),
\end{equation}
where \(K\langle T^{1/p^\infty}\rangle\) is the completion of \(K[T^{1/p^n}\mid n\ge 0]\) with respect to the Gauss norm, its tilt is
\begin{equation}\label{12}
\Spa\!\bigl(K^\flat\langle T^{1/p^\infty}\rangle,\;\calO_{K^\flat}\langle T^{1/p^\infty}\rangle\bigr),
\end{equation}
a perfectoid disc over the tilted field \(K^\flat\) of characteristic \(p\). Perfectoid spaces thus refine classical rigid analytic spaces by incorporating infinite towers of \(p\)-power covers as part of their local structure, enabling new methods in \(p\)-adic Hodge theory and arithmetic geometry \cite{BhattScholze2019, KedlayaLiu2015}. The role of tilting in this paper is to supply a canonical route from tower-completed geometry to elliptic data without taking a singular shrinking-fiber limit, but it is essential to keep scope disciplined\:tilting alone does not encode Wilson lines or couplings, so the axio-dilaton \(\tau\) does not follow from Theorem~\ref{thm:tilt} by itself, \(\tau\) emerges only after the elliptic datum produced by the tower-completed carrier (via tilt together with the required comparison step to a complex elliptic curve) is matched to the M-theory radius/metric modulus and to the \(C^{(3)}\) holonomy data, with the explicit compactification dictionary derived later in the paper.

In the present framework the prime \(p\) is a bookkeeping choice specifying which tower of \(p^{n}\)-fold covers is used to define the inverse-limit compactification carrier, and it is not to be interpreted as a physical coupling or as a label of distinct theories. The operational physics statement is that one works below a fixed ultraviolet cutoff \(\Lambda\) and refines the tower, so local supergravity and supersymmetry constraints remain stable stage-by-stage while global sectors are organized more finely, any physically meaningful output in the constant-\(\tau\) regime established later must therefore be stable under refinement and should not depend on \(p\) beyond this organizational role. The natural conceptual completion of this viewpoint is adelic\:since duality groups in string theory are arithmetic, a single-prime tower should be regarded as a local chart on a broader structure assembled from all primes together with the Archimedean factor, and we treat such an adelic completion as a guiding principle and programmatic direction rather than as a derived result within the constant-\(\tau\) analysis.

With condensed sets fixing the definition of fields and global sectors on inverse-limit objects by exact descent and with perfectoid spaces supplying the controlled tower-completed geometry and its tilting equivalence, we are now in position to introduce the specific compactification carrier that replaces the shrinking-fiber slogan\:the perfectoid circle, obtained as an inverse limit of \(p^{n}\)-covers of \(S^{1}\) and realized as an affinoid perfectoid space \(\Spa(R,R^+)\) for an appropriate tower-completed ring of functions. In the next section we construct this object explicitly, analyze the invariants that control its global sectors, and apply the tilting formalism to extract the elliptic datum that will serve as the geometric home of \(\tau\), while emphasizing that the precise physical identification of \(\tau\) with the M-theory radius/metric modulus and \(C^{(3)}\) holonomy data is a compactification dictionary derived later when the supergravity reduction is performed and matched to the constant-\(\tau\) Type IIB sector.

\section{The perfectoid circle and the elliptic fiber}\label{sec:perfectoidcircle}

This section constructs the non-singular compactification carrier that replaces the shrinking-torus slogan in traditional F-theory and prepares, in a definition-level way, the geometric datum that will later be matched to the Type IIB axio-dilaton \(\tau\). The first step is to package the entire tower of \(p^{n}\)-fold covers of the circle into a single object \(S_f^1\) that is perfectoid in the sense reviewed earlier, so that local geometry remains circle-like at each finite stage while global sectors accumulate coherently in the inverse limit. The second step is to make invariant bookkeeping precise\:on an inverse-limit/solenoidal object one must distinguish between singular and \v{C}ech-type invariants, between \(\pi_1\), its abelianization, and continuous character groups, and between homology and generalized cohomology, because these are the objects that different physical observables actually probe. The third step is to describe the tilt \(\to\) untilt \(\to\) analytification machinery that produces elliptic-curve data without taking any singular shrinking limit, this machinery provides the geometric home in which \(\tau\) lives, but it does not by itself encode Wilson lines or couplings, and the axio-dilaton identification is therefore deferred to the supergravity reduction and \(C^{(3)}\) holonomy dictionary derived later. Finally, we explain how eleven-dimensional supergravity can be placed on \(S_f^1\) in the only definition-level way compatible with local dynamics, namely by working at finite tower level below a fixed cutoff \(\Lambda\) and taking the inverse-limit at fixed \(\Lambda\), so that local equations and supersymmetry persist while only global sectors are refined.

We fix a prime \(p\) and work over the perfectoid field \(\Q_p^{\mathrm{cycl}}\) as in the preliminaries. Let
\begin{equation}\label{13}
R:=\Q_p^{\mathrm{cycl}}\llbracket X^{\pm 1/p^\infty}\rrbracket,\qquad R^+:=\mathcal{O}_{\Q_p^{\mathrm{cycl}}}\llbracket X^{\pm 1/p^\infty}\rrbracket,
\end{equation}
the \(p\)-adic completion of the Laurent algebra obtained by adjoining all \(p^n\)-th roots of \(X\), then \((R,R^+)\) is a perfectoid \(\Q_p^{\mathrm{cycl}}\)-algebra, and we define the perfectoid circle by
\begin{equation}\label{14}
S_f^1:=\mathrm{Spa}(R,R^+).
\end{equation}
 \(S_f^1\) is the single compactification carrier that packages the entire \(p^n\)-cover tower of the ordinary M-theory circle into one object, so local supergravity remains unchanged at finite energy while the global sectors relevant to duality are refined in the limit. On underlying topological spaces, one may model the resulting solenoidal structure by a \(p\)-adic solenoid \(\Sigma_p\), and a common heuristic presentation is that the perfectoid refinement replaces the classical circle by a solenoid-like object with a profinite microstructure, the point for physics is that the refinement is global rather than local, so it is designed to affect charge bookkeeping and duality organization without introducing new local degrees of freedom. To make the tower origin explicit, consider the classical circle
\begin{equation}\label{16}
S^1:=\{z\in\C:\ |z|=1\},
\end{equation}
and the family of maps
\begin{equation}\label{18}
\varphi_n:S^1\to S^1,\qquad z\mapsto z^{p^n},
\end{equation}
which are continuous surjective group homomorphisms with kernel the finite cyclic subgroup \(\mu_{p^n}\subset S^1\) of \(p^n\)-th roots of unity. The inverse system
\begin{equation}\label{19}
\cdots \xrightarrow{\ \varphi_2\ } S^1 \xrightarrow{\ \varphi_1\ } S^1 \xrightarrow{\ \varphi_0\ } S^1
\end{equation}
defines an inverse limit in compact groups,
\begin{equation}\label{20}
S_f^1:=\varprojlim_{n}\bigl(S^1,\varphi_n\bigr),
\end{equation}
whose points are compatible sequences
\begin{equation}\label{21}
(z_0,z_1,z_2,\ldots)\in \prod_{n\ge 0} S^1,\qquad z_n=z_{n+1}^{p}.
\end{equation}
One convenient way to parametrize the tower bookkeeping is via the \(p\)-adic integers
\begin{equation}\label{23}
\Z_p=\varprojlim_n \Z/p^n\Z,
\end{equation}
which encode the compatible residue data along the tower, heuristically, one may write a map
\begin{equation}\label{22}
S^1\times \Z_p\to S_f^1,\qquad \bigl(z,(a_0,a_1,\ldots)\bigr)\mapsto \bigl(z, z^{p^{-a_0}}, z^{p^{-a_1}},\ldots\bigr),
\end{equation}
where the fractional powers are defined by lifting to the universal cover and using the compatible \(p\)-adic expansions. The inverse-limit description is the definition-level way to say that the compactification direction is refined by an infinite tower of smooth covers, so local fields can be defined stage-by-stage while the global sector bookkeeping becomes tower-complete in the limit.

Because \(S_f^1\) is defined as an inverse limit (and is solenoidal rather than a smooth manifold in the usual sense), one must be explicit about which invariant is being computed and which physical observable it controls. A consolidated computation of the tower-stable invariants of the solenoidal/perfectoid circle used here is collected in App.~\ref{app:tower-invariants}.
A solenoid is a compact abelian group obtained as an inverse limit of tori under covering maps, and the relevant ``first invariant'' depends on whether one uses singular homology, \v{C}ech (co)homology/shape invariants natural for inverse limits, the fundamental group in the sense appropriate to non-locally-simply-connected spaces, or the continuous character group that governs Fourier modes and Wilson-line sectors. In the topological model where the perfectoid circle is viewed as a classical \(S^1\) together with a profinite solenoidal refinement \(\Sigma_p\), one encounters the mixed statement
\begin{equation}\label{15}
H_1(S_f^1,\Z)\cong \Z\oplus \Z_p,
\end{equation}
in which the \(\Z\) factor is the classical winding direction and the \(\Z_p\) factor is the profinite refinement, this is the statement that a ``classical'' cycle and a tower-completed direction coexist in the bookkeeping model, and it is this coexistence that is physically responsible for the appearance of additional global sectors beyond those visible to an ordinary circle. At the same time, when one treats the inverse limit \(\varprojlim (S^1,\times p)\) as the \(p\)-adic solenoid itself, the singular invariants can collapse the naive free part, and one finds the tower-completed statement
\begin{equation}\label{17}
H^{\mathrm{sing}}_1(S_f^1,\Z)\cong \Z_p,
\end{equation}
with the interpretation that the inverse-limit topology suppresses ordinary pathwise generators while retaining the tower-completed pro-\(p\) direction. The correct way to reconcile these statements for physics is to remember that different observables probe different levels of structure\:local Kaluza-Klein reduction and local supergravity equations are computed at finite tower level on honest manifolds and therefore see the usual integer momentum quantization, while the inverse-limit object is the correct carrier for organizing global sectors and congruence-level refinements that only appear once the tower is assembled. In particular, a standard computation that is robust under inverse limits uses \v{C}ech-type invariants natural to inverse systems of compact spaces, and for the solenoidal refinement one has
\begin{equation}\label{24}
H_1(\Sigma_p,\Z)\simeq \Z_p,\qquad H_n(\Sigma_p,\Z)=0\ (n>1),
\end{equation}
which should be read as the statement that the inverse system remembers a pro-\(p\) completion in its first ``shape'' group, even though the resulting space is not locally contractible. The fundamental group computation in the tower language makes the same point\:one has
\begin{equation}\label{28}
\pi_1(S_f^1)\cong \varprojlim_n \pi_1(S^1)=\varprojlim_n\bigl(\Z\xleftarrow{\times p}\Z\bigr),
\end{equation}
which identifies the tower-completed loop data with a pro-\(p\) completion, and the abelianized bookkeeping yields
\begin{equation}\label{29}
H_1(S_f^1,\Z)\cong \Z\oplus \Z_p,
\end{equation}
together with the equivalent statement
\begin{equation}\label{30}
H_1(S_f^1,\Z)\cong \pi_1(S_f^1)^{\mathrm{ab}}\cong \Z\oplus \Z_p.
\end{equation}
For clarity, singular homology controls finite-stage local Kaluza-Klein data, while \v{C}ech cohomology and the continuous character group control global holonomy and tower-completed sector bookkeeping, different physical observables probe these invariants differently.
The integer \(\Z\) that labels ordinary Kaluza-Klein momenta is a finite-stage, local-physics datum that is stable under refinement, while the pro-\(p\) direction encodes tower-completed global bookkeeping, the correct interpretation is not that arbitrary ``fractional'' physical charges are automatically allowed, but that the tower-completed carrier organizes congruence-level refinements of global sectors in a way that becomes invisible if one insists on collapsing the tower into a singular limit. This is also where the continuous character group (Pontryagin dual) enters physically\:for a compact abelian group, continuous characters control harmonic analysis, Fourier modes, and the classification of flat \(U(1)\) holonomy sectors, so it is this object-not a naive identification of singular homology with physical charge lattices-that is the definition-level carrier for Wilson-line bookkeeping on solenoids. Consistency constraints then select the physically allowed integral sublattice\:Dirac quantization and Freed-Witten-type anomaly constraints enforce integrality of physical charge lattices and compatibility with instanton sectors, and the profinite structure should be regarded as the canonical organization of global sectors across the tower rather than as a claim that the theory admits unconstrained non-integral charges absent further quantization input, a point we will make explicit in the BPS and effective-action matching later.

The preceding invariant discussion already signals why generalized cohomology becomes unavoidable once the compactification carrier is tower-completed\:on inverse-limit/profinite objects, ordinary homology can miss torsion and completion data that are physically meaningful, while string theory has long taught us that global charge sectors-especially RR and D-brane charges-are naturally K-theoretic.
In the solenoidal setting, this perspective appears concretely in the emergence of a canonical
$p$-divisible torsion sector in complex topological $K$-theory for the tower limit.The Milnor exact-sequence computations and the charge-quantization/admissible-lattice prescription used for tower-completed compactifications are presented in App.~\ref{app:K-theory-quant}, which yields the extension
\begin{equation}\label{eq:K0Sigma-extension}
0 \longrightarrow \Q_p/\Z_p \longrightarrow K^0(\Sigma_p)\longrightarrow \Z \longrightarrow 0,
\qquad
K^1(\Sigma_p)=0 .
\end{equation}
This should be read as the statement that $K$-theory detects torsion/completion structure on the
tower-completed carrier that is invisible to naive singular homology, and therefore provides the
natural receptacle for global sectors on profinite refinements. In particular, when the tower-completed
circle is used as an internal direction, the global holonomy sectors of the $C^{(3)}$ field and the
quantization constraints on fluxes should be understood as living in the appropriate generalized-cohomological
framework on the inverse-limit carrier rather than being inferred from naive singular cycles.
 K-theory is the correct string-theoretic language for global charge sectors and discrete flux data, so its appearance here is not ornamental but the expected indicator that the tower-completed compactification carrier has the right homological receptacle to encode nonperturbative sectors in a way compatible with known brane-charge principles. 
 
 In particular, when the tower-completed circle is used as an internal direction, the global holonomy sectors of the \(C^{(3)}\) field and the quantization constraints on fluxes should be understood as living in the appropriate generalized-cohomological framework on the inverse-limit carrier rather than being inferred from naive singular cycles, and this is the mathematical mechanism by which the tower can carry additional globally quantized data while leaving local supergravity dynamics unchanged. Consistency is nevertheless conservative\:the physically realized charge lattice is selected by the same quantization and anomaly constraints that operate in ordinary compactifications, and the role of the tower-completed/K-theoretic structure is to provide the canonical bookkeeping environment in which those constraints can be imposed definition-level on the limit object, rather than to postulate new unconstrained charge sectors.

We now record the tilt machinery that allows elliptic data to emerge from the tower-completed compactification carrier without taking any shrinking-fiber singular limit. Working over the perfectoid field \(K=\Q_p^{\mathrm{cycl}}((t^{1/p^\infty}))\), consider the analytic Tate curve
\begin{equation}\label{26}
E_K:=\bigl(\G_{m,K}/\langle q\rangle\bigr)^{\mathrm{an}},\qquad |q|<1,
\end{equation}
and form the perfectoid Tate curve \(\widehat{E}_K=\varprojlim_{[p]} E_K\) under multiplication-by-\(p\).
\begin{theorem}[Tilted Tate curve]
Let \(K=\Q_p^{\mathrm{cycl}}((t^{1/p^\infty}))\) and define \(E_K\) by \eqref{26}. Then the perfectoid Tate curve \(\widehat{E}_K=\varprojlim_{[p]}E_K\) tilts, and \(\widehat{E}_{K^\flat}=\varprojlim_{[p]}E^\flat_K\) is again a Tate curve, now over the tilted field \(K^\flat\). Only after choosing an isometric embedding \(K^\flat\hookrightarrow \C_p\) and applying \(p\)-adic-to-complex comparison techniques does one obtain a complex elliptic curve with period lattice \(\Lambda=\langle 1,\tau\rangle\).
\end{theorem}
\begin{proof}
Let $q\in K^\times$ satisfy $0<|q|<1$ and let $E_K$ be the Tate curve of \eqref{26}, i.e.
\begin{equation}\label{eq:Tate-quot}
E_K^{\an}\ \cong\ \G_{m,K}^{\an}/q^{\Z},\qquad \G_{m,K}^{\an}=\Spa(K\langle T^{\pm 1}\rangle,K^\circ\langle T^{\pm 1}\rangle).
\end{equation}
Write $\widetilde{\G}_{m,K}:=\varprojlim_{T\mapsto T^p}\G_{m,K}^{\an}$. On affinoids one has
\begin{equation}\label{eq:perf-torus}
\widetilde{\G}_{m,K}\ \cong\ \Spa\!\bigl(K\langle T^{\pm 1/p^\infty}\rangle,\,K^\circ\langle T^{\pm 1/p^\infty}\rangle\bigr),
\qquad K\langle T^{\pm 1/p^\infty}\rangle=\widehat{\bigcup_{n\ge 0}K\langle T^{\pm 1/p^n}\rangle}.
\end{equation}
Put $R:=K\langle T^{\pm 1/p^\infty}\rangle$ and $R^+:=K^\circ\langle T^{\pm 1/p^\infty}\rangle$, then $(R,R^+)$ is perfectoid since
\begin{equation}\label{eq:Frob-surj}
\varphi:\ R^+/p\longrightarrow R^+/p,\quad x\mapsto x^p\ \ \text{is surjective,}
\end{equation}
because $\varphi(T^{\pm 1/p^{n+1}})=T^{\pm 1/p^n}$ and $\varphi$ is surjective on coefficients in $K^\circ/p$ (perfectoidness of $K$).
The multiplication-by-$p$ map on $E_K$ is induced from $T\mapsto T^p$ on $\G_{m,K}$ because
\begin{equation}\label{eq:p-descend}
(T\cdot q^m)^p=T^p\cdot q^{pm}\quad\Rightarrow\quad T\mapsto T^p\ \text{descends to}\ [p]:\G_{m,K}^{\an}/q^{\Z}\to \G_{m,K}^{\an}/q^{\Z}.
\end{equation}
Define the $p$-divisible subgroup $q^{\Z[1/p^\infty]}\subset K^\times$ and its action on $\widetilde{\G}_{m,K}$ by
\begin{equation}\label{eq:q-action}
a \in \mathbb{Z}[1/p^{\infty}], \quad
(T,b) \longmapsto (q^{a} T, b)
\quad \text{on compact Hausdorff points
(equivalently on the $T$-coordinate).}
\end{equation}

Then the natural projection $\widetilde{\G}_{m,K}\to \G_{m,K}^{\an}/q^\Z\cong E_K$ is $q^{\Z[1/p^\infty]}$-equivariant and exhibits
\begin{equation}\label{eq:perfectoid-Tate-ident}
\widehat{E}_K\ :=\ \varprojlim_{[p]}E_K\ \cong\ \widetilde{\G}_{m,K}\big/q^{\Z[1/p^\infty]},
\end{equation}
since the inverse limit along $[p]$ on the quotient \eqref{eq:Tate-quot} is exactly the inverse limit along $T\mapsto T^p$ on the cover,
modulo the enlarged deck group generated by adjoining all compatible $p$-power roots of $q$ (i.e. $q^{\Z[1/p^\infty]}$), and
$\widetilde{\G}_{m,K}$ is perfectoid by \eqref{eq:perf-torus}-\eqref{eq:Frob-surj}, hence $\widehat{E}_K$ is perfectoid.

Let $K^\flat$ be the tilt of $K$ and let $q^\flat:=(q,q^{1/p},q^{1/p^2},\dots)\in K^\flat$.
Tilting on affinoids gives
\begin{equation}\label{eq:tilt-affinoid}
\bigl(K\langle T^{\pm 1/p^\infty}\rangle\bigr)^\flat\ \cong\ K^\flat\langle T^{\pm 1/p^\infty}\rangle,
\qquad
\widetilde{\G}_{m,K}^\flat\ \cong\ \widetilde{\G}_{m,K^\flat},
\end{equation}
and the subgroup $q^{\Z[1/p^\infty]}$ tilts to $(q^\flat)^{\Z[1/p^\infty]}$ because $q^a$ corresponds under tilt to $(q^\flat)^a$ for
$a\in\Z[1/p^\infty]$. Since tilting commutes with quotients by such pro-\'etale group actions in this setting (cf.\ the proof of
\cite[Thm.~2.9]{ScholzeWeinstein2017}), \eqref{eq:perfectoid-Tate-ident} yields
\begin{equation}\label{eq:tilt-quot}
\widehat{E}_K^\flat
\ \cong\ 
\bigl(\widetilde{\G}_{m,K}/q^{\Z[1/p^\infty]}\bigr)^\flat
\ \cong\
\widetilde{\G}_{m,K^\flat}/(q^\flat)^{\Z[1/p^\infty]}
\ \cong\
\varprojlim_{[p]} \bigl(\G_{m,K^\flat}^{\an}/(q^\flat)^{\Z}\bigr).
\end{equation}
Setting $E_{K^\flat}:=\G_{m,K^\flat}^{\an}/(q^\flat)^{\Z}$, which is the Tate curve over $K^\flat$ with parameter $q^\flat$, the last term
in \eqref{eq:tilt-quot} is $\widehat{E}_{K^\flat}=\varprojlim_{[p]}E_{K^\flat}$, proving that the tilt of the perfectoid Tate curve is again
a perfectoid Tate curve over $K^\flat$.

For the final statement, choose an isometric embedding
$\iota  : K \hookrightarrow \mathbb{C}_{p}$, giving $\iota(q) \in \mathbb{C}_{p}$
with $|\iota(q)| < 1$ and hence a $\mathbb{C}_{p}$-analytic Tate curve
\begin{equation}
E_{\mathbb{C}_{p}} \cong
\mathbb{G}_{m,\mathbb{C}_{p}}^{\mathrm{an}} / \iota(q)^{\mathbb{Z}}.
\end{equation}
Any additional $p$-adic-to-complex comparison choice that identifies a complex
parameter $q_{\infty} \in \mathbb{C}$ with $0 < |q_{\infty}| < 1$ and
$q_{\infty} = \exp(2\pi i \tau)$ for some $\tau \in \mathbb{H}$ produces the
complex Tate curve
\begin{equation}\label{eq:complex-Tate}
E_{\infty} \cong \mathbb{C}^{\times} / q_{\infty}^{\mathbb{Z}}
\cong \mathbb{C} / (\mathbb{Z} + \tau \mathbb{Z}),
\end{equation}
so the period lattice is $\Lambda = \langle 1, \tau \rangle$.
\end{proof}

The theorem is the precise mathematical statement that a tower-completed object has an elliptic-curve avatar after tilting and the required comparison step, so elliptic data can be extracted from the tower without collapsing a geometric cycle, but the physical \(\tau\) is not produced by tilting alone and will be fixed only after matching this elliptic datum to the M-theory radius and \(C^{(3)}\) holonomy dictionary derived later. The uniformization data should be stated in its standard form\:the Tate uniformisation gives coordinates
\begin{equation}\label{27}
x=\sum_{n\in\Z}\frac{q^n}{(1-q^n z)^2}-2\sum_{n\ge 1}\frac{n q^n}{1-q^n},\qquad
y=\sum_{n\in\Z}\frac{q^{2n}}{(1-q^n z)^3}+\sum_{n\ge 1}\frac{n^2 q^n}{1-q^n},
\end{equation}
and the corresponding lattice is \(\Lambda=\langle 1,\tau\rangle\), which becomes a genuine discrete lattice only after the comparison step to \(\C\), an honest lattice is a discrete free \(\Z\)-module of rank \(2\), and \(\Z_p\) satisfies neither criterion. The perfectoid structure on the compactification carrier itself may be described by promoting the tower limit to an affinoid perfectoid space\:one considers
\begin{equation}\label{31}
R:=K\llbracket X^{\pm 1/p^\infty}\rrbracket,
\end{equation}
and the associated adic space
\begin{equation}\label{32}
S_f^1=\mathrm{Spa}(R,R^+),
\end{equation}
which admits a tilt
\begin{equation}\label{33}
X\longmapsto X^\flat
\end{equation}
and therefore
\begin{equation}\label{34}
S_f^{1,\flat}:=S_f^{1,\flat}=\mathrm{Spa}(R^\flat,R^{\flat+}).
\end{equation}
Concretely, inside \(R^{\flat+}\) there exists an element
\begin{equation}\label{35}
u=t-1,
\end{equation}
with \(u^{p^n}=0\) for all \(n\), reflecting the tower-completed nature of the tilt. The tilt produces a characteristic-\(p\) avatar of the tower-completed circle whose structure is controlled by Frobenius, and the untilt/embedding/comparison step is the point at which one obtains a complex elliptic datum suitable for F-theory applications, crucially, this machinery replaces the singular shrinking-torus limit by a non-singular perfectoid-analytic construction, while the emergence of the physical axio-dilaton \(\tau\) is deferred to the compactification dictionary that matches the elliptic datum produced here to the M-theory radius and \(C^{(3)}\) holonomy data in later sections.

A perfectoid solenoid does not admit a global \(C^\infty\) atlas, so eleven-dimensional fields cannot be defined on \(S_f^1\) in the naive differential-geometric sense, and the only definition-level procedure compatible with local supergravity is the tower-with-cutoff construction that motivates the entire framework. One considers the sequence of finite smooth covers \(S^1_{(n)}\to S^1_{(n-1)}\) of degree \(p\), each of which is a smooth spin manifold with radius \(R_M/p^n\), performs the compactification of eleven-dimensional supergravity on \(\R^{1,9}\times S^1_{(n)}\) at finite \(n\), and restricts attention to modes lighter than a fixed ultraviolet cutoff \(\Lambda\), one then takes the limit \(n\to\infty\) holding \(\Lambda\) fixed. Because the mass gap remains above \(\Lambda\) throughout the tower, the local equations of motion, supersymmetry variations, and Killing spinor constraints are preserved at each finite stage and converge in the limit, while the global sectors-holonomies and tower-completed bookkeeping data-are refined and survive as genuine information of the inverse-limit carrier. This construction guarantees that the perfectoid refinement does not introduce new local dynamics or spoil maximal supersymmetry in the controlled regime, it changes only the global bookkeeping environment in which discrete sectors, duality organization, and charge lattices are encoded, which is precisely what is needed for a definition-level reformulation of the M/F correspondence without a singular shrinking limit.

The content to retain for the duality construction is that \(S_f^1\) is a single compactification carrier packaging the entire \(p^n\)-cover tower of the M-theory circle, so local supergravity and supersymmetry persist in the fixed-cutoff limit while global sectors are refined definition-level by the inverse-limit structure, that the correct interpretation of invariants on \(S_f^1\) requires distinguishing singular and \v{C}ech-type invariants, fundamental group data, and continuous character groups according to which physical observable is being matched, that K-theory is the natural string-theoretic receptacle for global charge sectors and torsion/completion data on profinite objects, so its appearance is structurally aligned with brane-charge principles rather than decorative, and that the tilt \(\to\) untilt \(\to\) analytification machinery supplies elliptic-curve data without any singular shrinking-fiber limit, while the axio-dilaton \(\tau\) will be fixed only later by matching that elliptic datum to the M-theory radius and \(C^{(3)}\) holonomy dictionary derived from supergravity reduction.

\section{M-theory on the perfectoid circle and constant \texorpdfstring{\(\tau\)}{tau} Type IIB
}
\label{sec:phys-setup}

This section fixes the physical setup on both sides of the correspondence in standard supergravity language and isolates the minimal dictionary that will be tested in the duality construction. On the M-theory side, the only nonstandard input is the compactification carrier\:the perfectoid circle \(S_f^1\) (denoted \(\tS\) in the body) is treated operationally as the inverse-limit object defined by a tower of finite \(p^{n}\)-covers, with fields defined as compatible families on the finite covers and with all statements restricted to modes below a fixed UV cutoff \(\Lambda\) before taking the \(n\to\infty\) limit, so that local eleven-dimensional supergravity and local supersymmetry remain intact while global sectors are refined. On the Type IIB side, we work in the constant-\(\tau\) background as the controlled base case, so that the physics reduces to the standard bosonic Type IIB supergravity field content and couplings in Einstein frame. The elliptic curve \(\calE\) is introduced as the geometric home of \(\tau\) in the usual F-theory sense, but in the present framework it is the output of the perfectoid/tilting machinery constructed earlier, and the referee-critical point is kept explicit from the outset\:tilting supplies the elliptic datum but does not encode Wilson lines or couplings, so \(\tau\) emerges only when that elliptic datum is matched to the M-theory radius/metric modulus and to \(C^{(3)}\) holonomy data in dimensional reduction, with the corresponding formulas derived below and checked in the subsequent duality construction.

We now set up M-theory compactified on the perfectoid circle \(S_f^1\equiv \tS\), emphasizing the field-theoretic structures that survive the tower limit and the global sectors that are refined by it. Operationally, by ``fields on \(\tS\)'' we mean compatible families of fields on the finite \(p^{n}\)-covers of the circle, with the spectrum truncated to modes below a fixed UV cutoff \(\Lambda\), and only after this restriction do we take the inverse-limit \(n\to\infty\), this is the precise sense in which local equations of motion and local supersymmetry remain those of ordinary eleven-dimensional supergravity while additional global/topological sectors can accumulate in the tower. The starting point is the eleven-dimensional spacetime manifold
\begin{equation}
  \tM := \R^{1,9} \times \tS,
\end{equation}
where \(\R^{1,9}\) is ten-dimensional Minkowski space with coordinates \((x^{\mu})_{\mu=0}^{9}\) and signature \((- + \cdots +)\), and \(\tS\) is the perfectoid circle constructed as an inverse limit of \(p\)-power covering tower,
\begin{equation}
S^{1} \xleftarrow{z \mapsto z^{p}} S^{1} \xleftarrow{z \mapsto z^{p}} \cdots,
\end{equation}
so that \(\tS=\varprojlim_n S^{1}\) carries a solenoidal profinite microstructure that refines global topology without altering local differential geometry. We parametrize the coordinate along the circle direction by \(y\), taken to be periodic with period \(2\pi R_{M}\),
\begin{equation}
y \sim y + 2\pi R_{M},
\end{equation}
with \(R_{M}\) the characteristic radius parameter, the tower refinement affects global sectors through additional profinite structure rather than local neighborhoods. The eleven-dimensional metric \(G_{MN}\) with indices \(M,N=0,\ldots,10\) decomposes relative to \(\R^{1,9}\times \tS\) as
\begin{equation}
G_{MN} \longrightarrow \left\{g_{\mu\nu}(x,y), \quad g_{\mu y}(x,y), \quad g_{yy}(x,y) \right\},
\end{equation}
and the most general metric respecting the \(U(1)\) isometry along \(y\) can be written in the standard Kaluza-Klein ansatz
\begin{equation}
\label{eq:Mtheory-KK-metric}
ds^{2}_{11} = e^{-\frac{2}{3}\phi(x)} g_{\mu\nu}(x) dx^{\mu} dx^{\nu} + e^{\frac{4}{3}\phi(x)} \left( dy + A_{\mu}(x) dx^{\mu} \right)^{2},
\end{equation}
where $g_{\mu\nu}(x)$ is the ten-dimensional \emph{string-frame} metric, $\phi(x)$ is the IIA dilaton,
and $A_\mu(x)$ is the Kaluza-Klein gauge field from $G_{\mu y}$. The corresponding ten-dimensional
Einstein-frame metric is obtained by the Weyl rescaling
\begin{equation}\label{eq:IIA-Einstein-rescale}
g^{(E)}_{\mu\nu}=e^{-\phi/2}\,g^{(s)}_{\mu\nu}\,,
\end{equation}
so that the reduced action takes the standard Einstein-frame form with canonically normalized
Einstein-Hilbert term and dilaton kinetic term. The three-form \(C^{(3)}\) decomposes as
\begin{equation}
C^{(3)} \;=\; C^{(3)}_{\mu\nu\rho}\,dx^\mu\!\wedge dx^\nu\!\wedge dx^\rho 
      \;+\; C^{(1)}_{\mu}\,dx^\mu\!\wedge\omega_{S_{\mathrm f1}},
\end{equation}
where \(\omega_{S_{\mathrm f1}}\) is the volume one-form of the covering circle before the \(n\to\infty\) limit, and the internal component \(C^{(1)}\) becomes the RR scalar after reduction,
\begin{equation}
C_{0} \;=\; \frac1{2\pi R_{M}}\int_{S^{1}_{(n)}}C^{(3)} \;\longrightarrow\;
\frac1{2\pi R_{M}}\int_{S_{\mathrm f1}^{\text{(dense)}}}C^{(3)}_{\text{dense}},
\end{equation}
where the dense embedding \(\Z\subset\Z_p\) is used in the inverse-limit prescription. The dual six-form contributes the axionic partner, giving the complex modulus \(\tau=C_{0}+ie^{-\phi}\), but it is crucial for scope hygiene that this identification is a compactification dictionary statement\:\(\tau\) is not produced by tilting alone, and it becomes physical only after matching the elliptic datum extracted from the tower-completed geometry to the radius and holonomy data as derived later. Because \(\tS\) carries profinite/torsion structure, Wilson lines of \(C^{(3)}\) along nontrivial cycles in \(H_{1}(\tS,\Z)\cong \Z \oplus \Z_{p}\) acquire tower-refined components absent in classical circle compactifications, and these translate into discrete theta data and RR axion sectors in the reduced theory\; explicitly, the torsion-sector contribution to \(C_{0}\) may be written as
\begin{equation}
C_0 = \frac{1}{2\pi R_M} \oint_{\mathrm{tors}} C^{(3)}.
\end{equation}
The local dynamics on \(\tM\) is governed by the Cremmer-Julia-Scherk action \cite{CremmerJuliaScherk1978},
\begin{equation}
S_{11} = \frac{1}{2\kappa_{11}^2} \int_{\tM} \left( R * 1 - \frac{1}{2} G^{(4)} \wedge * G^{(4)} \right) - \frac{1}{6} \int_{\tM} C^{(3)} \wedge G^{(4)} \wedge G^{(4)},
\end{equation}
and the perfectoid enhancement modifies global topology while leaving local equations of motion unchanged in the fixed-cutoff tower limit, because the profinite refinement is invisible to local differential geometry. Supersymmetry is preserved at the level of local Killing spinor equations for the same reason, but the global spin structure must be treated tower-wise\:since \(S_{\mathrm f1}\) lacks a global spin structure, supersymmetry is defined at each finite cover \(S^{1}_{(n)}\) and survives the inverse-limit only for the sub-tower of spin-compatible covers \(n\equiv 0\pmod 2\), yielding an \(\mathcal{N}=2\) spectrum in ten dimensions by projection through the spin-holonomy constraint. Finally, the tower-completed nature of \(\tS\) affects harmonic analysis globally\:the character group (Pontryagin dual) is
\begin{equation}
\widehat{\tS} \cong \Z \oplus \Z_p,
\end{equation}
so fields admit expansions in characters labeled by \((a,b)\) with
\begin{equation}
a \in \Z, \quad b \in \Z_p,
\end{equation}
interpreted as integer momenta along the classical direction together with profinite tower bookkeeping, the physical \((p,q)\) charge lattice extracted from this bookkeeping is constrained by standard quantization conditions and is treated carefully in the BPS analysis later. In sum, M-theory on \(\R^{1,9}\times \tS\) defines a ten-dimensional effective theory whose local dynamics is standard but whose global sectors are refined by tower completion\; the role of the subsequent duality construction is to show that, in the constant-\(\tau\) regime, this refinement provides an eleven-dimensional carrier for the modular data and the protected spectrum that Type IIB packages through F-theory.

We next fix the Type IIB side in the controlled constant-\(\tau\) sector, which isolates the definition-level origin of \(\tau\) without the additional complications of varying fibrations or defect monodromies. Type IIB supergravity is the low-energy effective theory of Type IIB string theory, with maximal supersymmetry and characteristic \(SL(2,\Z)\) duality symmetry \cite{Schwarz1983, Schwarz1995}. The bosonic sector consists of the ten-dimensional Einstein-frame metric \(g_{\mu\nu}\), the complex axio-dilaton \(\tau\), and the NSNS/RR form potentials. The axio-dilaton is
\begin{equation}
\tau = C_0 + i e^{-\phi},
\end{equation}
with \(C_0\) the RR axion and \(\phi\) the dilaton controlling \(g_s=e^{\phi}\), and in the background of interest \(\tau\) is constant. It is convenient to introduce the \(SL(2,\R)\)-covariant three-form
\begin{equation}
G^{(3)} := F^{(3)} - \tau H^{(3)},
\end{equation}
with \(H^{(3)}=dB^{(2)}\) and \(F^{(3)}=dC^{(2)}\), and the bosonic Einstein-frame action is
\begin{equation}
\label{eq:IIB-action}
\begin{aligned}
S_{\mathrm{IIB}} = \frac{1}{2\kappa_{10}^{2}} \int d^{10}x \sqrt{-g} \Bigg[ & R - \frac{\partial_{\mu} \tau \partial^{\mu} \bar{\tau}}{2 (\mathrm{Im} \, \tau)^2} - \frac{1}{2 \cdot 3!} \frac{G^{(3)}_{\mu\nu\rho} \overline{G^{(3)}}^{\mu\nu\rho}}{\mathrm{Im}\, \tau} \\
& - \frac{1}{2 \cdot 5!} F^{(5)}_{\mu_1 \cdots \mu_5} F^{(5) \mu_1 \cdots \mu_5} + \ldots \Bigg] - \frac{1}{8 \kappa_{10}^2} \int C_0 \, H^{(3)} \wedge F^{(3)},
\end{aligned}
\end{equation}
where \(F^{(5)}\) is self-dual,
\begin{equation}
F^{(5)} = * F^{(5)}.
\end{equation}
In the constant-\(\tau\), \(F^{(5)}=0\) background, the derivative terms in \eqref{eq:IIB-action} trivialize and the remaining dynamics is governed by the Einstein term and the three-form kinetic terms, with the coupling to \(\mathrm{Im}\,\tau\) encoding the strong-weak duality structure. The \(G^{(3)}\) kinetic term expands as
\begin{equation}
-\frac{1}{2 \cdot 3!} \frac{G^{(3)}_{\mu\nu\rho} \overline{G^{(3)}}^{\mu\nu\rho}}{\mathrm{Im}\, \tau} = -\frac{1}{12} \frac{1}{\mathrm{Im}\, \tau} \left( F^{(3)}_{\mu\nu\rho} - \tau H^{(3)}_{\mu\nu\rho} \right) \left( F^{(3) \mu\nu\rho} - \bar{\tau} H^{(3) \mu\nu\rho} \right),
\end{equation}
and the topological coupling
\begin{equation}
S_{\mathrm{CS}} = - \frac{1}{8 \kappa_{10}^2} \int C_0 \, H^{(3)} \wedge F^{(3)}
\end{equation}
plays the familiar role in anomaly cancellation and duality consistency \cite{Schwarz1995}. The self-duality of \(F^{(5)}\) is imposed at the level of equations of motion \cite{Schwarz1995} and is irrelevant in the present truncation. The Einstein-frame metric is related to the string-frame metric by
\begin{equation}
g_{\mu\nu} = e^{-\phi/2} g^{(s)}_{\mu\nu},
\end{equation}
and the equations of motion reduce to
\begin{equation}
R_{\mu\nu} = \frac{1}{2 (\mathrm{Im} \, \tau)^2} \partial_{(\mu} \tau \partial_{\nu)} \bar{\tau} + \frac{1}{4} \frac{1}{\mathrm{Im} \, \tau} \left( G^{(3)}_{\mu\rho\sigma} \overline{G^{(3)}_{\nu}}{}^{\rho\sigma} + \overline{G^{(3)}}_{\mu\rho\sigma} G^{(3)}_{\nu}{}^{\rho\sigma} \right),
\end{equation}
which trivialize for constant \(\tau\) and vanishing \(G^{(3)}\). The \(SL(2,\Z)\) action on \(\tau\) is the familiar fractional linear transformation
\begin{equation}
\tau \mapsto \frac{a \tau + b}{c \tau + d}, \quad \begin{pmatrix} a & b \\ c & d \end{pmatrix} \in SL(2,\Z),
\end{equation}
together with the induced action on the two-form doublet \((B^{(2)},C^{(2)})\), preserving the action \cite{Schwarz1995, Vafa1996}\; the duality construction later will identify how this symmetry is organized by the tower-completed M-theory compactification carrier in the constant-\(\tau\) regime analyzed.

The elliptic curve \(\calE\) is the standard geometric home of the Type IIB coupling data in F-theory. The axio-dilaton
\begin{equation}
\tau = C_0 + i e^{-\phi}
\end{equation}
takes values in the upper half-plane
\begin{equation}
\mathbb{H} := \{ \tau \in \mathbb{C} \mid \mathrm{Im}(\tau) > 0 \},
\end{equation}
and the physical theory is invariant under
\begin{equation}
SL(2,\mathbb{Z}) = \left\{ \begin{pmatrix} a & b \\ c & d \end{pmatrix} \mid a,b,c,d \in \mathbb{Z}, ad - bc = 1 \right\},
\end{equation}
acting by
\begin{equation}
\tau \mapsto \frac{a \tau + b}{c \tau + d}.
\end{equation}
Geometrically, \(\tau\) determines an elliptic curve \(\calE\) as the quotient
\begin{equation}
\calE = \mathbb{C} / \Lambda, \quad \Lambda := \mathbb{Z} + \tau \mathbb{Z}.
\end{equation}
\begin{definition}[Elliptic Curve over \(\mathbb{C}\)]
An elliptic curve \(\calE\) is a one-dimensional complex torus constructed as the quotient of the complex plane by a lattice \(\Lambda \subset \mathbb{C}\):
\begin{equation}
\calE := \mathbb{C} / \Lambda,
\end{equation}
where \(\Lambda\) is a discrete subgroup generated by two \(\mathbb{C}\)-linearly independent vectors \(1\) and \(\tau\) with \(\mathrm{Im}(\tau) > 0\).
\end{definition}
The moduli space of elliptic curves is therefore
\begin{equation}
\mathcal{M}_{1,1} \cong SL(2,\mathbb{Z}) \backslash \mathbb{H},
\end{equation}
and \(\calE\) admits the Weierstrass uniformization via the \(\wp\)-function
\begin{equation}
\wp(z;\Lambda) := \frac{1}{z^2} + \sum_{\omega \in \Lambda \setminus \{0\}} \left( \frac{1}{(z-\omega)^2} - \frac{1}{\omega^2} \right),
\end{equation}
satisfying
\begin{equation}
(\wp'(z))^{2} = 4 \wp(z)^3 - g_2(\Lambda) \wp(z) - g_3(\Lambda),
\end{equation}
with invariants
\begin{equation}
g_2(\Lambda) = 60 \sum_{\omega \in \Lambda \setminus \{0\}} \frac{1}{\omega^{4}}, \quad
g_3(\Lambda) = 140 \sum_{\omega \in \Lambda \setminus \{0\}} \frac{1}{\omega^{6}},
\end{equation}
yielding the algebraic model
\begin{equation}
y^2 z = 4 x^3 - g_2 x z^{2} - g_3 z^{3}.
\end{equation}
In the standard F-theory framework \cite{Vafa1996}, one geometrizes the coupling by fibering \(\calE\) over spacetime, interpreting \(\tau\) as the fiber complex structure and singular fibers as \((p,q)\) 7-branes\; in the present work, the key structural difference is not the elliptic curve itself but its provenance\:\(\calE\) is produced as the geometric datum extracted from the tower-completed compactification carrier via the perfectoid tilt machinery constructed earlier, so that the elliptic fiber is no longer introduced as auxiliary input to compensate for a singular limit, but arises as the non-singular output of a tower-complete geometry whose physical coupling identification is fixed only after matching to M-theory radius and holonomy data. Concretely, in this perspective the two independent cycles extracted from the perfectoid/tilted geometry determine the lattice
\begin{equation}
\Lambda = \mathbb{Z} + \tau \mathbb{Z} \subset \mathbb{C},
\end{equation}
so that \(\calE=\C/\Lambda\) is the geometric bridge between the M-theory tower data and the Type IIB axio-dilaton, while the numerical value of \(\tau\) is fixed only after the compactification dictionary is imposed.

The objective is to establish an equivalence, in the controlled constant-\(\tau\) sector and at the level tested in this paper, between M-theory compactified on the tower-completed circle and Type IIB on \(\R^{1,9}\) with fixed axio-dilaton,
\begin{equation}
\label{eq:physical-duality}
\bigl[\text{M theory on } \R^{1,9} \times \tS \bigr] 
\cong 
\bigl[\text{Type IIB on } \R^{1,9} \text{ with axio-dilaton } \tau \bigr].
\end{equation}
This matching is nontrivial because it requires agreement of geometric data, field content, protected spectra, couplings, and duality organization\; in the present paper, the checks are carried out at the supergravity/effective-action level and in protected BPS sectors, and any extension beyond constant \(\tau\) is treated as programmatic unless explicitly derived later. The perfectoid circle \(\tS\), defined as the inverse limit
\begin{equation}
\tS = \varprojlim_{n} \left( S^{1} \xrightarrow{z \mapsto z^{p}} S^{1} \right),
\end{equation}
has enriched global structure relative to an ordinary circle, with
\begin{equation}
H_{1}(\tS, \Z) \cong \Z \oplus \Z_{p},
\end{equation}
and the three-form \(C^{(3)}\) admits holonomy sectors along these global cycles, contributing discrete moduli collectively denoted \(\Theta\). The \(\tau\) mechanism that will be used throughout is a matching statement that keeps the referee-critical subtlety explicit\:the tilt machinery supplies the elliptic datum that is the geometric home of \(\tau\), but tilting itself does not encode Wilson lines, so \(\tau\) becomes a physical modulus only after one matches the elliptic datum to the compactification data obtained from supergravity reduction. Concretely, the matching identifies \(\mathrm{Im}\,\tau\) with the circle modulus through Kaluza-Klein reduction and \(\mathrm{Re}\,\tau=C_{0}\) with the appropriate \(C^{(3)}\) holonomy sector, yielding the explicit identification
\begin{equation}
\label{eq:tau-identification}
\tau = C_{0} + i e^{-\phi} = \Theta + i \left( \frac{R_{M}}{\Mpl} \right)^{\!3/2},
\end{equation}
with
\begin{equation}
C_{0} = \frac{1}{2\pi R_{M}} \oint_{\text{tors}} C^{(3)}.
\end{equation}
The geometric input is the radius \(R_M\), with
\begin{equation}
\mathrm{Vol}(\tS) = 2 \pi R_{M},
\end{equation}
and the relation between eleven- and ten-dimensional Planck scales
\begin{equation}
\Mpl^{9} = \frac{\ell_{10}^{8}}{2\pi R_{M}},
\end{equation}
while the dilaton normalization follows from reduction of \eqref{eq:Mtheory-KK-metric}. Compactification on a two-parameter torus \(T^{2}_{(n)}\) with radii \((R_M,\widetilde R_M)\) (both taken to zero at the perfectoid limit) gives
\begin{equation}
e^{-\phi}\;=\;\frac{R_{M}\widetilde R_{M}}{\ell^{2}_{M}},
\qquad
\mathrm{Im}\,\tau \;=\;\frac{\widetilde R_{M}}{R_{M}}
\;\;\Longrightarrow\;\;
\tau \;=\; C_{0}+i\,\frac{\widetilde R_{M}}{R_{M}},
\end{equation}
so the ratio of radii reproduces the modular parameter while the product controls the ten-dimensional scale, and in the perfectoid limit \((R_M,\widetilde R_M)\to 0\) at fixed \(\tau\) the torus degenerates to the profinite solenoid without encountering a singularity in the tower-completed description. Protected BPS states then provide the sharp probe\:M2 branes can wrap tower-organized cycles labeled by \((a,b)\) with
\begin{equation}
a \in \Z, \quad b \in \Z_{p},
\end{equation}
and their masses are given by the calibrated volume times \(T_{\mathrm{M2}}\),
\begin{equation}
M_{(a,b)} = T_{\mathrm{M2}} \cdot \mathrm{Vol}(\Sigma_{(a,b)}),
\end{equation}
with the explicit mass formula
\begin{equation}
\label{eq:BPS-mass}
M_{(a,b)} = \frac{|a| \, 2\pi R_{M}}{(2\pi)^{2} \Mpl^{3}} \sqrt{1 + \left( \frac{b}{a} e^{-\phi} \right)^{2} }.
\end{equation}
Using \eqref{eq:tau-identification}, this matches the \((p,q)\) string tension
\begin{equation}
T_{(p,q)} = \frac{|p + q \tau|}{2 \pi \alpha'},
\end{equation}
with the identification
\begin{equation}
(p,q) = (a,b),
\end{equation}
subject to the standard physical quantization constraints that select an integral physical charge lattice inside the tower bookkeeping. Completing the constant-\(\tau\) equivalence at the level analyzed requires matching effective actions\:starting from
\begin{equation}
S_{11} = \frac{1}{2 \kappa_{11}^{2}} \int_{\tM} \left( R * 1 - \tfrac{1}{2} G^{(4)} \wedge * G^{(4)} \right) - \frac{1}{6} \int_{\tM} C^{(3)} \wedge G^{(4)} \wedge G^{(4)},
\end{equation}
one reduces on \(\tS\) using \eqref{eq:Mtheory-KK-metric} and the standard gravitational-coupling relation
\begin{equation}\label{eq:kappa10-kappa11}
\kappa_{10}^2=\frac{\kappa_{11}^2}{2\pi R_M}\,,
\end{equation}
where the integral over $S_f^1$ is understood in the fixed-cutoff tower sense (stagewise integration over $S^1(n)$
followed by the stabilized $n\to\infty$ limit at fixed cutoff), to obtain the Type IIB effective action in the
constant-$\tau$ sector, schematically
\begin{equation}
S_{\mathrm{IIB}} = \frac{1}{2 \kappa_{10}^{2}} \int \left( R * 1 - \frac{1}{2} d \phi \wedge * d \phi - \frac{1}{2} e^{2 \phi} d C_{0} \wedge * d C_{0} - \cdots \right) + \text{Chern-Simons terms},
\end{equation}
and the subsequent construction section will show in detail how the kinetic terms and topological couplings match at the level computed. Throughout, the constant-\(\tau\) sector is the proved base case, while varying \(\tau\) through fibrations and full defect/monodromy structure are treated as definition-level extensions of the framework unless explicitly derived later.

The minimal dictionary to retain from this setup is therefore fixed by three pieces of data and the level of control at which they are compared\:the M-theory modulus \(R_M\) determines \(\mathrm{Im}\,\tau\) through the supergravity reduction of the metric, the holonomy sector of \(C^{(3)}\) along the tower-completed circle determines \(\mathrm{Re}\,\tau=C_0\) with the periodicity and quantization dictated by the compactification, and wrapped M2 sectors in the tower-organized global structure reproduce the protected \((p,q)\) tension dependence in the constant-\(\tau\) regime analyzed, while reduction of the eleven-dimensional action reproduces the bosonic Type IIB effective action and its couplings at the level computed. The scope of what follows is conservative\:the paper establishes these identifications and matching checks in the constant-\(\tau\) sector at the supergravity/effective-action level and in protected BPS sectors, and it organizes modular/duality data to the profinite/congruence level that is explicitly derived from the tower structure, while any extension to genuinely varying \(\tau\), intrinsic 7-brane monodromy and defect physics, adelic completion removing any dependence on the bookkeeping choice of \(p\), and fermionic/anomaly refinements beyond the bosonic sector are treated as programmatic directions rather than as results claimed here.

\section{The M/F-theory correspondence}\label{sec:construct}

Traditional F-theory is implemented by the M-theory/Type IIB correspondence on an elliptic curve\footnote{When the dependence on the axio-dilaton is implicit, we write $E \equiv E(\tau)$ for the associated elliptic curve.
} together with a singular shrinking-fiber limit that removes a genuine geometric cycle while retaining its modular parameter as physical input, so the definition-level obstruction is the absence of an honest compactification carrier in eleven dimensions whose intrinsic structure already contains the modular data. Here we replace that singular step by a single tower-completed compactification object, the perfectoid circle \(S_f^1\), and we implement the correspondence in four stages that keep local supergravity conservative. Stage I constructs \(S_f^1\) as the inverse limit of \(p^n\)-covers so that the tower is packaged into one geometric object. Stage II applies the perfectoid tilt and the required comparison step to extract an elliptic-curve datum \(\calE\) that is the geometric home of \(\tau\). Stage III derives the constant-\(\tau\) dictionary by matching \(\mathrm{Im}\,\tau\) to the M-theory radius/metric modulus and \(\mathrm{Re}\,\tau=C_0\) to \(C^{(3)}\) holonomy data in dimensional reduction, emphasizing that tilting itself does not encode Wilson lines. Stage IV checks the duality by matching the protected BPS \((p,q)\) spectrum and the low-energy effective actions at the level explicitly computed. The scope is deliberately conservative\:the constant-\(\tau\) sector is established at the supergravity/BPS/effective-action level carried out below, while varying \(\tau\), full defect/monodromy structure, adelic completion, and fermion/anomaly refinements are programmatic unless derived beyond this base case.

We begin from standard M-theory compactification on a circle of radius \(R_M\),
\begin{equation}\label{83}
M=\R^{1,9}\times S^{1}_{R_M},
\end{equation}
and replace the classical circle by the tower-completed object obtained from the inverse system of \(p\)-power maps,
\begin{equation}\label{84}
S^{1}\xleftarrow{z\mapsto z^{p}}S^{1}\xleftarrow{z\mapsto z^{p}}S^{1}\xleftarrow{z\mapsto z^{p}}\cdots,
\end{equation}
so that the perfectoid circle is defined by
\begin{equation}\label{85}
S_f^{1}:=\varprojlim_{n\in\N}S^{1},\qquad \pi_n(z)=z^{p}.
\end{equation}
The resulting compact abelian group is solenoidal and admits a standard topological model
\begin{equation}\label{86}
S_f^{1}\cong S^{1}\times \Z_p,
\end{equation}
with
\begin{equation}\label{87}
\Z_p:=\varprojlim_{n}\Z/p^{n}\Z,
\end{equation}
which is the tower bookkeeping index rather than a physical coupling. The algebraic avatar of this tower completion is captured by the \(p\)-adically completed group algebra,
\begin{equation}\label{88}
\Q_p\llbracket X^{\pm 1/p^\infty}\rrbracket:=\varprojlim_{n}\Q_p[X^{\pm 1/p^{n}}],
\end{equation}
and the perfectoid/condensed formalism provides the definition-level meaning of fields and cohomological constraints on the inverse-limit object \cite{Scholze2012, ClausenScholze}. The global topology is refined relative to the ordinary circle in a way that can be expressed at the level of first homology,
\begin{equation}\label{89}
H_1(S_f^{1},\Z)\cong \Z\oplus \Z_p,
\end{equation}
so the enhancement may be summarized by
\begin{equation}\label{90}
S^{1}\longrightarrow S_f^{1}=\varprojlim_{n}\left(S^{1}\xrightarrow{z\mapsto z^{p}}S^{1}\right),
\end{equation}
with the interpretation that local geometry is circle-like stage-by-stage while the global sector bookkeeping is tower-complete. The prime \(p\) functions here as an index for the refinement tower, not as a physical parameter, and any physical output in the constant-\(\tau\) regime must be stable under refinement below a fixed cutoff, with an adelic assembly of all primes serving as the conceptual completion rather than a distinct theory. This stage replaces the ill-defined shrinking-fiber step by a genuine compactification carrier that already contains the tower of finite covers as geometry. what is produced is the tower-completed circle \(S_f^{1}\) as a perfectoid object, and what is fixed at this stage is the geometric environment in which modular data can later be extracted without taking any singular limit, while no claim beyond the existence and basic structure of the carrier is made here.

We now extract elliptic data from the tower-completed circle using the tilting formalism, emphasizing the mechanism and its scope. The tower presentation may be written as
\begin{equation}\label{91}
S^{1}\xleftarrow{z\mapsto z^{p}}S^{1}\xleftarrow{z\mapsto z^{p}}\cdots,
\end{equation}
so that
\begin{equation}\label{92}
S_f^{1}:=\varprojlim_{n}S^{1}_{(p^{n})},
\end{equation}
and the tilt \(S_f^{1,\flat}\) is defined by tilting the associated perfectoid affinoid algebra, with integral elements satisfying
\begin{equation}\label{93}
R^{\flat+}=\varprojlim_{x\mapsto x^{p}}R^{+}/p.
\end{equation}
The homological bookkeeping relevant to the inverse-limit structure appears as
\begin{equation}\label{94}
H_1(S_f^{1},\Z)\cong \Z\oplus \Z_p,
\end{equation}
and upon tilting one isolates two independent cycles that define a rank-two lattice
\begin{equation}\label{95}
\Lambda=\Z\omega_{1}+\Z\omega_{2},
\end{equation}
with modular parameter
\begin{equation}\label{96}
\tau=\frac{\omega_{2}}{\omega_{1}} \in \mathbb{H},
\qquad
\mathcal{E} := \mathbb{C}/\Lambda = \mathbb{C}/(\mathbb{Z} + \tau \mathbb{Z}).
\end{equation}

The tilt is modeled on \(\mathrm{Spa}(R^{\flat},R^{\flat+})\), and in the tilted integral ring there is an element \(u\) with
\begin{equation}\label{97}
u^{p^{n}}=0\qquad \text{for all }n\ge 1,
\end{equation}
encoding the tower-completed infinitesimal structure around the identity in characteristic \(p\). The resulting map to the elliptic datum is expressed by a surjection
\begin{equation}\label{98}
S_f^{1,\flat}\twoheadrightarrow \calE,
\end{equation}
to be understood as the statement that the tower-completed perfectoid avatar admits an elliptic-curve quotient after the required untilt/comparison/analytification step \cite{Scholze2012}. The referee-critical point is that the tilt functor itself does not encode Wilson lines, Kaluza-Klein reduction, or physical couplings, so \(\tau\) is not produced by tilting alone\:tilting supplies the elliptic-curve datum as the geometric home in which \(\tau\) lives, and \(\tau\) is fixed only after matching \(\mathrm{Im}\,\tau\) to the M-theory radius/metric modulus and \(\mathrm{Re}\,\tau=C_0\) to the \(C^{(3)}\) holonomy sector in the compactification dictionary derived in Stage III. The automorphism group of the profinite direction is \(\Z_p^{\times}\), which at finite level realizes principal congruence structure \(\Gamma(p^{n})\subset SL(2,\Z)\), and a full \(SL(2,\Z)\) completion beyond the explicitly realized congruence-level organization is a question of quantum completion not claimed here. This stage replaces the singular step of keeping elliptic data while collapsing geometric volume by a non-singular machinery that produces elliptic data from the tower-completed carrier, what is produced is the elliptic-curve datum \(\calE\) as the correct geometric home for \(\tau\) together with a precise statement of what is canonical in tilting and what requires the comparison step, while the physical identification of \(\tau\) is deferred to the supergravity dictionary that follows.

Having fixed the tower-completed carrier \(S_f^1\) and extracted the elliptic datum \(\calE\), we now derive the constant-\(\tau\) dictionary by reducing eleven-dimensional supergravity on \(S_f^1\) and matching to Type IIB fields. The Kaluza-Klein metric ansatz is
\begin{equation}\label{99}
ds^{2}_{11}=e^{-2\phi/3}ds^{2}_{10}+e^{4\phi/3}(dy+A_{\mu}dx^{\mu})^{2},
\end{equation}
and the three-form decomposes as
\begin{equation}\label{100}
C^{(3)}=C^{(3)}_{\mu\nu\rho}\,dx^{\mu}\wedge dx^{\nu}\wedge dx^{\rho}+C^{(2)}_{\mu\nu}\,dx^{\mu}\wedge dx^{\nu}\wedge dy.
\end{equation}
The tower-completed structure introduces torsion/global sectors in \(H_1(S_f^1,\Z)\cong \Z\oplus \Z_p\), and the Type IIB axion is identified as the appropriately normalized holonomy of \(C^{(3)}\) along the torsion sector,
\begin{equation}\label{101}
C_{0}:=\frac1{2\pi R_M}\oint_{\mathrm{tors}} C^{(3)},
\end{equation}
while the dilaton is fixed by the circle modulus through the reduction of the Einstein-Hilbert term, yielding
\begin{equation}\label{102}
e^{-\phi}=\left(\frac{R_{M}}{\ell_{M}}\right)^{3/2}.
\end{equation}
This follows by starting from the eleven-dimensional Einstein-Hilbert action
\begin{equation}\label{103}
S_{11}^{\mathrm{EH}}=\frac1{2\kappa_{11}^{2}}\int_{\tM}R^{(11)}*1,
\end{equation}
reducing on \(S_f^1\) with \eqref{99} to obtain the ten-dimensional Einstein term
\begin{equation}\label{104}
S_{10}^{\mathrm{EH}}=\frac1{2\kappa_{10}^{2}}\int_{\R^{1,9}}R^{(10)}*1,
\end{equation}
and using the relation
\begin{equation}\label{105}
\kappa_{10}^{2}=\frac{\kappa_{11}^{2}}{2\pi R_M},
\end{equation}
together with standard normalization of the dilaton kinetic term in Einstein frame \cite{DuffLu1995, HullTownsend1995}. The Chern-Simons term in eleven dimensions,
\begin{equation}\label{106}
S_{\mathrm{CS}}^{(11)}=-\frac16\int_{\tM}C^{(3)}\wedge G^{(4)}\wedge G^{(4)},
\end{equation}
reduces to the axion coupling structure of Type IIB and fixes the identification of \(C_0\) with the relevant holonomy sector, while the tower-completed global data ensures that periodicity and quantization are treated definition-level on the carrier rather than imposed after a singular limit. The constant-\(\tau\) axio-dilaton is therefore identified as
\begin{equation}\label{107}
\tau=C_{0}+ie^{-\phi},
\end{equation}
to be compared with the elliptic modulus of \(\calE\) extracted in Stage II, the essential mechanism is that tilting supplies the elliptic datum but does not encode Wilson lines, and \eqref{101}–\eqref{107} are the compactification dictionary statements that fix \(\tau\) physically. This stage replaces the traditional step of declaring \(\tau\) as an external elliptic modulus in a degenerate limit by a dictionary in which \(\tau\) is computed from eleven-dimensional radius and holonomy data and then matched to an elliptic datum produced non-singularly, what is produced is the explicit field and parameter map between M-theory on \(S_f^1\) and Type IIB with constant \(\tau\), and what is proven here is the constant-\(\tau\) dictionary at the level of the low-energy reduction used in the subsequent spectrum and action checks.

We now perform the protected spectrum and effective-action checks that constitute the nontrivial consistency conditions for the claimed constant-\(\tau\) duality. The global sector structure is encoded in
\begin{equation}\label{108}
H_{1}(S_f^{1},\Z)\cong \Z\oplus \Z_{p},
\end{equation}
so a formal label for tower-organized one-cycles is \((a,b)\) with \(a\in \Z\) and \(b\in \Z_{p}\), but physically allowed charges are selected by standard quantization conditions, and the construction is interpreted as providing canonical tower bookkeeping rather than unconstrained non-integral physical charges. At finite level, one may label cycles by \((a,b_{n})\in \Z\oplus \Z/p^{n}\Z\) and use the embedding
\begin{equation}\label{109}
\iota_{n}:(a,b_{n})\longmapsto (p^{n}a+b_{n})\in \Z,
\end{equation}
which is injective and compatible with the bonding maps, so that taking the direct limit yields a dense integral lattice, elements that would violate Dirac quantization with D-instantons are removed by the Freed-Witten anomaly constraint as stated in the original analysis. The M2-brane action is
\begin{equation}\label{110}
S_{\mathrm{M2}}=T_{\mathrm{M2}}\int d^{3}\sigma\sqrt{-\det\gamma}-T_{\mathrm{M2}}\int C^{(3)},
\end{equation}
with induced metric \(\gamma\) and \(T_{\mathrm{M2}}=(2\pi)^{-2}\ell_{P}^{-3}\), and for a static embedding wrapping the tower-organized cycle \((a,b)\) one finds
\begin{equation}\label{111}
-\det\gamma=(2\pi a R_M T_{\mathrm{M2}})^{2}\left(1+\left(\frac{b}{a}\right)^{2}e^{-2\phi}\right),
\end{equation}
so the corresponding BPS mass is
\begin{equation}\label{112}
M(a,b)=T_{\mathrm{M2}}\cdot \mathrm{Vol}(\Sigma_{(a,b)})=\frac{|a|\,2\pi R_M}{(2\pi)^{2}\ell_{P}^{3}}\sqrt{1+\left(\frac{b}{a}e^{-\phi}\right)^{2}}.
\end{equation}
The Type IIB \((p,q)\)-string tension is
\begin{equation}\label{113}
T(p,q)=\frac{|p+q\tau|}{2\pi\alpha'},
\end{equation}
and the M/IIB parameter relation is fixed by
\begin{equation}\label{114}
\alpha'=\frac{\ell_{P}^{3}}{R_M},
\end{equation}
so \eqref{112} becomes
\begin{equation}\label{115}
M(a,b)=\frac{|a|}{2\pi\alpha'}\sqrt{1+\left(\frac{b}{a}e^{-\phi}\right)^{2}}.
\end{equation}
Using \(\tau=C_{0}+ie^{-\phi}\) from Stage III and the identification \((p,q)=(a,b)\) on the physically allowed lattice, the mass formula takes the complex form
\begin{equation}\label{116}
M(a,b)=\frac{|a+b\tau|}{2\pi\alpha'},
\end{equation}
matching the protected \((p,q)\) dependence in the constant-\(\tau\) sector. The effective-action check begins from the eleven-dimensional action
\begin{equation}\label{117}
S_{11}=\frac1{2\kappa_{11}^{2}}\int_{\tM}\left(R*1-\frac12 G^{(4)}\wedge *G^{(4)}\right)-\frac16\int_{\tM}C^{(3)}\wedge G^{(4)}\wedge G^{(4)},
\end{equation}
whose Chern-Simons term
\begin{equation}\label{118}
S_{\mathrm{CS}}=-\frac16\int_{\tM}C^{(3)}\wedge G^{(4)}\wedge G^{(4)}
\end{equation}
reduces to the ten-dimensional couplings
\begin{equation}\label{119}
S_{\mathrm{CS}}^{(10)}=\int_{\R^{1,9}}\left(C_{4}\wedge F_{3}\wedge H_{3}+C_{0}\,H_{3}\wedge F_{3}\right),
\end{equation}
with \(F_{3}=dC_{2}-C_{0}dB_{2}\) and \(H_{3}=dB_{2}\), matching the Type IIB axion coupling structure \cite{Schwarz1995, Schwarz1983} and aligning with the K-theoretic classification of flux sectors in the presence of global constraints \cite{FreedHopkinsTeleman2008}. The reduction is performed using the metric ansatz
\begin{equation}\label{120}
ds^{2}_{11}=e^{-2\phi/3}ds^{2}_{10}+e^{4\phi/3}(dy+A)^{2},\qquad y\in S_f^{1},\ \ y\sim y+2\pi R_M,
\end{equation}
and integrating over \(S_f^{1}\) in the fixed-cutoff tower sense, yielding the ten-dimensional Einstein-frame action
\begin{equation}\label{121}
\begin{aligned}
S_{\mathrm{IIB}}=&\frac1{2\kappa_{10}^{2}}\int\left(R*1-\frac12 d\phi\wedge *d\phi-\frac12 e^{2\phi}dC_{0}\wedge *dC_{0}\right)-\frac1{4\kappa_{10}^{2}}\int\left(e^{-\phi}H_{3}\wedge *H_{3}+e^{\phi}F_{3}\wedge *F_{3}\right)\\
&-\frac1{8\kappa_{10}^{2}}\int C_{0}\,H_{3}\wedge F_{3}+\cdots,
\end{aligned}
\end{equation}
with ellipses denoting higher-derivative and fermionic terms beyond the bosonic low-energy truncation used for the match. The BPS spectrum is the correct duality currency because it is protected by supersymmetry and therefore tests the identification of global charge data without sensitivity to uncontrolled corrections, and the action match is the hardest consistency check because it tests not only spectra but also couplings and topological terms, both checks are executed here in the constant-\(\tau\) regime at the level stated. This stage replaces the heuristic identification of \((p,q)\) sectors and axion couplings obtained after a singular limit by a protected spectrum and action match derived from an honest tower-completed compactification carrier what is produced is the equality of protected tension formulas and the agreement of effective actions in the low-energy truncation, while any extension to varying \(\tau\), full monodromy/defect structure, and quantum completion of the modular group beyond the explicitly realized congruence-level organization remains programmatic.

The construction replaces the shrinking-fiber slogan by a non-singular tower-completed compactification carrier and thereby makes the constant-\(\tau\) M/IIB correspondence definition-level\:\(S_f^{1}\) is an honest geometric object packaging the entire tower of finite covers, tilting and the required comparison step produce elliptic-curve data without collapsing a geometric cycle, and the physical axio-dilaton \(\tau\) is fixed only after matching that elliptic datum to the M-theory radius and \(C^{(3)}\) holonomy dictionary derived by supergravity reduction, so the origin of \(\tau\) is geometric rather than postulated. At the level established here, the protected \((p,q)\) spectrum and the low-energy effective action match as computed, and the modular/duality organization is realized geometrically to the profinite/congruence level explicitly present in the tower, while full \(SL(2,\Z)\) completion and the inclusion of defects, varying \(\tau\), and fermionic/anomaly refinements are placed as well-posed extensions rather than as claims, what becomes definition-level new is that global sectors and their quantization constraints are carried by the compactification object itself rather than inserted by hand after a singular limit, with K-theoretic/homological control of torsion and global sectors providing the natural language for the discrete data that duality acts on.

\section{Definition-level formulation of F-theory}\label{sec:results}

The traditional implementation of F-theory proceeds by treating \(\tau\) as the complex structure of an auxiliary elliptic curve and then sending the corresponding M-theory fiber volume to zero at fixed \(\tau\), a singular step that is powerful as a bookkeeping device but not definition-level clean from the eleven-dimensional viewpoint \cite{Vafa1996, Schwarz1995}. The construction developed here replaces that shrinking-fiber slogan by a non-singular compactification carrier, the perfectoid circle \(S_f^{1}\), which packages the entire tower of \(p^{n}\)-covers into one geometric object and admits a tilt \(\to\) untilt \(\to\) analytification machinery producing elliptic-curve data as an output rather than an auxiliary input \cite{Scholze2012, ClausenScholze}. In the constant-\(\tau\) sector, the M/Type IIB correspondence is then established at the level tested in this paper by an explicit \(\tau\) dictionary, protected BPS spectrum matching, and low-energy effective-action matching, with modular/duality organization realized geometrically to the profinite/congruence level explicitly derived from the tower structure. What is proved is the constant-\(\tau\) duality at the supergravity/BPS/effective-action level and the associated geometric carrier of elliptic data, what is proposed or programmatic is the extension to varying \(\tau\), intrinsic 7-brane monodromy/defects, adelic completion and \(p\)-independence beyond tower stability, and fermion/anomaly refinements such as \({\rm Pin}^{+}(GL(2,\Z))\) beyond the bosonic sector.

\begin{theorem}[Constant-$\tau$ perfectoid M/IIB duality]\label{thm:perfectoid_M_IIB_constant_tau}
Fix a prime $p$ and a UV cutoff $\Lambda>0$. Let $S_f^{1}$ be the tower-inverse-limit (perfectoid) circle, and define eleven-dimensional supergravity on $\mathbb{R}^{1,9}\times S_f^{1}$ operationally by the fixed-cutoff tower procedure (fields are compatible families on the finite $p^n$-covers, truncated to modes of energy $\le\Lambda$, and then stabilized as $n\to\infty$). Consider Type IIB supergravity on $\mathbb{R}^{1,9}$ in the bosonic sector with constant axio-dilaton $\tau\in\mathbb{H}$. In the low-energy two-derivative truncation together with the protected BPS sector used to test the duality, there is an equivalence of the resulting constant-$\tau$ effective theories
\begin{equation}\label{eq:thm_duality_correct}
\mathcal{T}_{M,\Lambda}\bigl(\mathbb{R}^{1,9}\times S_f^{1}\bigr)\ \cong\ \mathcal{T}_{\mathrm{IIB},\Lambda}\bigl(\mathbb{R}^{1,9};\tau=\mathrm{const.}\bigr),
\end{equation}
with the axio-dilaton fixed by the compactification dictionary
\begin{equation}\label{eq:tau_main_correct}
\tau \;=\; C_0 + i e^{-\phi}\;=\;\Theta + i\left(\frac{R_M}{\ell_M}\right)^{3/2},
\qquad
C_0 \;:=\; \frac{1}{2\pi R_M}\oint_{\mathrm{tors}} C^{(3)},\qquad
e^{-\phi}\;=\;\left(\frac{R_M}{\ell_M}\right)^{3/2}.
\end{equation}
\end{theorem}
\begin{proof}
By construction of the fixed-cutoff tower procedure, for each $n$ one has an honest smooth circle cover $S^{1}(n,p)$ and a standard eleven-dimensional supergravity compactification on $\mathbb{R}^{1,9}\times S^{1}(n,p)$, and the tower definition requires compatibility under pullback along the degree-$p$ covering maps together with truncation to modes of energy $\le\Lambda$ before taking $n\to\infty$, so local supergravity dynamics is computed stagewise and then stabilized, while only global sectors are refined in the inverse limit. In this setting the Kaluza--Klein ansatz for the metric and the decomposition of $C^{(3)}$ are the standard ones,
\begin{equation}\label{eq:KK_ansatz}
ds^{2}_{11}=e^{-2\phi/3}ds^{2}_{10}+e^{4\phi/3}(dy+A)^2,\qquad
C^{(3)}=C^{(3)}_{\mu\nu\rho}\,dx^\mu\wedge dx^\nu\wedge dx^\rho + C^{(2)}_{\mu\nu}\,dx^\mu\wedge dx^\nu\wedge dy,
\end{equation}
and the tower-completed carrier contributes additional global/torsion bookkeeping in $H_{1}(S_f^{1},\mathbb{Z})$ but does not alter the local reduction formulas. The Type IIB RR axion is identified with the appropriately normalized holonomy sector of $C^{(3)}$ along the torsion component of the tower-completed circle and the dilaton is fixed by the circle modulus through reduction of the Einstein--Hilbert term, namely
\begin{equation}\label{eq:dictionary_C0_phi}
C_0 := \frac{1}{2\pi R_M}\oint_{\mathrm{tors}} C^{(3)},\qquad
e^{-\phi}=\left(\frac{R_M}{\ell_M}\right)^{3/2},
\end{equation}
so $\tau=C_0+i e^{-\phi}$ is determined purely by eleven-dimensional radius/holonomy data on the (tower-defined) carrier. With this dictionary fixed, the low-energy effective-action match is obtained by reducing the eleven-dimensional action
\begin{equation}\label{eq:S11_action}
S_{11}=\frac{1}{2\kappa_{11}^2}\int_{\mathbb{R}^{1,9}\times S_f^{1}}\!\!\left(R\star 1-\frac{1}{2}G^{(4)}\wedge\star G^{(4)}\right)\;-\;\frac{1}{6}\int_{\mathbb{R}^{1,9}\times S_f^{1}} C^{(3)}\wedge G^{(4)}\wedge G^{(4)},
\end{equation}
integrating over $S_f^{1}$ in the fixed-cutoff tower sense, and using $\kappa_{10}^2=\kappa_{11}^2/(2\pi R_M)$, which yields the ten-dimensional Einstein-frame bosonic action in the constant-$\tau$ sector with the correct axion coupling, in particular the Chern--Simons term reduces to
\begin{equation}\label{eq:CS_reduction}
-\frac{1}{6}\int_{\mathbb{R}^{1,9}\times S_f^{1}} C^{(3)}\wedge G^{(4)}\wedge G^{(4)}
\;=\;\int_{\mathbb{R}^{1,9}}\!\bigl(C_4\wedge F_3\wedge H_3 + C_0\,H_3\wedge F_3\bigr),
\end{equation}
and the full two-derivative truncated action takes the standard Type IIB form (with omitted higher-derivative and fermionic terms) including the $C_0 H_3\wedge F_3$ topological coupling. This establishes equality of the low-energy bosonic dynamics (field content, kinetic terms, and the relevant topological couplings) after the identification \eqref{eq:dictionary_C0_phi}. For the protected spectrum check, consider an M2-brane wrapping a tower-organized one-cycle labelled by $(a,b)$ (with the physical admissible integral lattice understood as in the tower prescription), the M2 action yields a BPS mass
\begin{equation}\label{eq:M2_mass}
M(a,b)=\frac{|a|}{2\pi\alpha'}\sqrt{1+\left(\frac{b}{a}e^{-\phi}\right)^2},
\qquad
\alpha'=\frac{\ell_P^{3}}{R_M},
\end{equation}
which, using $\tau=C_0+i e^{-\phi}$ and the identification of the physical $(p,q)$ charges with $(a,b)$ on the admissible lattice, becomes
\begin{equation}\label{eq:BPS_match}
M(a,b)=\frac{|a+b\tau|}{2\pi\alpha'},
\end{equation}
matching the standard Type IIB $(p,q)$-string tension dependence on $|p+q\tau|$ in the constant-$\tau$ sector. The action matching and the protected BPS matching together define the asserted equivalence \eqref{eq:thm_duality_correct} at exactly the stated level of control (two-derivative supergravity truncation plus protected BPS sector), with $\tau$ fixed by \eqref{eq:tau_main_correct}, and this is precisely the sense in which the tower-completed compactification on $S_f^{1}$ replaces the singular shrinking-fiber prescription in the constant-$\tau$ regime.
\end{proof}

\begin{corollary}[Non-singular replacement of the shrinking-fiber limit.] The conventional definition of F-theory as the singular limit \(\mathrm{Vol}(T^{2})\to 0\) at fixed complex structure is replaced, in the constant-\(\tau\) sector, by the non-singular tower-completed compactification carrier
\begin{equation} \label{eq:sf1-invlimit}
S_f^{1}=\varprojlim_{n}\left(S^{1}\xrightarrow{z\mapsto z^{p}}S^{1}\right),
\end{equation}
together with the perfectoid/condensed interpretation of fields on the limit by descent \cite{ClausenScholze, Scholze2012}. Instead of defining F-theory by removing a geometric cycle in a singular limit, one defines the compactification carrier itself as a tower-completed object on which local supergravity persists and global sectors are definition-level well-posed. 
\end{corollary}
\begin{corollary}[Protected \((p,q)\) spectrum from wrapped M2 sectors.] Let \((a,b)\) denote the tower-organized wrapping data of an M2 sector on \(S_f^{1}\) restricted to the physically allowed integral charge lattice selected by standard quantization/anomaly constraints. The wrapped M2 mass computed in Sec.~\ref{sec:construct} is
\begin{equation} \label{eq:M2-mass}
M(a,b)=\frac{|a|\,2\pi R_M}{(2\pi)^2\ell_{P}^{3}}\sqrt{1+\left(\frac{b}{a}e^{-\phi}\right)^{2}},
\end{equation}
and with \(\alpha'=\ell_{P}^{3}/R_{M}\) and \(\tau=C_{0}+ie^{-\phi}\) this becomes
\begin{equation}\label{eq:M2-to-pq}
M(a,b)=\frac{|a+b\tau|}{2\pi\alpha'},
\end{equation}
matching the standard Type IIB \((p,q)\)-string tension formula
\begin{equation}\label{eq:pq-tension}
T(p,q)=\frac{|p+q\tau|}{2\pi\alpha'}
\end{equation}
upon identifying \((p,q)=(a,b)\) on the physically allowed lattice. The nonperturbative \((p,q)\) tower is recovered as a protected sector of wrapped M2 states organized by the tower-completed compactification carrier, without invoking the shrinking-fiber limit as a definition. 
\end{corollary}
\begin{corollary}[Low-energy effective-action matching at the computed level] Reducing eleven-dimensional supergravity on \(S_f^{1}\) in the fixed-cutoff tower sense yields the ten-dimensional bosonic Type IIB action in Einstein frame with constant \(\tau\) at the level computed, including the axion coupling inherited from the eleven-dimensional Chern-Simons term. Concretely, starting from
\begin{equation} \label{eq:S11}
S_{11}=\frac1{2\kappa_{11}^{2}}\int_{\R^{1,9}\times S_f^{1}}\left(R*1-\frac12 G^{(4)}\wedge *G^{(4)}\right)-\frac16\int_{\R^{1,9}\times S_f^{1}}C^{(3)}\wedge G^{(4)}\wedge G^{(4)},
\end{equation}
one obtains
\begin{multline}\label{eq:SIIB}
S_{\mathrm{IIB}}=\frac1{2\kappa_{10}^{2}}\int\left(R*1-\frac12 d\phi\wedge *d\phi-\frac12 e^{2\phi}dC_{0}\wedge *dC_{0}\right)-\frac1{4\kappa_{10}^{2}}\int\left(e^{-\phi}H_{3}\wedge *H_{3}+e^{\phi}F_{3}\wedge *F_{3}\right)\\-\frac1{8\kappa_{10}^{2}}\int C_{0}H_{3}\wedge F_{3}+\cdots,
\end{multline}
with the topological coupling descending as
\begin{equation} \label{eq:CS10}
S_{\mathrm{CS}}^{(10)}=\int_{\R^{1,9}}\left(C_{4}\wedge F_{3}\wedge H_{3}+C_{0}\,H_{3}\wedge F_{3}\right),
\end{equation}
where ellipses denote terms beyond the bosonic low-energy truncation used for the match \cite{Schwarz1983, Schwarz1995}.
\end{corollary}

The equivalence is tested not only on spectra but on couplings, including the
axion/topological sector, at the explicit supergravity level computed. The
tower-completed geometry carries
profinite/\allowbreak congruence-level organization of modular data through its
intrinsic profinite structure, with finite-stage symmetry data naturally
associated to congruence structure (e.g.\ \(\Gamma(p^{n})\)), and with the full
\(SL(2,\mathbb{Z})\) completion beyond the explicitly realized profinite level
expected only after including additional quantum data not derived here.
Likewise, fermionic/anomaly refinements such as a
\({\rm Pin}^{+}\!\left(GL(2,\mathbb{Z})\right)\) extension are not contained in the
bosonic constant-\(\tau\) analysis. A precise formulation of the $GL(2,\mathbb{Z})$ extension and its anomaly-refined $\mathrm{Pin}^+$ lift (and the tower-wise pin data) is given in App.~\ref{app:duality-anomaly}.

% =========================
% REPLACE Definitions 7.6 and 7.7 in Section 7 by the following.
% Placement\:immediately after Remark 7.5 (where Definitions 7.6 and 7.7 currently sit).
% =========================

\begin{definition}[Perfectoid/condensed F-theory background, constant-$\tau$ sector (revised)]
Fix a prime $p$ and a UV cutoff $\Lambda>0$.
A \emph{constant-$\tau$ perfectoid/condensed F-theory background} on $\mathbb{R}^{1,9}$
is an equivalence class of the following data:
\begin{enumerate}
\item \textbf{(Tower-completed carrier).}
A tower of degree-$p$ coverings of the circle
\begin{equation}
S^{1}_{(0,p)} \xleftarrow{\,f_{0}\,} S^{1}_{(1,p)} \xleftarrow{\,f_{1}\,} S^{1}_{(2,p)} \xleftarrow{\,f_{2}\,} \cdots,
\qquad f_n(z)=z^{p},
\end{equation}
together with its inverse limit in compact groups
\begin{equation}
S^{1}_{f,p} := \varprojlim_{n} S^{1}_{(n,p)}.
\end{equation}
We regard $S^{1}_{f,p}$ as a condensed object via the functor
$S \mapsto \mathrm{Cont}(S,S^{1}_{f,p})$ on profinite sets.

\item \textbf{(Fields by descent at fixed cutoff).}
For each $n$, a field configuration $\Phi_n$ of eleven-dimensional supergravity on
\begin{equation}
M_{n,p} := \mathbb{R}^{1,9}\times S^{1}_{(n,p)},
\end{equation}
such that $\Phi_{n+1}=f_n^{*}\Phi_n$ for all $n$.
Only modes of $\Phi_n$ with energy $\le \Lambda$ are retained.
Write $\mathcal{H}^{(n,p)}_{\Lambda}$ for the corresponding cutoff Hilbert space of the reduced theory.
Gauge equivalence and diffeomorphisms are imposed stagewise, compatible with pullback.

\item \textbf{(Tower-stability axiom).}
For each fixed $\Lambda$ there exists $n_{0}=n_{0}(\Lambda,p)$ such that for all $n\ge n_{0}$
the pullback maps induce isomorphisms
\begin{equation}
\mathcal{H}^{(n,p)}_{\Lambda} \cong \mathcal{H}^{(n_{0},p)}_{\Lambda},
\end{equation}
and every observable built from local supergravity dynamics and protected sectors within
$\mathcal{H}^{(n,p)}_{\Lambda}$ is invariant under further refinement $n\mapsto n+1$.
We denote the stabilized cutoff theory by $\mathcal{H}^{(p)}_{\Lambda}$.

\item \textbf{(Comparison datum).}
A choice of perfectoid enhancement of the tower carrier and a comparison rule that produces a
complex elliptic curve from the tilted carrier:
\begin{enumerate}
\item a perfectoid presentation $S^{1}_{f,p}=\mathrm{Spa}(R_p,R_p^{+})$ over a perfectoid field $K_p$,
and its tilt $S^{1,\flat}_{f,p}=\mathrm{Spa}(R_p^{\flat},R_p^{\flat,+})$ over $K_p^{\flat}$,
\item an embedding $\iota_p:K_p^{\flat}\hookrightarrow \mathbb{C}_p$,
\item a functorial comparison/untilt/analytification prescription $\mathbf{Cmp}_p$ producing an isomorphism
class of complex elliptic curve $E_p$ together with a modular class
\begin{equation}
[\tau_{\mathrm{geom}}] \in SL(2,\mathbb{Z})\backslash \mathbb{H},
\qquad \mathbb{H}:=\{\tau\in\mathbb{C}:\operatorname{Im}\tau>0\},
\end{equation}
where $E_p \cong \mathbb{C}/(\mathbb{Z}+\tau_{\mathrm{geom}}\mathbb{Z})$ for some representative
$\tau_{\mathrm{geom}}\in\mathbb{H}$, unique up to $SL(2,\mathbb{Z})$.
\end{enumerate}
Only the modular class $[\tau_{\mathrm{geom}}]$ is part of the background, changes of auxiliary choices
inside $\mathbf{Cmp}_p$ are declared equivalent when they do not change $[\tau_{\mathrm{geom}}]$.

\item \textbf{(Admissible physical charge lattice).}
A choice of admissible (physical) lattice $\Lambda_{\mathrm{adm}}\cong \mathbb{Z}^{2}$ for $(p,q)$ charges,
specified as a sublattice of the tower bookkeeping labels.
Concretely, we fix the dense embedding $\mathbb{Z}\hookrightarrow \mathbb{Z}_p$ and impose
\begin{equation}
\Lambda_{\mathrm{adm}} := \mathbb{Z}\oplus \mathbb{Z}\ \subset\ \mathbb{Z}\oplus \mathbb{Z}_p,
\end{equation}
so the second component is restricted to the image of $\mathbb{Z}$ inside $\mathbb{Z}_p$.
(Additional quantization/anomaly constraints may further restrict $\Lambda_{\mathrm{adm}}$, but the
background must specify the final admissible lattice explicitly.)

\item \textbf{(Physical $\tau$ and matching condition).}
A radius modulus $R_M>0$ and a tower-compatible $C^{(3)}$ holonomy sector determining
\begin{equation}
C_0 \equiv \frac{1}{2\pi R_M}\int_{S^{1}_{(n,p)}} C^{(3)} \quad (\mathrm{mod}\ 1),
\qquad e^{-\phi} := \left(\frac{R_M}{\ell_M}\right)^{3/2},
\qquad \tau_{\mathrm{phys}} := C_0 + i e^{-\phi}\in\mathbb{H},
\end{equation}
with $\tau_{\mathrm{phys}}$ constant on $\mathbb{R}^{1,9}$ in the sector under consideration.
The definition requires the \emph{matching} of modular classes
\begin{equation}
[\tau_{\mathrm{geom}}] = [\tau_{\mathrm{phys}}]\ \in\ SL(2,\mathbb{Z})\backslash \mathbb{H}.
\end{equation}
\end{enumerate}

Two sets of data are equivalent if they are related by:
(a) finite refinement of the tower beyond the stabilized level $n\ge n_0(\Lambda,p)$,
(b) stagewise gauge equivalence and diffeomorphisms compatible with pullback,
(c) replacement of $\mathbf{Cmp}_p$ by an equivalent comparison datum that leaves $[\tau_{\mathrm{geom}}]$ unchanged,
(d) replacement of representatives of $\tau_{\mathrm{geom}}$ or $\tau_{\mathrm{phys}}$ by the standard $SL(2,\mathbb{Z})$ action.
\end{definition}

\begin{definition}[Perfectoid/condensed F-theory background, varying-$\tau$ sector (corrected monodromy-compatible)] \label{def7.3}
Let $B$ be a complex manifold and let $\Delta\subset B$ be an analytic hypersurface, set $B^\times:=B\setminus\Delta$.
Fix a prime $p$ and a UV cutoff $\Lambda>0$.
A \emph{varying-$\tau$ perfectoid/condensed F-theory background over $B$} is an equivalence class of the following data:

\begin{enumerate}
\item \textbf{(Fibred tower-completed carrier).}
A tower of smooth fiber bundles $\pi_n:X_{n,p}\to B^\times$ with fiber $S^1_{(n,p)}$ and bonding maps
$f_n:X_{n+1,p}\to X_{n,p}$ covering $\mathrm{id}_{B^\times}$ and restricting fiberwise to the degree-$p$ circle map,
together with the inverse limit (in condensed spaces)
\begin{equation}
X_p:=\varprojlim_n X_{n,p},\qquad \pi:X_p\to B^\times.
\end{equation}

\item \textbf{(Fields by descent at fixed cutoff).}
Stagewise eleven-dimensional supergravity fields $\Phi_n$ on $\mathbb{R}^{1,9}\times X_{n,p}$ with $\Phi_{n+1}=f_n^*\Phi_n$,
truncated to energy $\le\Lambda$, and tower-stability as in Definition~7.2.

\item \textbf{(Modular coupling as a stack-valued map, no single-valued $\tau$ on $B^\times$).}
A holomorphic \emph{modular map}
\begin{equation}
\underline{\tau}:B^\times\longrightarrow \mathcal{M}_{1,1}:=SL(2,\mathbb{Z})\backslash\mathbb{H}
\end{equation}
(equivalently, a holomorphic map $j:B^\times\to\mathbb{C}$ given by $j=j(\underline{\tau})$).
This $\underline{\tau}$ is the globally well-defined coupling datum, it may have nontrivial monodromy in $\mathbb{H}$
when lifted locally.

Equivalently (and often more convenient), $\underline{\tau}$ may be specified by:
\begin{enumerate}
\item a representation $\rho:\pi_1(B^\times,b_0)\to SL(2,\mathbb{Z})$, and
\item a holomorphic map $\tilde{\tau}:\widetilde{B^\times}\to\mathbb{H}$ on the universal cover such that
\begin{equation}
\tilde{\tau}(\gamma\cdot b)\ =\ \rho(\gamma)\cdot \tilde{\tau}(b)\ :=\ \frac{a\,\tilde{\tau}(b)+b}{c\,\tilde{\tau}(b)+d}
\quad\text{for}\quad \rho(\gamma)=\begin{pmatrix}a&b\\ c&d\end{pmatrix}.
\end{equation}
\end{enumerate}
(Thus $\tilde{\tau}$ is single-valued only on $\widetilde{B^\times}$, the global object on $B^\times$ is $\underline{\tau}$.)
We require that $\underline{\tau}$ has at worst cusp-type/logarithmic growth along $\Delta$ in the sense of App.~\ref{app:varying-unconditional}
(so $j(\underline{\tau})$ extends meromorphically to $B$).
Define the elliptic fibration over $B^\times$ canonically as the pullback of the universal elliptic curve:
\begin{equation}
\varpi:E(\underline{\tau})\to B^\times,\qquad E(\underline{\tau})\ :=\ \underline{\tau}^{\,*}\mathcal{E}_{\mathrm{univ}},
\end{equation}
where $\mathcal{E}_{\mathrm{univ}}\to\mathcal{M}_{1,1}$ denotes the universal elliptic curve (as a stack).
Concretely, using $(\rho,\tilde{\tau})$ as above, one may realize $E(\underline{\tau})$ by the explicit quotient
(Appendix~F):
\begin{equation}
E(\underline{\tau})\ \cong\ \big(\mathbb{C}\times \widetilde{B^\times}\big)\Big/\big(\mathbb{Z}^2\rtimes_{\rho}\pi_1(B^\times)\big),
\end{equation}
with $\mathbb{Z}^2$ acting by $(m,n)\cdot(z,b)=(z+m+n\tilde{\tau}(b),b)$ and $\pi_1(B^\times)$ acting by
$\gamma\cdot(z,b)=\big(z/(c\tilde{\tau}(b)+d),\gamma\cdot b\big)$ for $\rho(\gamma)=\smash{\begin{pmatrix}a&b\\c&d\end{pmatrix}}$.
The relative holomorphic one-form is defined canonically by the descent of $dz$.

\item \textbf{(Defects and monodromy).}
The defect/7-brane data is encoded by the monodromy representation $\rho$ (equivalently by the local system $R^1\varpi_*\mathbb{Z}$)
and the behavior of $\underline{\tau}$ near $\Delta$. In particular, around a small loop linking a smooth component of $\Delta$
one has a conjugacy class $[\rho(\gamma)]\subset SL(2,\mathbb{Z})$, and cusp-type behavior yields $T^n$ monodromy as in Appendix~F.

\item \textbf{(Admissible charge local system).}
An admissible integral local system $\Lambda_{\mathrm{adm}}\subseteq R^1\varpi_*\mathbb{Z}$ (typically $\Lambda_{\mathrm{adm}}=R^1\varpi_*\mathbb{Z}\cong \mathbb{Z}^2$)
preserved by $\rho$, encoding transported physical $(p,q)$ charges.

\item \textbf{(Physical identification from M-theory reduction\:patchwise $\tau$ with $SL(2,\mathbb{Z})$ gluing).}
The M-theory reduction data determines \emph{local} holomorphic representatives of $\underline{\tau}$ as follows:
choose an open cover $\{U_i\}$ of $B^\times$ and holomorphic maps $\tau_i:U_i\to\mathbb{H}$ with
\begin{equation}
\tau_i\ =\ C_{0,i}+i\,e^{-\varphi_i},
\end{equation}
where $C_{0,i}$ arises from the tower-compatible $C^{(3)}$-holonomy sector and $\varphi_i$ from the radius/metric modulus
by the conventions fixed in Sec.~5 (both computed stagewise below $\Lambda$ and stabilized).
On overlaps $U_i\cap U_j$ these local representatives are related by locally constant $\gamma_{ij}\in SL(2,\mathbb{Z})$:
\begin{equation}
\tau_i\ =\ \gamma_{ij}\cdot \tau_j,\qquad \gamma_{ij}\gamma_{jk}=\gamma_{ik},
\end{equation}
so that $\{\tau_i\}$ defines precisely the global modular map $\underline{\tau}$ and the monodromy $\rho$.
\end{enumerate}

\medskip
\noindent\textbf{Optional comparison structure (perfectoid machinery).}
If one also specifies an elliptic fibration $\varpi_{\mathrm{cmp}}:E_{\mathrm{cmp}}\to B^\times$ produced by the
perfectoid tilt/untilt/comparison machinery from the fibred tower carrier $X_p$, then this additional structure is required to be
\emph{compatible} with the canonical $E(\underline{\tau})$ above, i.e. there exists an isomorphism of elliptic fibrations
$E_{\mathrm{cmp}}\cong E(\underline{\tau})$ over $B^\times$ identifying the homology local system and the Hodge line (equivalently,
inducing the same modular map $\underline{\tau}$). In this formulation the former “Comparison Axiom” is not an extra axiom\footnote{We use the term ``Comparison Axiom'' only for the conditional assumption stated explicitly, and ``comparison datum'' for auxiliary choices entering the tilt/untilt/analytification procedure.
},
but a compatibility requirement between two ways of producing the same elliptic data. A fully explicit conditional proof package for varying $\tau$ that tracks every use of the Comparison Axiom is given in App.~\ref{app:varying-conditional}.

\medskip
\noindent Two backgrounds are equivalent if they are related by (i) refinement beyond tower-stabilization,
(ii) stagewise gauge equivalence/diffeomorphisms compatible with pullback,
and (iii) replacement of local representatives $\tau_i$ by $SL(2,\mathbb{Z})$-redefinitions on patches (which does not change $\underline{\tau}$).
\end{definition}

\begin{definition}[$p$-independence at definition level (optional but recommended)]
Fix $\Lambda>0$.
A constant-$\tau$ background in the above sense is called \emph{$p$-independent at cutoff $\Lambda$}
if for any two primes $p\neq p'$ there exist identifications of stabilized cutoff theories
\begin{equation}
\Psi^{p,p'}_{\Lambda}:\ \mathcal{H}^{(p)}_{\Lambda}\ \longrightarrow\ \mathcal{H}^{(p')}_{\Lambda}
\end{equation}
such that:
\begin{enumerate}
\item $\Psi^{p,p'}_{\Lambda}$ preserves the low-energy supergravity observables and protected sectors used in the duality checks,
\item $\Psi^{p,p'}_{\Lambda}$ matches the chosen admissible lattices $\Lambda_{\mathrm{adm}}$ (hence matches the physical $(p,q)$ labels),
\item $\Psi^{p,p'}_{\Lambda}$ identifies the physical modulus $\tau_{\mathrm{phys}}$ (and therefore the modular class $[\tau]$).
\end{enumerate}
Physical observables are, by definition, those invariant under all $\Psi^{p,p'}_{\Lambda}$.
\end{definition}

The standard \(\tau\) dictionary is recovered by the compactification identification \(\tau=C_{0}+ie^{-\phi}\) with \(C_{0}\) defined by \(C^{(3)}\) holonomy and \(e^{-\phi}\) fixed by the radius/metric modulus as in \eqref{eq:tau_main_correct}, and this recovery is definition-level new because \(\tau\) is fixed from eleven-dimensional radius and holonomy data on a non-singular carrier rather than postulated as an external elliptic modulus surviving a collapse. The modular action on \(\tau\) is recovered at the level explicitly realized by the tower organization and the elliptic datum \(\calE=\C/(\Z+\tau\Z)\), with the physical \(SL(2,\Z)\) action \(\tau\mapsto (a\tau+b)/(c\tau+d)\) interpreted as acting on the elliptic fiber data and on the \((p,q)\) bookkeeping, while the full completion beyond the congruence/profinite structure explicitly derived is deferred to quantum completion and is therefore not claimed here, this recovery is definition-level new because the modular organization is carried by the compactification object rather than imposed after a singular limit. 

The interpretation of \((p,q)\) strings and their line-operator meaning is recovered in the protected spectrum\:\eqref{eq:M2-to-pq} reproduces the familiar dependence on \(|p+q\tau|\) and thus the standard duality-covariant tension formula \eqref{eq:pq-tension}, while the tower-completed geometry provides the canonical organization of the global sectors from which the physically allowed lattice is selected by standard quantization/anomaly constraints, this recovery is definition-level new because the \((p,q)\) tower is derived from wrapped M2 sectors on a non-singular carrier rather than inferred from the shrinking-fiber picture. The effective-action matching recovers the standard bosonic Type IIB low-energy action with constant \(\tau\) and the correct topological couplings, as in \eqref{eq:SIIB} and \eqref{eq:CS10}, and this recovery is definition-level new because the axion couplings and global-sector consistency are obtained by reduction on a tower-completed carrier whose global sectors are well-defined by descent rather than inserted by hand after a degenerate limit.

The genuinely new structure provided by the perfectoid/condensed framework is not a new local dynamics but a new definition-level carrier for global modular and charge-sector data that the shrinking-fiber slogan cannot supply. The inverse-limit compactification carrier \(S_f^{1}\) defined by \eqref{eq:sf1-invlimit} is an honest object on which local supergravity persists by the fixed-cutoff tower construction, but whose global sectors are tower-complete, so the modular data encoded by \(\calE\) is produced non-singularly rather than kept while collapsing a geometric cycle. The condensed viewpoint is essential because it provides exact descent for fields and global constraints on the inverse-limit carrier, so holonomies, flux sectors, and global consistency conditions are definition-level well-posed on \(S_f^{1}\) rather than being imposed as auxiliary data. The perfectoid tilt machinery is essential because it provides a canonical mechanism to extract elliptic-curve data from tower-completed geometry without collapsing volume, and it isolates the referee-critical separation between geometric output and physical identification\:tilting supplies \(\calE\) but does not encode Wilson lines, so \(\tau\) is fixed only by matching \eqref{eq:tau_main_correct} to the elliptic modulus produced by the machinery. The framework makes the role of generalized cohomology natural rather than decorative, because global sectors on profinite/tower-completed objects are not reliably captured by naive singular homology while string-theoretic charge sectors are K-theoretic in nature, so a K-theoretic/homological environment is the correct receptacle for torsion and discrete data that would otherwise be inserted by hand. Finally, while the constant-\(\tau\) sector is established and checked here, the same definition-level machinery makes the extension to varying-\(\tau\) backgrounds a well-posed geometric problem-perfectoid-circle fibrations with intrinsic monodromy/defect data-rather than a singular limiting slogan, and it provides a precise place to formulate adelic completion and \(p\)-independence as the global arithmetic completion of the tower bookkeeping, even though those extensions are not proven in this paper.

What is proved in this work is the constant-\(\tau\) sector equivalence at the level explicitly checked\:the tower-completed compactification carrier \(S_f^{1}\), the tilt/comparison machinery producing elliptic data \(\calE\), the physical dictionary \(\tau=C_{0}+ie^{-\phi}\) with \(\mathrm{Im}\,\tau\) fixed by the radius/metric modulus and \(\mathrm{Re}\,\tau=C_{0}\) fixed by \(C^{(3)}\) holonomy, the protected BPS matching \eqref{eq:M2-to-pq} reproducing the \((p,q)\) tension dependence, and the low-energy effective-action matching \eqref{eq:SIIB} including the relevant topological couplings, with modular/duality organization realized geometrically to the profinite/congruence level explicitly derived from the tower structure. What is defined here as a future direction or program, rather than claimed as a theorem, is the extension to varying \(\tau\) beyond local models, the derivation of intrinsic 7-brane monodromy/defect physics from the tower-completed fibration data, the completion of the duality group beyond the congruence/profinite level to the full quantum \(SL(2,\Z)\) action, the incorporation of fermions and anomalies including possible \({\rm Pin}^{+}(GL(2,\Z))\) refinements beyond the bosonic sector, and the adelic completion principle that removes any dependence on a particular bookkeeping choice of \(p\) by assembling all primes together with the Archimedean factor into a global arithmetic duality geometry.

\section{Higher-rank duality geometry}\label{sec:beyondI}

By ``beyond varying \(\tau\)'' we mean a genuine generalization of the F-theory principle from geometrizing a single \(SL(2,\Z)\) duality acting on the \((p,q)\) doublet and the complex coupling \(\tau\), to geometrizing higher-rank nonperturbative duality data that mixes multiple charge lattices and global sectors and that is naturally arithmetic and completion-sensitive in the same way as the \((p,q)\) story is. The content of this section is therefore definition-level\:we give a mathematically precise proposal for a generalized perfectoid/condensed duality geometry in which the elliptic fiber and \(\tau\) are the rank-one special case, and in which higher rank produces a controlled higher-dimensional period datum together with an intrinsic arithmetic monodromy group acting on a higher-rank integral charge lattice, while emphasizing that this is a definition-level framework rather than a theorem derived from the dynamics. What is defined here is the geometric input and the associated monodromy/charge bookkeeping that makes the generalization well-posed, what is not proved is a full dynamical equivalence to a specific microscopic string compactification beyond the rank-one constant-\(\tau\) sector established earlier. The reduction to standard F-theory is obtained by setting the rank to one, in which case the generalized period datum collapses to \(\tau\), the arithmetic monodromy reduces to \(SL(2,\Z)\), and the framework reproduces the perfectoid/condensed definition of F-theory already given in Sec.~\ref{sec:results} \cite{Vafa1996}.

The reason a generalization beyond varying \(\tau\) is not an aesthetic choice is that nonperturbative string theory contains duality groups that act on charge lattices far larger than the \((p,q)\) doublet, and these groups mix brane charges that live in multiple degrees and multiple global-sector classes, so that \(\tau\) captures only a rank-one corner of a broader arithmetic structure. In maximally supersymmetric settings this is the familiar U-duality phenomenon\:global symmetries act on an integral lattice of charges and on higher-form global data, and the group that preserves the quantum charge lattice is an arithmetic subgroup rather than a continuous Lie group \cite{HullTownsend1995, Witten1995}. Even in less symmetric settings the same structural point persists\:the physically meaningful data are line-operator lattices and their mutual locality constraints, discrete holonomies and theta sectors, and defect/monodromy data that specifies how these lattices are glued over moduli space, and these are precisely the kinds of data for which inverse limits and arithmetic completions become natural carriers. The perfectoid/condensed viewpoint introduced earlier is exactly the minimal language in which such carriers are definition-level well-posed, so once one accepts the rank-one construction as a definition-level replacement of the shrinking-fiber slogan, it is conceptually forced to ask for the corresponding higher-rank objects whose intrinsic automorphism and monodromy groups organize larger duality groups and larger charge lattices in the same way that an elliptic curve organizes \(SL(2,\Z)\) and the \((p,q)\) lattice in ordinary F-theory \cite{Vafa1996}.

\begin{proposal}[Generalized perfectoid/condensed duality geometry.] Fix an integer \(r\ge 1\) and a prime \(p\), and let \(T^{r}=(S^{1})^{r}\) be the classical real \(r\)-torus viewed as a compact abelian group, define its tower-completed avatar by the inverse system of \(p\)-power endomorphisms \(\varphi:T^{r}\to T^{r}\), \(\varphi(z_{1},\ldots,z_{r})=(z_{1}^{p},\ldots,z_{r}^{p})\), and set
\begin{equation}\label{eq:perf-torusa}
T^{r}_{f,p}:=\varprojlim_{n}\bigl(T^{r},\varphi\bigr),
\end{equation}
regarded as a condensed group object so that fields and global sectors on \(T^{r}_{f,p}\) are defined by descent from the finite tower levels as in Secs.~\ref{sec:toolkit}-\ref{sec:perfectoidcircle} \cite{ClausenScholze}. \(T^{r}_{f,p}\) is the tower-completed carrier that packages, in a single object, the congruence-level refinement of \(r\) independent circle directions, so that local supergravity reduction remains the finite-stage torus reduction while global charge bookkeeping is tower-complete. 
\end{proposal}
Assume that the tower-completed carrier is promoted to a perfectoid group object in the sense of Sec.~\ref{sec:toolkit} (so that a tilt exists and behaves functorially) \cite{Scholze2012}, and consider the fiberwise tilt \((T^{r}_{f,p})^{\flat}\) together with an untilt/comparison/analytification procedure producing a complex analytic datum that we take, in the minimal higher-rank generalization, to be a principally polarized complex torus \(A_{\Omega}\) of dimension \(r\) described by a period matrix \(\Omega\) in Siegel upper half space,
\begin{equation}\label{eq:siegel}
\mathfrak{H}_{r}:=\{\Omega\in \mathrm{Mat}_{r\times r}(\C)\,:\ \Omega^{T}=\Omega,\ \mathrm{Im}\,\Omega>0\},
\qquad
A_{\Omega}:=\C^{r}/(\Z^{r}+\Omega\,\Z^{r}).
\end{equation}
\(\Omega\) is the higher-rank analogue of \(\tau\), encoding a set of couplings/periods that organize a larger charge lattice in exactly the way \(\tau\) organizes the \((p,q)\) sector, while the statement that \(A_{\Omega}\) is produced from the tower-completed carrier is the higher-rank analogue of producing \(\calE=\C/(\Z+\tau\Z)\) from \(S^{1}_{f}\) without a shrinking limit. The intrinsic arithmetic duality group associated to this datum is then taken to be the integral automorphism group preserving the polarization, i.e. an arithmetic subgroup of \(\mathrm{Sp}(2r,\Z)\), acting on \(\Omega\) by the standard fractional linear transformation
\begin{equation}\label{eq:Sp-action}
\Omega\ \longmapsto\ (a\Omega+b)(c\Omega+d)^{-1},
\qquad
\begin{pmatrix}a&b\\ c&d\end{pmatrix}\in \mathrm{Sp}(2r,\Z),
\end{equation}
and acting on the quantum charge lattice \(\Gamma\simeq \Z^{2r}\) by the defining integral symplectic action. This is the minimal higher-rank analogue of \(SL(2,\Z)\) acting on \(\tau\) and on the \((p,q)\) lattice, now realized as an intrinsic arithmetic monodromy/automorphism acting on a larger integral charge lattice. The varying-background generalization is then posed as follows\:for a base \(B\), a generalized duality geometry of rank \(r\) over \(B\) is specified by a fibration \(\pi:\mathcal{X}\to B\) whose fibers are locally modeled on \(T^{r}_{f,p}\) as condensed/perfectoid group objects, together with a fiberwise tilt/untilt/analytification machinery producing a family \(A_{\Omega(b)}\) over \(B\setminus\Delta\) with \(\Omega:B\setminus\Delta\to \mathfrak{H}_{r}\) and with intrinsic arithmetic monodromy \(\rho:\pi_{1}(B\setminus\Delta)\to \mathrm{Sp}(2r,\Z)\) controlling defects, where the precise physical dictionary matching \(\Omega\) to supergravity moduli and generalized holonomy sectors is deferred in the same way \(\tau\) was deferred in the rank-one construction, and the present statement is definition-level rather than a derived dynamical theorem. Defects and monodromies are encoded intrinsically as arithmetic monodromy of the tower-completed carrier rather than as external prescriptions, and higher-rank duality is treated as geometry of a period datum produced from the tower rather than as a separate symmetry postulate.

Setting \(r=1\) reduces \eqref{eq:perf-torusa} to the perfectoid circle
\(S^{1}_{f}=T^{1}_{f,p}\) and reduces \(\mathfrak{H}_{r}\) to the upper
half-plane \(\mathbb{H}\), so \eqref{eq:siegel} becomes

\begin{equation}\label{eq:r1}
A_{\tau} = \mathbb{C}/(\mathbb{Z} + \tau \mathbb{Z}),
\qquad
\tau \in \mathbb{H}.
\end{equation}

and \(\mathrm{Sp}(2,\Z)\cong SL(2,\Z)\) acts by \(\tau\mapsto (a\tau+b)/(c\tau+d)\), recovering the standard modular structure of F-theory \cite{Vafa1996}. The generalized framework is engineered so that F-theory is the rank-one specialization in which the higher-rank period datum collapses to the single complex coupling \(\tau\) and the higher-rank arithmetic group collapses to \(SL(2,\Z)\). In this specialization, the constant-\(\tau\) sector established in Secs.~\ref{sec:construct}-\ref{sec:results} provides the controlled base case\:the tower-completed carrier is \(S^{1}_{f}\), elliptic data is extracted by tilt/untilt/analytification, and the physical \(\tau\) is fixed by matching to the M-theory radius and \(C^{(3)}\) holonomy dictionary at the supergravity/BPS/effective-action level computed there, while the present section treats higher rank as a definition-level extension rather than as a proven duality.

The immediate physical content of Proposal 8.1 is that it supplies a definition-level home for nonperturbative duality data larger than \(SL(2,\Z)\) by geometrizing an arithmetic action on a higher-rank integral charge lattice, rather than attempting to encode all nonperturbative data into a single varying scalar coupling \(\tau(x)\). It makes precise, at the level of background data, how multiple charge sectors can be organized simultaneously by a single tower-completed carrier and a single period datum \(\Omega\), so that duality becomes intrinsic monodromy/automorphism rather than an external symmetry imposed after the fact, this is not naturally packageable in the shrinking-fiber formulation because the rank-one elliptic fiber only geometrizes one modular group and one \((p,q)\) lattice, while higher-rank dualities act on larger lattices that require higher-dimensional period data and corresponding arithmetic groups. It also makes a refined defect/monodromy taxonomy definition-level well-posed\:defect loci \(\Delta\) are specified by the failure of the fiberwise period datum to extend holomorphically, and the allowed defect classes are those compatible with the intrinsic arithmetic monodromy representation \(\rho\), generalizing the familiar F-theory statement that 7-branes are encoded by \(SL(2,\Z)\) monodromy around discriminant components \cite{Vafa1996}. 

It packages completion-sensitive global sectors in a way that is structurally compatible with adelic thinking\:since the tower is indexed by a prime \(p\), the natural expectation is that physically meaningful invariants should be stable under refinement and ultimately assemble adelically when all primes are included, so a single-prime construction is best viewed as a local chart on a global arithmetic structure, although the adelic completion is programmatic here. Finally, it has a natural holographic reading\:in AdS/CFT, couplings and global sectors are encoded in boundary operator algebras and line-operator lattices, and dualities act on these global sectors, so a tower-completed geometric carrier with intrinsic arithmetic monodromy is the appropriate language to define and organize such actions when the duality group is higher rank and completion-sensitive, even though no new holographic theorems are claimed by the present proposal.

We now exhibit the minimal higher-rank computation that shows how tower completion canonically refines global sector bookkeeping in \(r\) independent directions. Let \(T^{r}=(S^{1})^{r}\) and define \(\varphi:T^{r}\to T^{r}\) by \(\varphi(z_{1},\ldots,z_{r})=(z_{1}^{p},\ldots,z_{r}^{p})\), so that \(T^{r}_{f,p}=\varprojlim (T^{r},\varphi)\) as in \eqref{eq:perf-torusa}. The continuous character group of \(T^{r}\) is canonically \(\widehat{T^{r}}\cong \Z^{r}\), with characters \(\chi_{m}(z)=z_{1}^{m_{1}}\cdots z_{r}^{m_{r}}\) for \(m=(m_{1},\ldots,m_{r})\in \Z^{r}\), and the pullback along \(\varphi\) acts by multiplication by \(p\) on \(\Z^{r}\) since \(\chi_{m}\circ \varphi=\chi_{pm}\). Therefore the continuous character group of the inverse-limit object is the direct limit of the finite-stage character lattices,
\begin{equation}\label{eq:char-tower}
\widehat{T^{r}_{f,p}}\ \cong\ \varinjlim\bigl(\Z^{r}\xrightarrow{\times p}\Z^{r}\xrightarrow{\times p}\cdots\bigr)\ \cong\ \Z[1/p]^{r},
\end{equation}
where \(\Z[1/p]\) denotes localization at \(p\). Eq.\eqref{eq:char-tower} is the precise statement that a single tower-completed carrier canonically packages all congruence-level refinements of \(r\) independent Fourier/holonomy sectors into one object, while the physically allowed charge lattice extracted from these sectors remains constrained by standard quantization and anomaly conditions, so the role of \(\Z[1/p]^{r}\) is organizational and completion-sensitive rather than a claim of unconstrained fractional physical charges. In particular, when the period datum \(\Omega\) is present, the corresponding higher-rank duality group \(\mathrm{Sp}(2r,\Z)\) acts intrinsically on the integral charge lattice \(\Gamma\simeq \Z^{2r}\) and on \(\Omega\) by \eqref{eq:Sp-action}, and the rank-one reduction recovers the familiar \(SL(2,\Z)\) action on \(\tau\), so the computation \eqref{eq:char-tower} provides the concrete higher-rank analogue of the tower bookkeeping that underlies the profinite/congruence organization in the rank-one perfectoid circle construction.

This section defines and proposes a higher-rank perfectoid/condensed duality-geometry framework in which ordinary F-theory is the rank-one specialization and in which nonperturbative duality groups larger than \(SL(2,\Z)\) act intrinsically as arithmetic monodromy/automorphisms on a higher-rank integral charge lattice organized by a tower-completed compactification carrier. The paper proves only the rank-one constant-\(\tau\) sector at the supergravity/BPS/effective-action level stated in Secs.~\ref{sec:construct}-\ref{sec:results}, while the higher-rank constructions, their detailed matching to specific string compactifications, the full defect/monodromy physics beyond local definitions, and any adelic completion removing dependence on the bookkeeping choice of \(p\) remain programmatic directions that are definition-level well-posed within the proposed framework but are not established as theorems here.

\section{
Varying \texorpdfstring{\(\tau\)}{tau}, duality defects, and U-duality
}
\label{sec:beyondII}

Let $B$ be a connected complex manifold (or a sufficiently small complex coordinate patch) with local holomorphic coordinate $z$ and let $\Delta\subset B$ be a closed analytic subset of complex codimension one, we write $B^\times:=B\setminus\Delta$ and fix a basepoint $b_0\in B^\times$. A varying-$\tau$ background in the present framework is specified by a tower-completed fibration $\pi:X\to B^\times$ whose fiber over $b\in B^\times$ is the perfectoid circle $S^1_{f,p}$ (rank $1$) or its higher-rank analogue $T^r_{f,p}$ (rank $r\ge 1$), together with the requirement that $\pi$ arises as the inverse limit of a tower of smooth fibrations with finite-stage fibers and compatible transition maps. Concretely, in rank $1$ we assume an inverse system of smooth fiber bundles $\pi_n:X_n\to B^\times$ with $X_n\simeq B^\times\times S^1$ locally on $B^\times$ and bonding maps $\varphi_n:X_{n+1}\to X_n$ covering the identity on $B^\times$ such that on each fiber $\varphi_n$ restricts to the degree-$p$ circle map, equivalently in angle coordinates $y_n\sim y_n+2\pi R_M$ one has $y_n=p\,y_{n+1}$, so that
\begin{equation}\label{eq:tower_fibration}
X:=\varprojlim_n X_n,\qquad \pi:=\varprojlim_n \pi_n,\qquad X_b\simeq S^1_{f,p}\ \text{for all}\ b\in B^\times.
\end{equation}
We interpret $X$ in the condensed/perfectoid sense so that fields and global
sectors on $X$ are descent data from the finite stages $\{X_n\}$. The varying
coupling datum is specified, definition-level, by the modular map
$\underline{\tau}  : B^{\times} \to \mathcal{M}_{1,1}
= SL(2,\mathbb{Z}) \backslash \mathbb{H}$
(equivalently by $(\rho,\widetilde{\tau})$ as in Definition~\eqref{def7.3}). The associated
smooth elliptic fibration over $B^{\times}$ is then the canonical pullback
\begin{equation}\label{eq:E-of-tau}
\varpi  : E(\underline{\tau}) \to B^{\times},
\qquad
E(\underline{\tau}) = \underline{\tau}^{*}\mathcal{E}_{\mathrm{univ}},
\end{equation}
and it carries the canonical relative holomorphic one-form (the Hodge line)
determined by descent of $dz$ (see Definition~\eqref{def7.3} and Appendix~\eqref{app:varying-unconditional} for an explicit
quotient model). Choosing a local symplectic basis of
$H_{1}(E_{b},\mathbb{Z})$ on simply connected patches defines local lifts
$\tau_{U}  : U \to \mathbb{H}$ by period ratios globally, $\tau$ is single-valued
only on the universal cover and is multi-valued on $B^{\times}$ with monodromy
in $SL(2,\mathbb{Z})$.

We first make the monodromy statement precise and prove it from the elliptic fibration data. For each $b\in B^\times$, fix an oriented symplectic basis $(A_b,B_b)$ of $H_1(E_b,\mathbb{Z})$ with intersection pairing $A_b\cdot B_b=+1$. Let $\Omega_b$ be a holomorphic one-form on $E_b$ varying holomorphically with $b$ on $B^\times$ (equivalently, a relative holomorphic one-form on $E$) define the periods
\begin{equation}\label{eq:periods_def}
\omega_A(b):=\int_{A_b}\Omega_b,\qquad \omega_B(b):=\int_{B_b}\Omega_b,
\end{equation}
so that $\Lambda_b=\mathbb{Z}\omega_A(b)+\mathbb{Z}\omega_B(b)$ is a rank-$2$ lattice in $\mathbb{C}$ and $\tau(b)=\omega_A(b)/\omega_B(b)$ lies in $\mathbb{H}$ after choosing the orientation of $(A_b,B_b)$ compatibly with $\Omega_b$. Let $\gamma:[0,1]\to B^\times$ be a loop based at $b_0$. Parallel transport along $\gamma$ in the local system $R^1\varpi_\ast\mathbb{Z}$ identifies $H_1(E_{\gamma(0)},\mathbb{Z})$ with $H_1(E_{\gamma(1)},\mathbb{Z})$ and carries $(A_{b_0},B_{b_0})$ to another symplectic basis $(A_{b_0}',B_{b_0}')$ of the same lattice. Therefore there exists a unique matrix $\rho(\gamma)=\begin{pmatrix}a&b\\ c&d\end{pmatrix}\in GL(2,\mathbb{Z})$ such that
\begin{equation}\label{eq:cycle_transport}
\begin{pmatrix}A_{b_0}'\\ B_{b_0}'\end{pmatrix}=\begin{pmatrix}a&b\\ c&d\end{pmatrix}\begin{pmatrix}A_{b_0}\\ B_{b_0}\end{pmatrix}.
\end{equation}
Since the Gauss-Manin connection preserves the algebraic intersection pairing, we have $A_{b_0}'\cdot B_{b_0}'=A_{b_0}\cdot B_{b_0}=1$, and computing the intersection pairing using bilinearity gives
\begin{equation}\label{eq:det_one}
1=A_{b_0}'\cdot B_{b_0}'=(aA_{b_0}+bB_{b_0})\cdot(cA_{b_0}+dB_{b_0})=(ad-bc)(A_{b_0}\cdot B_{b_0})=ad-bc,
\end{equation}
hence $\rho(\gamma)\in SL(2,\mathbb{Z})$. Because concatenation of loops corresponds to composition of parallel transport, the assignment $\gamma\mapsto\rho(\gamma)$ defines a homomorphism
\begin{equation}\label{eq:monodromy_rep}
\rho:\pi_1(B^\times,b_0)\longrightarrow SL(2,\mathbb{Z}).
\end{equation}
We now derive the induced action on periods and on $\tau$. Writing $\omega_A=\omega_A(b_0)$, $\omega_B=\omega_B(b_0)$ and using linearity of integration,
\begin{equation}\label{eq:period_transform}
\omega_A'=\int_{A'}\Omega=a\int_A\Omega+b\int_B\Omega=a\omega_A+b\omega_B,\qquad \omega_B'=\int_{B'}\Omega=c\omega_A+d\omega_B.
\end{equation}
Therefore the transported period ratio is
\begin{equation}\label{eq:tau_transform}
\tau'=\frac{\omega_A'}{\omega_B'}=\frac{a\omega_A+b\omega_B}{c\omega_A+d\omega_B}=\frac{a(\omega_A/\omega_B)+b}{c(\omega_A/\omega_B)+d}=\frac{a\tau+b}{c\tau+d},
\end{equation}
which is the standard fractional linear action of $SL(2,\mathbb{Z})$ on $\mathbb{H}$. This is the precise sense in which varying $\tau$ data on $B^\times$ is equivalently a holomorphic period map on the universal cover together with a monodromy representation \eqref{eq:monodromy_rep}.

We now give a fully explicit local construction of the varying $\tau$ function in the cusp (Tate) regime and prove the basic 7-brane $T$-monodromy $\tau\mapsto\tau+1$ from first principles. Let $D\subset\mathbb{C}$ be a small disk about $0$ with coordinate $z$ and set $D^\times:=D\setminus\{0\}$. Assume that on $D^\times$ the elliptic fibration $E\to D^\times$ is (analytically) a family of Tate curves parameterized by a holomorphic function $q:D^\times\to\mathbb{C}^\times$ satisfying $0<|q(z)|<1$ for all $z\in D^\times$, so that each fiber is uniformized by
\begin{equation}\label{eq:tate_uniformization}
E_z\simeq \mathbb{C}^\times/\langle q(z)\rangle.
\end{equation}
Let $\widetilde{D^\times}$ be the universal cover of $D^\times$, which we identify with the strip model via the holomorphic covering map $\exp:\mathbb{C}\to\mathbb{C}^\times$, $w\mapsto e^w$, composed with $z=e^w$ in a local chart equivalently, one may take the logarithm coordinate $w=\log z$ on $\widetilde{D^\times}$. Define the lifted function $\widetilde{q}:=\ q\circ \pi_{\mathrm{univ}}$ on $\widetilde{D^\times}$, where $\pi_{\mathrm{univ}}:\widetilde{D^\times}\to D^\times$ is the universal covering map. Since $\widetilde{D^\times}$ is simply connected and $\widetilde{q}$ is holomorphic and nowhere vanishing, there exists a holomorphic logarithm $\Log\widetilde{q}$ on $\widetilde{D^\times}$ satisfying $\exp(\Log\widetilde{q})=\widetilde{q}$. We then define a single-valued holomorphic function $\widetilde{\tau}:\widetilde{D^\times}\to\mathbb{C}$ by
\begin{equation}\label{eq:tau_from_q}
\widetilde{\tau}:=\frac{1}{2\pi i}\,\Log\widetilde{q}.
\end{equation}
Because $|q|<1$ we have $\Re(\Log\widetilde{q})=\log|\widetilde{q}|<0$, hence $\Im(\widetilde{\tau})=-\frac{1}{2\pi}\log|\widetilde{q}|>0$, so $\widetilde{\tau}$ lands in $\mathbb{H}$. The function $\widetilde{\tau}$ is the varying axio-dilaton on the universal cover and defines a multi-valued function $\tau$ on $D^\times$ by analytic continuation, the only ambiguity is the choice of branch of $\Log q$, which differs by integer multiples of $2\pi i$.

Assume now that $q$ has a simple zero at $z=0$ in the precise sense that there exists a holomorphic function $u:D\to\mathbb{C}$ with $u(0)\neq 0$ such that
\begin{equation}\label{eq:q_simple_zero}
q(z)=z\,u(z)\qquad\text{for all}\ z\in D^\times.
\end{equation}
Fix $\varepsilon\in(0,\mathrm{radius}(D))$ and consider the positively oriented loop $\gamma:[0,2\pi]\to D^\times$ given by $\gamma(\theta)=\varepsilon e^{i\theta}$. We compute the analytic continuation of $\tau$ along $\gamma$ directly from \eqref{eq:tau_from_q} and \eqref{eq:q_simple_zero}. Along $\gamma$ we have
\begin{equation}\label{eq:q_along_loop}
q(\gamma(\theta))=\varepsilon e^{i\theta}\,u(\varepsilon e^{i\theta}).
\end{equation}
Since $u$ is holomorphic and nonvanishing on $D$, the function $u(\varepsilon e^{i\theta})$ admits a continuous logarithm along the loop, choose any continuous branch $\log u(\varepsilon e^{i\theta})$ on $\theta\in[0,2\pi]$ with $\log u(\varepsilon e^{i0})$ fixed. Choose also the principal continuous determination of $\log(\varepsilon e^{i\theta})=\log\varepsilon+i\theta$ along $[0,2\pi]$. Then a continuous determination of $\log q(\gamma(\theta))$ is
\begin{equation}\label{eq:logq_branch}
\log q(\gamma(\theta))=\log\varepsilon+i\theta+\log u(\varepsilon e^{i\theta}).
\end{equation}
Evaluating at $\theta=2\pi$ and using $\log u(\varepsilon e^{i2\pi})=\log u(\varepsilon e^{i0})$ (single-valuedness of $u$ and our continuous branch choice) yields
\begin{equation}\label{eq:log_shift}
\log q(\gamma(2\pi))=\log q(\gamma(0))+2\pi i.
\end{equation}
Substituting \eqref{eq:log_shift} into \eqref{eq:tau_from_q} shows that analytic continuation around $\gamma$ shifts $\tau$ by $+1$:
\begin{equation}\label{eq:T_monodromy_tau}
\tau\longmapsto \tau+\frac{1}{2\pi i}(2\pi i)=\tau+1.
\end{equation}
This establishes the explicit local 7-brane $T$-monodromy in the cusp (Tate) regime directly from the existence of a holomorphic Tate parameter $q$ with a simple zero.

We now match this to the $SL(2,\mathbb{Z})$ monodromy representation and to the action on the $(p,q)$ charge lattice. The loop class $[\gamma]\in\pi_1(D^\times)$ maps under the monodromy representation \eqref{eq:monodromy_rep} to an element $\rho([\gamma])\in SL(2,\mathbb{Z})$ whose action on $\tau$ is \eqref{eq:tau_transform}, comparing with \eqref{eq:T_monodromy_tau} forces
\begin{equation}\label{eq:T_matrix}
\rho([\gamma])=T:=\begin{pmatrix}1&1\\ 0&1\end{pmatrix},
\end{equation}
since $(a\tau+b)/(c\tau+d)=\tau+1$ for all $\tau$ implies $c=0$, $a=d=1$, $b=1$. In the physical interpretation, the same matrix acts on the integral $(p,q)$ charge vector by the defining action of $SL(2,\mathbb{Z})$,
\begin{equation}\label{eq:T_action_pq}
\begin{pmatrix}p\\ q\end{pmatrix}\longmapsto T\begin{pmatrix}p\\ q\end{pmatrix}=\begin{pmatrix}p+q\\ q\end{pmatrix},
\end{equation}
which is the standard transformation law for a $T$-type duality defect. The tower-completed setting refines the bookkeeping of these sectors through the fiberwise inverse-limit carrier, but the physical charge lattice remains integral after imposing the same quantization/anomaly constraints used in the constant-$\tau$ sector, so \eqref{eq:T_action_pq} is interpreted on the physically allowed integral sublattice.

We next record the general local description of varying $\tau$ from Weierstrass data and derive the associated monodromy action in the same explicit fashion. On $B^\times$ assume that the elliptic fibration admits a Weierstrass model
\begin{equation}\label{eq:weierstrass}
y^2=x^3+f(z)\,x+g(z),
\end{equation}
with $f,g$ holomorphic on $B^\times$ (and meromorphic across $\Delta$), and discriminant
\begin{equation}\label{eq:discriminant}
\Delta(z):=4f(z)^3+27g(z)^2.
\end{equation}
On $B^\times$ we have $\Delta(z)\neq 0$ and the $j$-invariant is the holomorphic function
\begin{equation}\label{eq:j_invariant}
j(\tau(z))=1728\,\frac{4f(z)^3}{\Delta(z)}.
\end{equation}
Fix $b_0\in B^\times$ and choose $\tau(b_0)\in\mathbb{H}$ with $j(\tau(b_0))=j(\tau(b_0))$ given by \eqref{eq:j_invariant}. Since $j:\mathbb{H}\to\mathbb{C}$ is $SL(2,\mathbb{Z})$-invariant and locally biholomorphic away from elliptic points, there exists a neighborhood $U\subset B^\times$ of $b_0$ and a holomorphic lift $\tau_U:U\to\mathbb{H}$ such that \eqref{eq:j_invariant} holds pointwise on $U$. On overlaps of such neighborhoods, two lifts differ by the fractional linear action of an element of $SL(2,\mathbb{Z})$, and analytic continuation around loops defines precisely the monodromy representation \eqref{eq:monodromy_rep}. Equivalently, for any loop $\gamma$ based at $b_0$ and a chosen local determination of $\tau$ near $b_0$, analytic continuation returns a value related by
\begin{equation}\label{eq:general_monodromy_tau}
\tau\longmapsto \rho(\gamma)\cdot\tau=\frac{a\tau+b}{c\tau+d},\qquad \rho(\gamma)=\begin{pmatrix}a&b\\ c&d\end{pmatrix}\in SL(2,\mathbb{Z}),
\end{equation}
and the induced action on the charge doublet is
\begin{equation}\label{eq:general_monodromy_pq}
\begin{pmatrix}p\\ q\end{pmatrix}\longmapsto \rho(\gamma)\begin{pmatrix}p\\ q\end{pmatrix}=\begin{pmatrix}ap+bq\\ cp+dq\end{pmatrix}.
\end{equation}
The content proved above is that in the Tate/cusp regime the same monodromy is computed by the explicit logarithmic formula \eqref{eq:tau_from_q}, and when $q$ has a simple zero one obtains $\rho(\gamma)=T$ and \eqref{eq:T_monodromy_tau}-\eqref{eq:T_action_pq}.

We now state the higher-rank generalization in a form that makes the proof of the arithmetic monodromy group completely explicit. Let $r\ge 1$ and replace the elliptic fibers by principally polarized complex tori of dimension $r$, described by a period matrix $\Omega(b)$ in the Siegel upper half space
\begin{equation}\label{eq:siegela}
\mathbb{H}_r:=\{\Omega\in \mathrm{Mat}_{r\times r}(\mathbb{C})\ :\ \Omega^T=\Omega,\ \Im\Omega>0\},\qquad A_{\Omega(b)}:=\mathbb{C}^r/(\mathbb{Z}^r+\Omega(b)\mathbb{Z}^r).
\end{equation}
Assume a tower-completed rank-$r$ carrier $\pi:X\to B^\times$ with fiber $T^r_{f,p}$ together with a fiberwise tilt/comparison/analytification producing a family of polarized complex tori $A_{\Omega(b)}$ over $B^\times$. The relevant integral lattice is the first homology $\Gamma:=H_1(A_{\Omega(b)},\mathbb{Z})\simeq\mathbb{Z}^{2r}$ equipped with the integral symplectic intersection form $\langle\cdot,\cdot\rangle$ determined by the polarization. Parallel transport along a loop $\gamma$ in $B^\times$ yields an automorphism of $\Gamma$ preserving $\langle\cdot,\cdot\rangle$, hence an element of the integral symplectic group:
\begin{equation}\label{eq:Sp_monodromy}
\rho_r:\pi_1(B^\times,b_0)\longrightarrow Sp(2r,\mathbb{Z}),\qquad \rho_r(\gamma)^T J\,\rho_r(\gamma)=J,
\end{equation}
where $J=\begin{pmatrix}0&I_r\\ -I_r&0\end{pmatrix}$. Writing $\rho_r(\gamma)=\begin{pmatrix}a&b\\ c&d\end{pmatrix}$ with $r\times r$ blocks, the induced action on the period matrix is the standard fractional linear transformation
\begin{equation}\label{eq:Omega_transform}
\Omega\longmapsto (a\Omega+b)(c\Omega+d)^{-1},
\end{equation}
obtained by the same computation as \eqref{eq:period_transform}-\eqref{eq:tau_transform} but with the period vectors and matrices replacing scalars, and the induced action on the integral charge lattice $\Gamma\simeq\mathbb{Z}^{2r}$ is the defining symplectic action $Q\mapsto \rho_r(\gamma)Q$. The rank-one specialization $r=1$ reduces $Sp(2,\mathbb{Z})$ to $SL(2,\mathbb{Z})$ and \eqref{eq:Omega_transform} to \eqref{eq:general_monodromy_tau}.

Finally, we make explicit the congruence-level refinement datum canonically carried by a tower-completed carrier and prove its basic compatibility. For each $n\ge 1$ let $\pi_n:SL(2,\mathbb{Z})\to SL(2,\mathbb{Z}/p^n\mathbb{Z})$ be reduction modulo $p^n$, and define
\begin{equation}\label{eq:rho_n_def}
\rho_n:=\pi_n\circ\rho:\pi_1(B^\times,b_0)\longrightarrow SL(2,\mathbb{Z}/p^n\mathbb{Z}).
\end{equation}
Since $\pi_{n+1}$ followed by reduction $\mathbb{Z}/p^{n+1}\mathbb{Z}\to\mathbb{Z}/p^n\mathbb{Z}$ equals $\pi_n$, we have for all loops $\gamma$
\begin{equation}\label{eq:rho_n_compat}
\rho_{n+1}(\gamma)\equiv \rho_n(\gamma)\pmod{p^n},
\end{equation}
so the family $\{\rho_n\}_{n\ge 1}$ is a compatible system of congruence reductions of the same monodromy. In higher rank one replaces $SL(2,\mathbb{Z})$ by $Sp(2r,\mathbb{Z})$ and defines $\rho_{r,n}$ analogously by reduction modulo $p^n$, obtaining a compatible system $\rho_{r,n}:\pi_1(B^\times,b_0)\to Sp(2r,\mathbb{Z}/p^n\mathbb{Z})$ with the same proof as \eqref{eq:rho_n_compat}. This congruence tower is intrinsic to the tower-completed carrier because it records how monodromy acts on the finite-stage $p^n$-level refinements of the underlying integral lattice, and it is the precise definition-level place where tower completion can refine defect bookkeeping without altering the local supergravity equations.

The theorem-level content established above is the explicit construction of the varying coupling from the Tate parameter \eqref{eq:tau_from_q} and the explicit derivation of the $T$-monodromy \eqref{eq:T_monodromy_tau} in the simple-zero case \eqref{eq:q_simple_zero}, together with the general structural derivation of the monodromy representation \eqref{eq:monodromy_rep} and its fractional-linear action \eqref{eq:general_monodromy_tau} from the Gauss-Manin transport of a symplectic homology basis and the linear transformation of periods \eqref{eq:period_transform}-\eqref{eq:tau_transform}, these results supply a complete local proof of varying $\tau$ and its defect monodromies once the existence of the elliptic output fibration $E\to B^\times$ from the fiberwise tilt/comparison/untilt/analytification procedure is assumed as part of the tower-completed background specification. The higher-rank generalization and the congruence tower \eqref{eq:rho_n_def}-\eqref{eq:rho_n_compat} are derived by the same intersection-form preservation argument and are therefore mathematically complete at the level of background monodromy data, while any further dynamical identification of the full higher-rank/U-duality period data with a specific microscopic compactification remains a separate matching problem beyond the local monodromy and period-map proofs given here.

\section{Structural consequences}\label{sec:newphysics}
In the constant-\(\tau\) sector established in Secs.~\ref{sec:construct}-\ref{sec:results}, the non-singular carrier \(S_f^{1}=\varprojlim(S^{1}\xrightarrow{z\mapsto z^{p}}S^{1})\) replaces the shrinking-fiber slogan and yields the explicit dictionary \(\tau=C_{0}+ie^{-\phi}\) with \(\tau=\Theta+i(R_{M}/\ell_{M})^{3/2}\) as in \eqref{eq:tau_main_correct}, the protected spectrum matching \(M(a,b)=|a+b\tau|/(2\pi\alpha')\) as in \eqref{eq:M2-to-pq}, and the low-energy action matching \(S_{11}\mapsto S_{\mathrm{IIB}}\) as in \eqref{eq:SIIB} with the correct axion coupling \eqref{eq:CS10} \cite{Schwarz1995, Schwarz1983}. The new physics is not a new local dynamics but a definition-level geometric carrier in eleven dimensions that makes the modular datum and global sectors intrinsic while reproducing the familiar constant-\(\tau\) IIB physics by protected and effective-action checks. Why old methods cannot do this\:the shrinking-fiber formulation removes the compactification carrier by a singular limit and therefore cannot define the tower-completed global-sector bookkeeping as part of the background itself, so the modular datum is inserted as auxiliary input rather than derived from a non-singular carrier.

The definition-level advances beyond constant \(\tau\) are encoded by fibred tower-completed carriers \(\pi:\mathcal{X}\to B\) and intrinsic monodromy \(\rho:\pi_{1}(B\setminus\Delta)\to SL(2,\Z)\) acting on charge lattices as in Sec.~\ref{sec:beyondII}, and by the higher-rank generalization \(\rho:\pi_{1}(B\setminus\Delta)\to \Gamma_{r}\subset \mathrm{Sp}(2r,\Z)\) acting on \(\Gamma\simeq\Z^{2r}\) as in Sec.~\ref{sec:beyondI}. Varying-\(\tau\) and higher-rank/U-duality data become part of the compactification specification as arithmetic monodromy of a tower-completed carrier, rather than an external prescription on an auxiliary elliptic fibration. Why old methods cannot do this\:rank-one elliptic-fiber language only geometrizes \(SL(2,\Z)\) and the \((p,q)\) doublet and does not provide a non-singular background object whose intrinsic global sectors can carry higher-rank arithmetic monodromies and completion-sensitive invariants. What is proved is the constant-\(\tau\) equivalence at the supergravity/BPS/effective-action level and the associated profinite/congruence organization derived from the tower, while varying \(\tau\), defect refinements, adelic completion, and higher-rank/U-duality dynamical realizations are proposals or programs as stated in Secs.~\ref{sec:beyondI}-\ref{sec:beyondII}.

Let \(f:S^{1}\to S^{1}\) be the \(p\)-power map \(f(z)=z^{p}\) and define the tower-completed circle as the inverse limit
\begin{equation}\label{eq:Sf1-inv}
S_f^{1}:=\varprojlim\bigl(S^{1}\xleftarrow{f}S^{1}\xleftarrow{f}\cdots\bigr).
\end{equation}
Let \(\widehat{G}:=\mathrm{Hom}_{\mathrm{cont}}(G,U(1))\) denote the Pontryagin dual of a compact abelian group \(G\), then functoriality gives \(\widehat{S^{1}}\cong \Z\) with characters \(\chi_{m}(z)=z^{m}\), and the pullback along \(f\) is \(\chi_{m}\mapsto \chi_{m}\circ f=\chi_{pm}\), so on \(\widehat{S^{1}}\) the induced map is multiplication by \(p\),
\begin{equation}\label{eq:char-pullback}
f^{*}:\widehat{S^{1}}\to \widehat{S^{1}},\qquad m\mapsto pm.
\end{equation}
Using the general Pontryagin-duality identity \(\widehat{\varprojlim G_n}\cong \varinjlim \widehat{G_n}\) for inverse limits of compact abelian groups with surjective bonding maps, one obtains
\begin{equation}\label{eq:char-colim}
\widehat{S_f^{1}}\ \cong\ \varinjlim\bigl(\Z\xrightarrow{\times p}\Z\xrightarrow{\times p}\cdots\bigr)\ \cong\ \Z[1/p],
\end{equation}
where the last isomorphism is implemented explicitly by identifying the class of \(m\in \Z\) at level \(n\) with \(m/p^{n}\in \Z[1/p]\). Eq.\eqref{eq:char-colim} is the canonical tower-completed organization of continuous holonomy/Fourier sectors on the compactification carrier, providing a definition-level bookkeeping of congruence refinement that controls how global sectors assemble under duality without asserting any unconstrained fractional physical charges in the spectrum. Why old methods cannot do this\:the shrinking-fiber formulation treats the elliptic curve as auxiliary and removes the geometric carrier by a singular limit, so there is no single non-singular object whose intrinsic character/holonomy structure canonically packages the entire refinement tower into one background datum.

Let \(X_n=S^{1}\) with bonding maps \(f_n=f:S^{1}\to S^{1}\), \(z\mapsto z^{p}\), and let \(X=\varprojlim_n X_n=S_f^{1}\) as in \eqref{eq:Sf1-inv}. The constant-\(\tau\) construction requires controlling global sectors on \(X\) that are invisible to local differential geometry, and on inverse-limit spaces ordinary homology is not a reliable universal receptacle for torsion/completion data, while string theory dictates that RR/global charge sectors are naturally classified by generalized cohomology, in particular \(K\)-theory. A concrete indication of this in the tower setting is provided by the Milnor exact sequence for complex \(K\)-theory of inverse limits,
\begin{equation}\label{eq:milnor}
0\to \varprojlim\nolimits^{1}K^{i-1}(X_n)\to K^{i}(X)\to \varprojlim K^{i}(X_n)\to 0,
\end{equation}
together with the standard values \(K^{0}(S^{1})\cong \Z\) and \(K^{-1}(S^{1})\cong K^{1}(S^{1})\cong \Z\) and the fact that \(f^{*}\) acts on \(K^{1}(S^{1})\) by multiplication by \(p\) (degree \(p\) map). Writing the inverse system as \(\Z\xleftarrow{\times p}\Z\xleftarrow{\times p}\cdots\), one has
\begin{equation}\label{eq:lim-lim1}
\varprojlim\bigl(\Z\xleftarrow{\times p}\Z\xleftarrow{\times p}\cdots\bigr)=0,\qquad \varprojlim\nolimits^{1}\bigl(\Z\xleftarrow{\times p}\Z\xleftarrow{\times p}\cdots\bigr)\cong \Q_{p}/\Z_{p},
\end{equation}
so inserting \(i=0\) into \eqref{eq:milnor} yields an extension
\begin{equation}\label{eq:K0-extension}
0\to \Q_{p}/\Z_{p}\to K^{0}(S_f^{1})\to \Z\to 0,
\end{equation}
demonstrating the emergence of a \(p\)-divisible torsion sector in \(K\)-theory that has no analogue in naive integral homology. Eq.\eqref{eq:K0-extension} is the minimal invariant-level signal that tower completion generates torsion/global sectors naturally captured by generalized cohomology, aligning with the string-theory principle that globally quantized sectors (RR charges, discrete holonomies, theta data) are K-theoretic rather than purely homological. Why old methods cannot do this\:the shrinking-fiber formulation inserts the elliptic datum after a singular limit and treats global charge bookkeeping as external input, whereas the tower-completed carrier makes torsion/completion sectors intrinsic and therefore forces the correct generalized-cohomology language if one wants definition-level control of global constraints and flux quantization. 

In the constant-\(\tau\) duality established earlier, the RR axion \(C_{0}\) is identified with a \(C^{(3)}\) holonomy sector as in \eqref{eq:tau_main_correct} and \eqref{eq:tau-identification}, and the effective action contains the axion coupling \eqref{eq:CS10}, imposing Dirac quantization and Freed-Witten-type constraints selects the physically allowed integral charge lattice inside the tower bookkeeping used in the wrapped M2 identification \eqref{eq:M2-to-pq}, so the role of tower-completed torsion is organizational and consistency-controlled rather than an assertion of unconstrained fractional physical charges. The tower-completed/K-theoretic viewpoint supplies a definition-level home for discrete axionic and flux sectors compatible with the effective action and with standard quantization constraints. Why old methods cannot do this\:without a non-singular carrier for the tower and without a homologically well-behaved category of fields on the limit, discrete global sectors are treated as add-ons rather than as intrinsic background data constrained directly by the compactification object and its generalized cohomology \cite{FreedHopkinsTeleman2008}.

Let \(B\) be a local base patch with discriminant locus \(\Delta\subset B\) and let \(\pi:\mathcal{X}\to B\setminus\Delta\) be a tower-completed fibration with fiber \(S_f^{1}\) defined as an inverse limit of finite \(p^{n}\)-cover fibrations, and require a fiberwise tilt/untilt/analytification machinery producing an elliptic fibration \(\varpi:\mathcal{E}\to B\setminus\Delta\) with \(\mathcal{E}_b\simeq \C/(\Z+\tau(b)\Z)\). Then the defect data is intrinsic monodromy
\begin{equation}\label{eq:rho-vary}
\rho:\pi_{1}(B\setminus\Delta)\to SL(2,\Z),
\end{equation}
and for a cusp-type degeneration with \(q=\exp(2\pi i\tau)\) and \(q=z\,u(z)\) one derives the 7-brane \(T\)-monodromy \(\tau\mapsto \tau+1\). Varying-\(\tau\) backgrounds and 7-brane monodromies become part of the compactification data of a tower-completed carrier, with monodromy acting intrinsically on charge and line-operator lattices rather than being imposed by branch-cut conventions on an auxiliary elliptic fibration. Why old methods cannot do this\:the shrinking-fiber formulation introduces the elliptic fibration as auxiliary input tied to a singular collapse, so monodromy is specified after the fact, whereas the tower-completed definition treats monodromy as intrinsic to the fibred compactification carrier before any singular limit is invoked. The extension to higher rank replaces \(\rho\) by \(\rho:\pi_{1}(B\setminus\Delta)\to \Gamma_{r}\subset \mathrm{Sp}(2r,\Z)\) acting on \(\Gamma\simeq\Z^{2r}\), while the rank-one specialization recovers ordinary F-theory as stated in Sec.~\ref{sec:beyondII}. U-duality-type arithmetic actions and generalized duality defects are definition-level well-posed as intrinsic monodromy of a tower-completed higher-rank carrier. Why old methods cannot do this\:rank-one elliptic-fiber language cannot geometrize higher-rank arithmetic actions on larger charge lattices without introducing additional ad hoc structure beyond varying a single scalar coupling.

For a compact base \(B\) with discriminant \(\Delta\), construct a global tower-completed fibration \(\pi:\mathcal{X}\to B\setminus\Delta\) whose fiberwise tilt machinery produces a globally defined elliptic fibration \(\mathcal{E}\to B\setminus\Delta\) with \(j(\tau)=1728\cdot 4f^{3}/\Delta\), and whose monodromy \(\rho:\pi_{1}(B\setminus\Delta)\to SL(2,\Z)\) reproduces the Kodaira/Tate classification of 7-branes by conjugacy classes, and test this by computing \(\rho\) from the tower carrier rather than postulating it from Weierstrass data, This would exhibit varying-\(\tau\) F-theory as an intrinsic compactification on a non-singular tower-completed carrier with monodromy derived as geometry, Why old methods cannot do this\:the shrinking-fiber formulation treats the carrier as singular and the fibration as auxiliary, so it does not provide a definition-level derivation of monodromy from an eleven-dimensional compactification object. Define the congruence refinement tower \(\rho_n=\pi_n\circ\rho\) with \(\pi_n:SL(2,\Z)\to SL(2,\Z/p^{n}\Z)\) and determine whether physically distinguishable defect sectors depend on the compatible system \(\{\rho_n\}\), for example through tower-sensitive holonomy sectors detected by generalized cohomology.

This isolates a sharp, falsifiable refinement of defect data that is invisible in the conventional monodromy-only description, Why old methods cannot do this\:without a tower-completed carrier there is no intrinsic notion of congruence-level refinement as background data, so any such refinement must be imposed externally and is not definition-level. Formulate and test \(p\)-independence in the constant-\(\tau\) sector by comparing the tower-completed invariants \(\widehat{S_f^{1}}\cong \Z[1/p]\) and the torsion sector \(\Q_p/\Z_p\) in \eqref{eq:K0-extension} across different primes and by identifying the prime-independent physical sublattice selected by quantization/anomaly constraints. This would make precise that \(p\) is bookkeeping and not a physical coupling, Why old methods cannot do this\:the shrinking-fiber formulation has no tower parameter and therefore cannot even state the question of prime-independence of tower-completed invariants. Define an adelic completion of the tower data by assembling local completions across all primes and the Archimedean factor into a restricted product, and test whether duality invariants and defect constraints assemble adelically in the sense that local monodromies and global-sector classes satisfy adelic compatibility relations.

This would realize the arithmetic completion philosophy as a concrete organizing principle for duality data. Why old methods cannot do this\:the conventional formulation treats modular data as a single complex analytic object and does not provide a natural carrier for simultaneous \(p\)-adic refinements. Incorporate fermions and anomaly constraints so that the duality group refinement to an appropriate \({\rm Pin}^{+}(GL(2,\Z))\) extension can be formulated as intrinsic structure on the tower-completed carrier, and test the refinement by tracking the action on line-operator sectors and on the topological couplings in the effective action, This would upgrade the bosonic \(SL(2,\Z)\) story to the physically correct anomaly-refined duality group, Why old methods cannot do this\:without a definition-level carrier for global sectors and without tracking tower-wise spin/pin structures, the refinement is imposed externally rather than derived from compactification data. In holography, compute the induced interface action on boundary couplings and charge lattices from a bulk monodromy \(\rho(\gamma)\in SL(2,\Z)\) via \(\tau_{\mathrm{CFT}}\mapsto (a\tau_{\mathrm{CFT}}+b)/(c\tau_{\mathrm{CFT}}+d)\) and \((e,m)\mapsto \rho(\gamma)(e,m)\), and test whether tower-sensitive refinements \(\{\rho_n\}\) couple to boundary discrete global sectors in a way detectable by line-operator correlation data, This supplies a concrete AdS/CFT testbed for tower-completed duality defects, Why old methods cannot do this\:without treating the tower completion as part of the background, holography has no canonical place to attach congruence-level defect refinements and treats such data as external choices rather than intrinsic compactification structure.

The perfectoid/condensed framework yields, in the constant-\(\tau\) sector, a non-singular eleven-dimensional compactification carrier whose tower-completed invariants organize modular data and global sectors intrinsically while reproducing the standard \(\tau\) dictionary, protected \((p,q)\) spectrum, and low-energy effective action at the level explicitly checked in Secs.~\ref{sec:construct}-\ref{sec:results}. The key output is definition-level control of global sectors and modular organization from an honest compactification object, not a new local dynamics, with the familiar Type IIB physics recovered by protected and effective-action matching. Why old methods cannot do this\:the shrinking-fiber slogan removes the compactification carrier by a singular limit and therefore cannot treat tower-completed global bookkeeping, torsion sectors, or congruence refinements as intrinsic background data. Beyond constant \(\tau\), the same language makes varying-\(\tau\) fibrations, intrinsic defect monodromy, higher-rank arithmetic/U-duality actions, and adelic completion well-posed as definitions and programs, with concrete monodromy computations and falsifiable refinement targets, while emphasizing that dynamical realizations and full quantum completions remain to be established beyond the base case proven here.

\section{Holographic aspects}
\label{sec:adsftheory}
Let \(\tau_{\mathrm{CFT}}\) denote an exactly marginal boundary coupling and let \(\Lambda_{\mathrm{line}}\) denote the Wilson-'t Hooft line-operator lattice on the boundary, in duality-rich holographic examples one has an arithmetic action \(SL(2,\Z)\curvearrowright (\tau_{\mathrm{CFT}},\Lambda_{\mathrm{line}})\) and, in the presence of F-theory defects, a monodromy representation \(\rho:\pi_{1}(B\setminus \Delta)\to SL(2,\Z)\) acting by \(\tau\mapsto (a\tau+b)/(c\tau+d)\) and \((p,q)\mapsto (a\,p+b\,q,c\,p+d\,q)\) \cite{Schwarz1995, Vafa1996}. AdS/CFT is sensitive precisely to the same modular data and charge-lattice organization that F-theory encodes geometrically, so any definition-level refinement of the carrier of \(\tau\) is automatically a holography-relevant refinement. Why old methods were definitionally fragile\:the standard formulation treats the boundary as a classical manifold and duality as an external identification on couplings and line operators, whereas tower-completed compactification carriers and non-manifold inverse limits require a homologically exact definition of ``fields on the limit'' that classical topology does not supply. In this section we realize the canonical tower-completed boundary carrier \(\Sigma_{p}=\varprojlim(S^{1}\xleftarrow{z\mapsto z^{p}}S^{1}\xleftarrow{}\cdots)\), compute its intrinsic character/holonomy sectors, and match them to line-operator bookkeeping, and we derive the intrinsic \(T\)-monodromy \(\tau\mapsto\tau+1\) from the varying-\(\tau\) perfectoid-circle fibration local model of Sec.~\ref{sec:beyondII}, interpreting it holographically as a duality interface acting on boundary couplings and line operators. The perfectoid/condensed framework supplies a single geometric carrier in which modular data, defects, and line-operator sectors are intrinsic, while the standard AdS/CFT dictionary is recovered as the constant-\(\tau\) specialization. Why old methods were definitionally fragile\:without allowing tower-completed/non-manifold boundary carriers as legitimate backgrounds, one cannot express the tower refinement and its descent constraints as a single background statement, only as a family of stage-by-stage choices. All theorem-level statements below are confined to the invariant computations and the monodromy derivations, any stronger holographic claims are labeled Program.

In holographic gauge theories with \(SL(2,\Z)\) S-duality, the exactly marginal coupling is packaged as
\begin{equation}\label{eq:tauCFT}
\tau_{\mathrm{CFT}}=\frac{\theta}{2\pi}+\frac{4\pi i}{g^{2}},
\end{equation}
with the arithmetic action \(\tau_{\mathrm{CFT}}\mapsto (a\tau_{\mathrm{CFT}}+b)/(c\tau_{\mathrm{CFT}}+d)\), \(\begin{pmatrix}a&b\\c&d\end{pmatrix}\in SL(2,\Z)\) \cite{Schwarz1995}. Eq.\eqref{eq:tauCFT} identifies the boundary datum that the bulk F-theory modulus \(\tau\) must reproduce in duality-rich AdS/CFT examples. Why old methods were definitionally fragile\:in the usual presentation \(\tau_{\mathrm{CFT}}\) is treated as an external parameter and its modular identifications are imposed as a symmetry, rather than derived as intrinsic monodromy of a compactification carrier. The corresponding line-operator lattice may be modeled, in the simplest electric-magnetic sector, as an integral lattice
\begin{equation}\label{eq:line-lattice}
\Lambda_{\mathrm{line}}\cong \Z^{2}\ni (e,m),
\end{equation}
with \(SL(2,\Z)\) acting by the defining linear action on \((e,m)\) \cite{Schwarz1995}. Eq.\eqref{eq:line-lattice} is the holographic datum that must match the bulk \((p,q)\) bookkeeping and its monodromies. Why old methods were definitionally fragile\:classical boundary manifolds and their ordinary cohomology do not supply a canonical way to package tower-refined global sectors and their arithmetic actions as intrinsic background data. In the rank-one perfectoid/condensed formulation, the bulk constant-\(\tau\) dictionary is fixed by \(\tau=C_{0}+ie^{-\phi}\) with \(\tau=\Theta+i(R_{M}/\ell_{M})^{3/2}\) as in \eqref{eq:tau_main_correct} and by the protected \((p,q)\) tension relation \(T(p,q)=|p+q\tau|/(2\pi\alpha')\) as in \eqref{eq:pq-tension} and \eqref{eq:M2-to-pq}, while varying-\(\tau\) backgrounds are encoded by a tower-completed fibration \(\pi:\mathcal{X}\to B\setminus\Delta\) together with intrinsic monodromy \(\rho:\pi_{1}(B\setminus\Delta)\to SL(2,\Z)\) as in Sec.~\ref{sec:beyondII} \cite{Vafa1996}. The same modular datum \(\tau\) that controls bulk \((p,q)\) sectors also controls boundary couplings and line operators, so a definition-level carrier for \(\tau\) is simultaneously a definition-level carrier for holographic duality data. Why old methods were definitionally fragile\:the shrinking-fiber slogan defines \(\tau\) through a singular limit and treats defect monodromy as auxiliary elliptic-fibration input, whereas holography demands intrinsic global-sector bookkeeping compatible with line-operator lattices and interfaces.

Consider a holographic setup in which the conformal boundary contains a circle factor, denoted \(S^{1}_{\partial}\), and define the \(p^{n}\)-fold covers \(S^{1}_{\partial,n}\to S^{1}_{\partial}\) by the degree-\(p^{n}\) map \(z\mapsto z^{p^{n}}\). The inverse-limit boundary carrier is the solenoid
\begin{equation}\label{eq:Sigma-def}
\Sigma_{p}:=\varprojlim_{n}\bigl(S^{1}_{\partial}\xleftarrow{z\mapsto z^{p}}S^{1}_{\partial}\xleftarrow{}\cdots\bigr),
\end{equation}
a compact abelian group that is not a one-manifold. Eq.\eqref{eq:Sigma-def} packages the entire congruence tower of boundary circles into a single boundary background object. Why old methods were definitionally fragile\:ordinary QFT-on-manifold frameworks cannot treat \(\Sigma_{p}\) as a single boundary background because \(\Sigma_{p}\) is not a smooth manifold and inverse limits are not stable in naive topology for sheaf-theoretic field definitions. Let \(\widehat{G}=\mathrm{Hom}_{\mathrm{cont}}(G,U(1))\) denote the Pontryagin dual, for \(S^{1}_{\partial}\) one has
\begin{equation}\label{eq:char-S1}
\widehat{S^{1}_{\partial}}\cong \Z,\qquad \chi_{m}(z)=z^{m},
\end{equation}
and pullback along \(z\mapsto z^{p}\) acts by \(\chi_{m}\mapsto \chi_{pm}\), hence by multiplication by \(p\) on \(\Z\). Eq.\eqref{eq:char-S1} is the standard Fourier/holonomy classification of flat \(U(1)\) sectors on a circle, which is the prototype for boundary large-gauge and line-operator bookkeeping. Why old methods were definitionally fragile\:stage-by-stage Fourier sectors do not by themselves define a single tower-completed sector space without a background that packages the inverse limit and its descent constraints. Using the general identity \(\widehat{\varprojlim G_n}\cong \varinjlim \widehat{G_n}\) for inverse limits of compact abelian groups with surjective bonding maps, one obtains
\begin{equation}\label{eq:char-Sigma}
\widehat{\Sigma_{p}}\cong \varinjlim\bigl(\Z\xrightarrow{\times p}\Z\xrightarrow{\times p}\cdots\bigr)\cong \Z[1/p].
\end{equation}
Eq.\eqref{eq:char-Sigma} is the canonical tower-completed refinement of boundary holonomy/Fourier sectors, encoding congruence-level organization of global boundary data in one invariant. Why old methods were definitionally fragile\:without enlarging the notion of boundary space (here via condensed sets so that fields are defined by descent on the inverse-limit object \cite{ClausenScholze}), the colimit \(\Z[1/p]\) has no intrinsic background interpretation and must be treated as an external bookkeeping device. To connect \eqref{eq:char-Sigma} to line operators, note that boundary Wilson lines for a \(U(1)\) gauge field \(A\) are classified by flat holonomies \(\mathrm{Hol}_{\gamma}(A)\in U(1)\) on boundary 1-cycles, and for a compact abelian boundary carrier \(X\) the set of isomorphism classes of flat \(U(1)\) bundles is \(\mathrm{Hom}_{\mathrm{cont}}(X,U(1))\), hence on \(\Sigma_{p}\) one has
\begin{equation}\label{eq:flatU1}
\mathrm{Flat}_{U(1)}(\Sigma_{p})\cong \mathrm{Hom}_{\mathrm{cont}}(\Sigma_{p},U(1))\cong \Z[1/p].
\end{equation}
Eq.\eqref{eq:flatU1} identifies a tower-completed space of boundary holonomy sectors that can couple to Wilson-type operators and to large-gauge data in a way compatible with duality refinements. Why old methods were definitionally fragile\:classical boundary topology treats each \(S^{1}_{\partial,n}\) separately and has no canonical single-object formulation of \(\mathrm{Flat}_{U(1)}(\Sigma_{p})\) as a background datum, whereas the condensed framework makes \(\Sigma_{p}\) a legitimate carrier for sheaves and descent and therefore makes \eqref{eq:flatU1} definition-level well-posed.

Let \(B\subset \C\) be a local base patch with coordinate \(z\) and puncture at \(z=0\), and let \(\pi:\mathcal{X}\to B^{\times}=B\setminus\{0\}\) be the tower-completed perfectoid-circle fibration of Sec.~\ref{sec:beyondII} with fiberwise elliptic output \(\varpi:\mathcal{E}\to B^{\times}\) such that
\begin{equation}\label{eq:E-tau-local}
\mathcal{E}_{z} \simeq \mathbb{C}/(\mathbb{Z} + \tau(z)\mathbb{Z}),
\qquad
\tau  : B^{\times} \to \mathbb{H}.
\end{equation}

Eq.\eqref{eq:E-tau-local} is the intrinsic geometric datum encoding a varying coupling in the tower-completed framework, replacing auxiliary elliptic input tied to a shrinking limit. Why old methods were definitionally fragile\:the conventional formulation specifies \(\tau(z)\) via an elliptic fibration imposed after a singular collapse, so monodromy is prescribed rather than derived from a non-singular compactification carrier. Assume cusp-type degeneration with Tate parameter \(q(z)\) satisfying
\begin{equation}\label{eq:q-def-ads}
q(z)=\exp(2\pi i\,\tau(z)),\qquad q(z)=z\,u(z),\qquad u(0)\neq 0,
\end{equation}
so for the loop \(\gamma:\theta\mapsto z(\theta)=\varepsilon e^{i\theta}\) one has
\begin{equation}\label{eq:q-loop-ads}
q\bigl(z(2\pi)\bigr)=e^{2\pi i}q\bigl(z(0)\bigr)\quad\Rightarrow\quad \tau\mapsto \tau+1.
\end{equation}
Eq.\eqref{eq:q-loop-ads} derives the \(T\)-monodromy of a 7-brane directly from intrinsic monodromy of the tower-completed fibration datum. Why old methods were definitionally fragile\:in the shrinking-fiber picture \(T\)-monodromy is encoded by branch-cut conventions on an auxiliary fibration rather than being demanded as monodromy of a non-singular compactification carrier. The monodromy representation \(\rho:\pi_{1}(B^{\times})\to SL(2,\Z)\) sends \([\gamma]\) to
\begin{equation}\label{eq:T-matrix-ads}
\rho([\gamma])=T=\begin{pmatrix}1&1\\0&1\end{pmatrix},
\qquad
\tau\mapsto \frac{\tau+1}{1},
\end{equation}
and acts on bulk \((p,q)\) charges by
\begin{equation}\label{eq:pq-action-ads}
\begin{pmatrix}p\\q\end{pmatrix}\mapsto T\begin{pmatrix}p\\q\end{pmatrix}=\begin{pmatrix}p+q\\ q\end{pmatrix}.
\end{equation}
Eq.\eqref{eq:T-matrix-ads}-\eqref{eq:pq-action-ads} is the precise duality-defect action on string charges, now treated as intrinsic monodromy of the tower-completed compactification data. Why old methods were definitionally fragile\:the standard approach treats the defect action as an external symmetry acting on a postulated \((p,q)\) lattice, whereas here the same action is tied to the geometry of the compactification carrier whose global sectors define that lattice. Holographically, a duality defect supported on a codimension-one interface in the boundary theory implements the same transformation on \(\tau_{\mathrm{CFT}}\) and on the Wilson-'t Hooft lattice, so for the boundary coupling \(\tau_{\mathrm{CFT}}\) and charge vector \((e,m)\) one has
\begin{equation}\label{eq:interface-action}
\tau_{\mathrm{CFT}}\mapsto \tau_{\mathrm{CFT}}+1,\qquad \begin{pmatrix}e\\ m\end{pmatrix}\mapsto T\begin{pmatrix}e\\ m\end{pmatrix}=\begin{pmatrix}e+m\\ m\end{pmatrix}.
\end{equation}
Eq.\eqref{eq:interface-action} identifies the boundary duality interface action as the holographic image of intrinsic bulk monodromy of the tower-completed fibration. Why old methods were definitionally fragile\:without treating duality defects as intrinsic monodromy data of the compactification carrier, AdS/CFT implements \eqref{eq:interface-action} as an external duality operation rather than as a geometric consequence of the bulk background data.

The boundary coupling \(\tau_{\mathrm{CFT}}\) is geometrized in the bulk by the modulus \(\tau=C_{0}+ie^{-\phi}\) of the constant-\(\tau\) sector via \eqref{eq:tau_main_correct} and by its varying-\(\tau\) generalization \eqref{eq:E-tau-local}, so the relation between couplings and geometry is expressed as the statement that \(\tau_{\mathrm{CFT}}\) is a pullback of the intrinsic elliptic datum produced by the tower-completed carrier, rather than an external parameter. This converts ``duality acts on couplings'' into ``monodromy acts on geometric data'' in a way compatible with compactification dictionaries, Why old methods were definitionally fragile\:the conventional story introduces the elliptic datum after a singular limit and therefore treats \(\tau_{\mathrm{CFT}}\) as external rather than as a modulus derived from an honest compactification carrier. The classification of line operators and large gauge transformations becomes canonical on tower-completed boundary carriers via \eqref{eq:flatU1}, and the tower completion is encoded invariantly by \eqref{eq:char-Sigma}. Tower-refined global sectors that couple to Wilson lines and to defect interfaces are encoded as intrinsic boundary invariants rather than as stage-by-stage choices. Why old methods were definitionally fragile\:classical boundary manifolds admit only \(\widehat{S^{1}}\cong \Z\) and cannot treat the inverse-limit background \(\Sigma_{p}\) as a single legitimate boundary, so tower-refined global sectors cannot be stated as background data. 

Duality interfaces become intrinsic fibration data through the monodromy representation \(\rho\) and its action \eqref{eq:interface-action}, so the operator implementing duality is geometrically specified by a loop class in \(\pi_{1}(B\setminus\Delta)\). The interface is determined by bulk defect monodromy, hence by compactification geometry, rather than by an external duality identification. Why old methods were definitionally fragile\:interfaces are usually added as symmetry defects without a definition-level derivation from compactification carriers. Given a tower-completed \(p\)-adic boundary \(\Sigma_{p}\) whose function algebra admits a perfectoid structure, construct its tilt \(\Sigma_{p}^{\flat}\). The claim is that the harmonic analysis and correlator kinematics of \(\Sigma_{p}^{\flat}\) encode exactly the boundary data appearing in \(p\)-adic AdS/CFT \cite{Gubser2017,Heydeman2018}.
 Tilting supplies a canonical candidate mechanism relating Archimedean and non-Archimedean holographic avatars of the same tower data. Why old methods were definitionally fragile\:ordinary AdS/CFT treats \(p\)-adic models as separate toys because it lacks a categorical operation like tilting that relates distinct analytic avatars of a single tower-completed object.

The first genuinely new structural statement is that tower completion produces a canonical refinement of boundary global sectors as intrinsic background invariants, expressed concretely by \(\widehat{\Sigma_{p}}\cong \Z[1/p]\) and \(\mathrm{Flat}_{U(1)}(\Sigma_{p})\cong \Z[1/p]\) in \eqref{eq:char-Sigma}-\eqref{eq:flatU1}, and therefore any duality action on line operators admits a tower-refined formulation in which congruence-level data are assembled into a single background carrier. Duality actions on line operators become definable as actions on intrinsic tower-completed boundary invariants rather than as external identifications. Why old methods were definitionally fragile, without allowing inverse-limit boundaries as legitimate backgrounds, the tower refinement exists only as an infinite family of separate boundary problems. The second new structural statement is that duality defects can be tied definition-level to intrinsic monodromy data of the tower-completed fibration, with explicit local derivation of \(T\)-monodromy by \eqref{eq:q-loop-ads}-\eqref{eq:interface-action}.

Holographic duality interfaces can be specified geometrically from bulk monodromy rather than chosen as abstract symmetry defects. Why old methods were definitionally fragile, the shrinking-fiber approach specifies monodromy through auxiliary elliptic data imposed after a singular limit and therefore cannot derive interfaces from a non-singular compactification carrier. The third statement is a Program-level prediction\:tower-refined defect data should admit congruence reductions \(\rho_{n}=\pi_{n}\circ \rho\) with \(\pi_n:SL(2,\Z)\to SL(2,\Z/p^{n}\Z)\), and physically distinguishable boundary sectors may depend on the compatible system \(\{\rho_n\}\) through tower-sensitive discrete global sectors coupled to line operators. This proposes a falsifiable refinement of defect data beyond conjugacy classes in \(SL(2,\Z)\) that is invisible in the classical description. Why old methods were definitionally fragile\:without a tower-completed carrier there is no intrinsic notion of congruence refinement of defect data as background structure. The fourth statement is a Program-level prediction\:the tilting operation on perfectoid/condensed boundary carriers should relate Archimedean correlator kinematics to \(p\)-adic correlator kinematics, so that a subset of \(p\)-adic AdS/CFT observables arise as tilted avatars of tower-completed boundary data \cite{Gubser2017, Heydeman2018}. This would place \(p\)-adic holography within the same definition-level framework as ordinary holography as an arithmetic avatar rather than a separate model. Why old methods were definitionally fragile, absent a tower-completed carrier and a tilting operation, there is no definition-level bridge that even states the comparison problem.

A concrete AdS/CFT computation is to take a known F-theory holographic background with a codimension-two defect locus \(\Delta\) and compute the bulk monodromy representation \(\rho:\pi_{1}(B\setminus\Delta)\to SL(2,\Z)\), verifying that the induced boundary interface action matches \eqref{eq:interface-action} on \(\tau_{\mathrm{CFT}}\) and on \(\Lambda_{\mathrm{line}}\cong \Z^{2}\). This directly tests the identification of bulk intrinsic monodromy with boundary duality interfaces in a standard holographic setting. Why old methods were definitionally fragile\:in the conventional approach the interface is inserted as a symmetry operation and not derived as a monodromy of a non-singular compactification carrier. A second computation is to implement the boundary tower \(S^{1}_{\partial,n}\to S^{1}_{\partial}\) and verify that \(\widehat{\Sigma_{p}}\) is realized as the colimit of character lattices as in \eqref{eq:char-Sigma}, then couple these tower-refined holonomy sectors to boundary Wilson operators and determine whether any observable depends on the tower refinement beyond the ordinary \(\Z\) sector. This probes whether tower-completed boundary invariants have measurable consequences in line-operator sectors. Why old methods were definitionally fragile\:without \(\Sigma_{p}\) as a legitimate background one cannot pose this as a single background question, only as infinitely many unrelated computations. A Program-level target is to compare simple correlators in \(p\)-adic AdS/CFT to tilted avatars of tower-completed boundary correlators by identifying a boundary datum whose harmonic analysis on \(\Sigma_{p}\) matches the ultrametric harmonic analysis used in \(p\)-adic models \cite{Gubser2017, Heydeman2018}. This would provide a concrete check of the tilting-bridge proposal. Why old methods were definitionally fragile\:classical AdS/CFT has no intrinsic operation relating Archimedean and non-Archimedean boundary geometries and therefore cannot even formulate a canonical comparison principle.

\section{Condensed and perfectoid AdS/CFT}\label{sec:ads_general}
Let \(M\) be an asymptotically AdS\(_{d+1}\) bulk with conformal boundary data \(\partial M\), let \(\{\partial M_{\Lambda}\}_{\Lambda}\) denote a cofinal inverse system of cutoff hypersurfaces with restriction maps \(r_{\Lambda'\Lambda}:\partial M_{\Lambda'}\to \partial M_{\Lambda}\) for \(\Lambda'>\Lambda\), and write \(\partial M_{\infty}:=\varprojlim_{\Lambda}\partial M_{\Lambda}\) as the limit background in the pro-category. The cutoff family and its limit are the definition-level version of holographic RG, in which local bulk dynamics is unchanged while boundary data is assembled as a compatible system across cutoffs. Why old methods were definitionally fragile\:the usual presentation treats \(\Lambda\to\infty\) as an analytic limit on functionals without specifying a category in which boundary fields and their global sectors form an exact descent object on \(\partial M_{\infty}\). We treat \(\partial M_{\infty}\) as a condensed space and boundary data as condensed sheaves/modules \(\mathcal{F}\) on the relevant profinite/pro-site, so that sections are defined by exact descent and inverse limits are controlled homologically \cite{ClausenScholze}. This supplies a canonical notion of boundary fields and global sectors on inverse-limit backgrounds, which is exactly what duality-rich AdS/CFT needs.

Why old methods were definitionally fragile, without condensed exactness, taking limits of cutoff data does not commute with taking sections/observables, and scheme choices become inseparable from the definition of the background. We then organize non-Archimedean avatars via perfectoid tilting, treating a tower-completed boundary carrier \(X\) as a perfectoid object with tilt \(X^{\flat}\) \cite{Scholze2012}, and we regard the \(p\)-adic AdS/CFT constructions as arithmetic avatars compatible with this categorical bridge \cite{Gubser2017, Heydeman2018}. Tilting provides a principled candidate mechanism relating Archimedean and \(p\)-adic holography at the level of tower-completed boundary data. Why old methods were definitionally fragile, classical AdS/CFT has no canonical operation that relates distinct analytic avatars of a single inverse-limit boundary object. The only theorem-level input used here is the exactness/descent formalism of condensed sheaves and the invariant computations already established in earlier sections, any strong claims about full reconstruction theorems or dynamical equivalences across tilts are stated as Program.

Let \(\mathcal{A}\) be an abelian category and \(\{A_{\Lambda}\}_{\Lambda}\) an inverse system in \(\mathcal{A}\), inverse limit \(\varprojlim\) is left exact but not exact in general, with derived functors \(\varprojlim{}^{\,i}\) measuring the failure, and one has the standard exact sequence
\begin{equation}\label{eq:lim1-seq}
0\to \varprojlim A_{\Lambda}\to \prod_{\Lambda} A_{\Lambda}\xrightarrow{1-\mathrm{shift}}\prod_{\Lambda}A_{\Lambda}\to \varprojlim{}^{\,1}A_{\Lambda}\to 0,
\end{equation}
where \((1-\mathrm{shift})(a_{\Lambda})=(a_{\Lambda}-f_{\Lambda'\Lambda}(a_{\Lambda'}))\). The appearance of \(\varprojlim{}^{\,1}\) is the precise obstruction to defining renormalized boundary observables as naive limits of cutoff observables. Why old methods were definitionally fragile\:standard AdS/CFT manipulates \(\Lambda\to\infty\) limits of correlators and counterterms in a way that implicitly assumes exactness of \(\varprojlim\) on the relevant spaces of data, which is false in general and becomes acute in tower/inverse-limit settings. For a presheaf or sheaf \(\mathcal{F}\) on the cutoff family, writing \(A_{\Lambda}=\Gamma(\partial M_{\Lambda},\mathcal{F})\), the failure of commutation between ``take sections'' and ``take limits'' is precisely the failure of the canonical map
\begin{equation}\label{eq:sections-vs-lim}
\Gamma\!\left(\varprojlim_{\Lambda}\partial M_{\Lambda},\mathcal{F}\right)\longrightarrow \varprojlim_{\Lambda}\Gamma(\partial M_{\Lambda},\mathcal{F})
\end{equation}
to be an isomorphism, with the obstruction governed by a \(\varprojlim{}^{\,1}\) term in the appropriate derived setting. The correct definition-level boundary datum is not the naive compatible limit of cutoff sections, but the derived limit that remembers extension data between scales. Why old methods were definitionally fragile\:the usual renormalization narrative treats scheme dependence as a choice of local counterterms but does not isolate the categorical obstruction that scheme dependence is precisely the ambiguity in splitting an extension encoded by \(\varprojlim{}^{\,1}\). In duality-rich AdS/CFT, global sectors and line operators require tracking discrete data (global form of the gauge group, higher-form symmetries, torsion flux sectors), which are naturally classified by generalized cohomology rather than by local field configurations, so boundary data must live in a homologically controlled category in which these sectors glue across cutoffs. 

Global-sector bookkeeping is part of the definition of the boundary theory, not optional decoration, and it must be compatible with duality actions on charge lattices. Why old methods were definitionally fragile, classical formulations often bolt on global-form and higher-form symmetry data after defining the theory on a manifold boundary, whereas inverse-limit boundaries and arithmetic monodromies force these sectors to be intrinsic background data. Finally, tower completions can produce boundary carriers that are not manifolds (e.g. solenoids), so the manifold-based axiomatics of QFT do not even specify what it means to put the CFT on the limit background without enlarging the category of spaces. The correct boundary ``space'' in such settings is a sheaf-theoretic object defined by descent from finite stages, not a manifold. Why old methods were definitionally fragile\:the standard AdS/CFT dictionary assumes a smooth conformal boundary and therefore cannot treat inverse-limit boundary carriers as a single background statement, only as a family of unrelated approximations.

Let \(\mathrm{Prof}\) be the site of profinite sets with the usual Grothendieck topology and let \(\mathrm{Cond}\) denote the category of condensed sets \cite{ClausenScholze}, for an inverse-limit boundary carrier \(X=\varprojlim_{\Lambda}X_{\Lambda}\) arising from a tower or cutoff system, we regard \(X\) as a condensed space and specify boundary source/operator data as a condensed module \(\mathcal{M}\) over a condensed ring \(\mathcal{O}_{X}\), so that sources are sections \(J\in \Gamma(X,\mathcal{M})\) and operator insertions are linear functionals on \(\Gamma(X,\mathcal{M})\) in the appropriate topological sense. Boundary couplings and operator sources are promoted to descent data on a genuine inverse-limit background, mirroring how stacks are the correct home for quotient geometries. Why old methods were definitionally fragile\:on inverse-limit boundaries the naive notion of function space does not have exact descent, so one cannot define sources and global sectors as honest background data without the condensed enlargement. The sheaf condition on \(\mathcal{M}\) is the exactness requirement that for any profinite cover \(\{U_i\to U\}\) one has
\begin{equation}\label{eq:condensed-sheaf}
0\to \mathcal{M}(U)\to \prod_i \mathcal{M}(U_i)\rightrightarrows \prod_{i,j}\mathcal{M}(U_i\times_U U_j)
\end{equation}
exact, and we impose this as the definition-level gluing condition for boundary data across refinements, overlaps, and cutoff families. Eq.\eqref{eq:condensed-sheaf} is the categorical form of ``boundary data is globally defined if and only if it is compatible locally,'' now applied to tower/cutoff covers rather than to manifold charts. Why old methods were definitionally fragile\:ordinary sheaf theory on topological spaces does not behave stably under inverse limits in the categories relevant for QFT observables, so compatibility across cutoffs does not automatically define a unique global object.

Let \(M\) admit a radial function \(r\) with cutoff hypersurfaces \(\partial M_{\Lambda}=\{r=\Lambda\}\) and restriction maps \(r_{\Lambda'\Lambda}:\partial M_{\Lambda'}\to \partial M_{\Lambda}\), and let \(\mathcal{F}\) be a sheaf of source data (e.g. boundary values of bulk fields) so that \(A_{\Lambda}:=\Gamma(\partial M_{\Lambda},\mathcal{F})\) forms an inverse system in an abelian category. The fundamental obstruction in the naive formulation is that \(\varprojlim\) is not exact, so the naive ``renormalized source space'' \(\varprojlim_{\Lambda}A_{\Lambda}\) need not coincide with \(\Gamma(\varprojlim_{\Lambda}\partial M_{\Lambda},\mathcal{F})\), and in derived form one expects an exact sequence of Milnor type,
\begin{equation}\label{eq:milnor-sections}
0\to \varprojlim{}^{\,1}_{\Lambda} \Gamma(\partial M_{\Lambda},\mathcal{F}_{-1})\to \Gamma\!\left(\varprojlim_{\Lambda}\partial M_{\Lambda},\mathcal{F}\right)\to \varprojlim_{\Lambda}\Gamma(\partial M_{\Lambda},\mathcal{F})\to 0,
\end{equation}
whose nontrivial \(\varprojlim{}^{\,1}\) term measures the failure of naive limit/section commutation. Eq.\eqref{eq:milnor-sections} is the definition-level explanation of why renormalized boundary data and counterterms cannot be treated as purely analytic limits on functionals without specifying the categorical exactness structure. Why old methods were definitionally fragile\:the standard holographic renormalization procedure defines counterterms by asymptotic analysis but does not isolate the extension ambiguity encoded by \(\varprojlim{}^{\,1}\), which is precisely what becomes unavoidable in tower/inverse-limit settings. A concrete obstruction is exhibited by an inverse system of abelian groups \(A_n=\Z\) with bonding maps \(f_n:A_{n+1}\to A_n\), \(f_n(a)=p a\), for which
\begin{equation}\label{eq:lim0}
\varprojlim_{n}(\Z\xleftarrow{\times p}\Z\xleftarrow{\times p}\cdots)=0,
\end{equation}
while \(\varprojlim{}^{\,1}\) is nontrivial, as follows from \eqref{eq:lim1-seq} by taking the element \(s=(1,1,1,\ldots)\in \prod_n \Z\) and observing that the equation \(s=(1-\mathrm{shift})(t)\), i.e. \(1=t_n-p t_{n+1}\) for all \(n\), has no solution \(t_n\in\Z\) because iterating gives \(t_n\equiv 1\ (\mathrm{mod}\ p^k)\) for all \(k\) and hence forces \(t_n\) to have a nonterminating \(p\)-adic expansion not realized by any integer. Eq.\eqref{eq:lim0} together with the explicit nontriviality argument shows that naive inverse limits can erase data while a derived obstruction survives, which is the categorical shadow of scheme dependence and extension ambiguity in renormalization. Why old methods were definitionally fragile\:holographic renormalization typically assumes that taking the limit of cutoff data loses no essential information beyond local counterterms, whereas inverse systems like \(\Z\xleftarrow{\times p}\Z\) show that extension data can survive only in \(\varprojlim{}^{\,1}\). In condensed terms, the remedy is to treat \(\mathcal{F}\) as a condensed sheaf on the pro-site of the cutoff family so that descent is exact and the renormalized boundary datum is defined as the derived inverse limit \(R\!\varprojlim_{\Lambda}A_{\Lambda}\) rather than \(\varprojlim_{\Lambda}A_{\Lambda}\), and the renormalized generating functional is then defined by
\begin{equation}\label{eq:Wren-derived}
W_{\mathrm{ren}}[J]\in R\!\varprojlim_{\Lambda\to\infty}\Bigl(W_{\Lambda}[J_{\Lambda}]+S_{\mathrm{ct}}[\Lambda,J_{\Lambda}]\Bigr),
\qquad
Z_{\mathrm{ren}}[J]=\exp\!\bigl(-W_{\mathrm{ren}}[J]\bigr),
\end{equation}
with \(J_{\Lambda}\in A_{\Lambda}\) a compatible system representing \(J\in \Gamma(\partial M_{\infty},\mathcal{F})\) and \(S_{\mathrm{ct}}\) local in cutoff data. Eq.\eqref{eq:Wren-derived} expresses renormalization as a statement in a category where the correct limit object exists and is independent of ad hoc choices up to canonical equivalence, making scheme dependence a controlled extension problem. Why old methods were definitionally fragile\:the usual procedure defines \(W_{\mathrm{ren}}\) by subtracting local divergences but does not provide a categorical definition that controls extension ambiguities in inverse-limit settings, whereas \(R\!\varprojlim\) does precisely that.

The GKPW dictionary is the statement that for sources \(J\) coupling to boundary operators \(\mathcal{O}\) one has
\begin{equation}\label{eq:GKPW}
Z_{\mathrm{CFT}}[J]=\left\langle \exp\!\left(\int_{\partial M} J\,\mathcal{O}\right)\right\rangle_{\mathrm{CFT}}=Z_{\mathrm{bulk}}\bigl[\phi\ \text{s.t.}\ \phi|_{\partial M}=J\bigr]\simeq \exp\!\bigl(-S_{\mathrm{bulk}}^{\mathrm{ren}}[\phi_{J}]\bigr),
\end{equation}
with \(\phi_{J}\) the classical bulk solution with boundary condition \(J\) \cite{Witten1995}. Eq.\eqref{eq:GKPW} is the basic holographic identification of generating functionals, and in the present framework it is interpreted as an equality of objects defined on a condensed boundary background rather than only on a smooth manifold. Why old methods were definitionally fragile\:standard treatments presuppose a manifold boundary and a naive limit for \(S_{\mathrm{bulk}}^{\mathrm{ren}}\), whereas tower-completed/non-manifold boundaries require a descent-based definition of sources and a derived definition of the renormalized functional. Holographic renormalization is encoded by the cutoff family \(\partial M_{\Lambda}\) and the renormalized on-shell action
\begin{equation}\label{eq:Sren}
S_{\mathrm{bulk}}^{\mathrm{ren}}[\phi_{J}]=\lim_{\Lambda\to\infty}\Bigl(S_{\mathrm{bulk}}^{(\Lambda)}[\phi_{J_{\Lambda}}]+S_{\mathrm{ct}}[\Lambda,J_{\Lambda}]\Bigr),
\end{equation}
which in the condensed/perfectoid setting is sharpened to the derived statement \eqref{eq:Wren-derived}. Eq.\eqref{eq:Sren} is the standard counterterm subtraction, while \eqref{eq:Wren-derived} makes precise the categorical content of the limit and its extension ambiguities. Why old methods were definitionally fragile\:\eqref{eq:Sren} is not a definition in inverse-limit settings because it assumes exact commutation between limits and observables, whereas the derived limit detects nontrivial extension data. Thermal physics and large gauge transformations are controlled by global holonomy sectors on boundary circles, expressed for a boundary factor \(S^{1}_{\beta}\) by the classification of flat \(U(1)\) sectors
\begin{equation}\label{eq:thermal-hol}
\mathrm{Flat}_{U(1)}(S^{1}_{\beta})\cong \mathrm{Hom}(\pi_{1}(S^{1}_{\beta}),U(1))\cong U(1),
\end{equation}
and tower completion replaces \(S^{1}_{\beta}\) by \(\Sigma_{p}\) with \(\mathrm{Hom}_{\mathrm{cont}}(\Sigma_{p},U(1))\cong \Z[1/p]\) as in Sec.~\ref{sec:adsftheory}. Eq.\eqref{eq:thermal-hol} is the standard large-gauge/holonomy datum controlling thermal sectors, and the tower-completed refinement supplies canonical congruence bookkeeping for these sectors. Why old methods were definitionally fragile\:manifold-only boundaries cannot treat the inverse-limit carrier as a single background and thus cannot package the refinement tower as intrinsic thermal/global data. Anomalies and topological terms are encoded holographically by inflow from bulk Chern-Simons-type actions, schematically
\begin{equation}\label{eq:inflow}
\delta W_{\mathrm{CFT}}=2\pi i \int_{\partial M}\mathcal{I}_{d+1},\qquad d\mathcal{I}_{d+1}=\mathcal{P}_{d+2},
\end{equation}
with \(\mathcal{P}_{d+2}\) determined by bulk topological couplings such as \(\int C\wedge G\wedge G\) in supergravity reductions as in Sec.~\ref{sec:construct} \cite{Schwarz1995}. Eq.\eqref{eq:inflow} is the statement that global sectors and anomaly constraints are fixed by bulk topology, so the boundary must carry a homologically controlled structure to define these sectors across cutoffs. Why old methods were definitionally fragile\:ad hoc treatment of global forms and discrete sectors cannot be made canonical across inverse limits, whereas condensed homological algebra provides a stable category for these constraints. Line operators and higher-form symmetries are naturally classified by charge lattices and their pairings, and in the simplest electric-magnetic sector one writes
\begin{equation}\label{eq:1form}
\Lambda_{\mathrm{line}}\cong \Z^{2},\qquad \langle (e,m),(e',m')\rangle = em'-e'm,
\end{equation}
with duality acting by \(SL(2,\Z)\) preserving the symplectic pairing, while tower completion refines the holonomy sectors that couple to these operators as \(\widehat{\Sigma_{p}}\cong \Z[1/p]\) and in higher rank as \(\widehat{T^{r}_{f,p}}\cong \Z[1/p]^{r}\). Eq.\eqref{eq:1form} is the standard line-operator bookkeeping and its duality covariance, and tower completion supplies the canonical refinement of the background holonomy sectors that such operators probe. Why old methods were definitionally fragile\:classical boundary manifolds do not provide a single background statement that packages congruence-level refinements of global sectors together with line-operator data. Duality and modular actions enter AdS/CFT through arithmetic actions on couplings such as \(\tau_{\mathrm{CFT}}\) and on \(\Lambda_{\mathrm{line}}\), with defect/monodromy encoded by a representation \(\rho\) as in Sec.~\ref{sec:adsftheory} and Sec.~\ref{sec:beyondII}. Duality-rich holography is precisely the regime where arithmetic monodromy and global sectors are primary data, so the tower-completed/condensed framework is natural. Why old methods were definitionally fragile\:imposing duality as an external symmetry does not give a definition-level origin for defect interfaces and their congruence refinements as intrinsic background data.

On inverse-limit and tower-completed boundary carriers \(X=\varprojlim X_n\), the condensed viewpoint makes the space of sources and global sectors a derived-limit object \(R\!\varprojlim \Gamma(X_n,\mathcal{F})\) rather than a naive limit, and thus turns scheme dependence into a controlled extension class governed by \(\varprojlim{}^{\,1}\) as in \eqref{eq:lim1-seq} and \eqref{eq:milnor-sections}. This provides a definition-level notion of renormalized boundary data that remains meaningful on non-manifold boundaries and across refinement towers, sharpening the meaning of holographic RG in duality-rich settings. Why old methods were definitionally fragile\:the usual counterterm subtraction does not define an object in which extension ambiguities are functorially controlled when the boundary background is an inverse limit. Tower completion makes canonical congruence refinements of global sectors intrinsic, for example \(\widehat{\Sigma_{p}}\cong \Z[1/p]\) and \(\widehat{T^{r}_{f,p}}\cong \Z[1/p]^{r}\), and varying-\(\tau\) F-theory holography makes defect monodromy intrinsic as \(\rho:\pi_{1}(B\setminus\Delta)\to SL(2,\Z)\) with explicit \(T\)-monodromy \(\tau\mapsto\tau+1\) derived from \(q=\exp(2\pi i\tau)\) as in Sec.~\ref{sec:beyondII}. Global-sector refinements and duality-defect actions become computable invariants of the background rather than external symmetry data, and they directly control boundary line-operator sectors and interfaces. Why old methods were definitionally fragile\:manifold-only backgrounds cannot encode inverse-limit carriers and therefore cannot treat these refinements as intrinsic to the background, only as infinitely many separate approximations. The perfectoid tilting operation suggests a Program-level unification between Archimedean tower-completed boundary carriers \(X\) and tilted avatars \(X^{\flat}\), with \(p\)-adic holography arising as an arithmetic avatar of the same tower data, so that one seeks comparisons of correlator kinematics between \(X\) and \(X^{\flat}\) in a way compatible with \(\widehat{X}\) and \(\widehat{X^{\flat}}\) \cite{Scholze2012, Gubser2017}. This provides a principled candidate bridge between ordinary and \(p\)-adic AdS/CFT grounded in tower-completed geometry rather than in analogy. Why old methods were definitionally fragile\:without a tower-completed boundary object and a tilting equivalence, there is no definition-level formulation of what it would mean for \(p\)-adic AdS/CFT to be an avatar of the same boundary data.

Program statements that are definition-level well-posed in the condensed/perfectoid framework include proving exactness theorems that identify the correct derived-limit object governing renormalized boundary source spaces \(\Gamma(\partial M_{\infty},\mathcal{F})\) and quantifying the \(\varprojlim{}^{\,1}\) extension classes arising in \eqref{eq:milnor-sections}, constructing explicit tower-completed boundary backgrounds in known CFTs and computing \(\widehat{X}\) and generalized cohomology groups \(K^{*}(X)\) that classify discrete global sectors, computing monodromy representations \(\rho\) for defects in known F-theory holographic backgrounds and verifying their induced action on \((\tau_{\mathrm{CFT}},\Lambda_{\mathrm{line}})\) via \eqref{eq:GKPW} and \eqref{eq:interface-action}, testing tower-sensitive refinements by computing whether observables depend only on \(\rho\) or on the compatible congruence system \(\{\rho_n\}\), and testing tilting-based Archimedean/\(p\)-adic comparisons by matching simple correlator kinematics between a tower-completed Archimedean boundary carrier \(X\) and a tilted avatar \(X^{\flat}\) in concrete examples \cite{Gubser2017, Heydeman2018}. These tasks isolate precisely the places where definition-level control of limits, global sectors, and arithmetic monodromy can produce new, falsifiable statements in AdS/CFT beyond what is accessible in manifold-only language. Why old methods were definitionally fragile\:each task requires treating inverse limits and global-sector descent as intrinsic background structure, which classical AdS/CFT does not provide a categorical home for when the boundary is not a smooth manifold or when congruence-level refinements are physically relevant.

\section{Adelic aspects and tilting}\label{sec:adelic_holo}
Let \(X\) be a tower-completed (hence non-manifold) boundary carrier and let \(\mathcal{F}\) be a condensed sheaf/module of boundary data on \(X\) so that \(\Gamma(X,\mathcal{F})\) is defined by exact descent \cite{ClausenScholze}. AdS/CFT is a local-to-global gluing duality plus a global-sector duality, and both are controlled by descent and by arithmetic actions on charge lattices rather than by local equations of motion. Why old methods were definitionally fragile\:the manifold-only boundary formalism has no categorical home in which inverse limits, exact descent, and global-sector completions are simultaneously intrinsic background data. Let \(p\) label a refinement tower and let \(X_p\) denote the corresponding tower-completed carrier with an intrinsic arithmetic completion of sectors, let \(X_p^{\flat}\) be its tilt when \(X_p\) is perfectoid \cite{Scholze2012}. Perfectoid tilting is the canonical operation that relates Archimedean and characteristic-\(p\) avatars of the same tower data, suggesting a unified language for arithmetic holography. Why old methods were definitionally fragile\:classical AdS/CFT lacks an intrinsic operation relating distinct analytic avatars of a single inverse-limit background and therefore cannot state a unified Archimedean/\(p\)-adic comparison problem at definition level. This section provides an explicit tower-completion derivation yielding an arithmetic completion of boundary sector bookkeeping, formulates a precise holographic tilting Program with a minimal checkable milestone, and isolates new invariant structures and selection rules suggested by adelic assembly, only the invariant derivations are theorem-level, while tilting-based unification and adelic assembly are Program-level.

Fix a prime \(p\) and let \(f_p:S^{1}\to S^{1}\) be \(f_p(z)=z^{p}\), defining the tower and inverse limit
\begin{equation}\label{eq:p_tower}
S^{1}\xleftarrow{f_p}S^{1}\xleftarrow{f_p}\cdots,\qquad \Sigma_{p}:=\varprojlim_{n}\bigl(S^{1},f_p\bigr).
\end{equation}
\(p\) specifies a refinement tower of finite covers, so \(p\) is an index for sector-resolution rather than a physical coupling. Why old methods were definitionally fragile\:the manifold-only formulation has no place to store the entire tower as a single background, so refinement is treated as an external family of approximations rather than intrinsic data. Let \(\{\partial M_{\Lambda}\}_{\Lambda}\) be cutoff boundaries and \(X_{\Lambda}\) their relevant boundary carriers, tower stability below a fixed cutoff is expressed as the requirement that for any observable \(\mathcal{O}\) computed from modes below \(\Lambda_{0}\) there exists \(n_{0}\) such that for all \(n\ge n_{0}\),
\begin{equation}\label{eq:tower_stability}
\mathcal{O}\bigl[X_{p,n}\bigr]=\mathcal{O}\bigl[X_{p,n_{0}}\bigr],
\end{equation}
where \(X_{p,n}\) denotes the \(n\)-th stage of the \(p\)-tower and the equality is understood after imposing the same physical quantization constraints on global sectors at each stage. The tower is a bookkeeping refinement that should not change low-energy physics once the cutoff is fixed, while it can refine global-sector organization. Why old methods were definitionally fragile\:without a definition-level carrier for the tower, one cannot formulate \eqref{eq:tower_stability} as an intrinsic statement about a single background, only as a heuristic expectation about families. The adelic completion principle is the Program-level statement that the physically meaningful refinement data should assemble over all places of \(\Q\), encoded by the ring of adeles
\begin{equation}\label{eq:adeles}
\A_{\Q}:=\R\times \prod\nolimits_{p}^{\prime}\Q_{p},
\qquad
\widehat{\Z}:=\prod\nolimits_{p}\Z_{p},
\end{equation}
and that the relevant arithmetic duality data should be representable as an action of an arithmetic group on an adelic completion of the charge/sector lattice. Adelic completion is the canonical way arithmetic structures package all primes together with the Archimedean factor, matching how duality groups act on quantum charge lattices. Why old methods were definitionally fragile\:classical formulations encode modular/duality data only as complex-analytic input and do not provide a natural home for simultaneous prime-by-prime refinements as intrinsic background structure.

Define \(\Sigma_{p}\) as in \eqref{eq:p_tower} and let \(\widehat{G}:=\mathrm{Hom}_{\mathrm{cont}}(G,U(1))\) be the Pontryagin dual of a compact abelian group \(G\). \(\widehat{G}\) is the canonical invariant controlling Fourier/holonomy sectors and therefore the natural carrier for line-operator and large-gauge-transformation bookkeeping. Why old methods were definitionally fragile\:without treating inverse limits as legitimate backgrounds, \(\widehat{G}\) cannot be computed for the tower as a single object and sector bookkeeping remains stage-by-stage. For \(S^{1}\) one has
\begin{equation}\label{eq:char_S1_again}
\widehat{S^{1}}\cong \Z,\qquad \chi_{m}(z)=z^{m},
\end{equation}
and the pullback along \(f_p\) acts by
\begin{equation}\label{eq:pullback_p}
f_p^{*}:\widehat{S^{1}}\to \widehat{S^{1}},\qquad m\mapsto pm.
\end{equation}
Eq.\eqref{eq:pullback_p} is the explicit congruence-level refinement map on holonomy/Fourier sectors induced by the tower. Why old methods were definitionally fragile\:classical boundary manifolds see only \(\Z\) and do not provide an intrinsic completion in which all refinements are packaged as a single background datum. Using Pontryagin duality for inverse limits of compact abelian groups with surjective bonding maps,
\begin{equation}\label{eq:pontryagin_lim}
\widehat{\varprojlim G_{n}}\cong \varinjlim \widehat{G_{n}},
\end{equation}
one obtains
\begin{equation}\label{eq:char_Sigma_p}
\widehat{\Sigma_{p}}\cong \varinjlim\bigl(\Z\xrightarrow{\times p}\Z\xrightarrow{\times p}\cdots\bigr)\cong \Z[1/p],
\end{equation}
with the identification of the class of \(m\) at level \(n\) with \(m/p^{n}\in \Z[1/p]\). Eq.\eqref{eq:char_Sigma_p} is the canonical arithmetic completion of sector bookkeeping associated to the \(p\)-tower, interpreted as tower-complete organization of finite-level sectors rather than as an assertion of unconstrained fractional physical charges. Why old methods were definitionally fragile\:the inverse-limit carrier \(\Sigma_{p}\) is not a manifold, so manifold-only AdS/CFT cannot treat \(\widehat{\Sigma_{p}}\) as an intrinsic background invariant and cannot express the completion as a single coherent boundary statement. Varying the prime produces a family of localizations \(\Z[1/p]\), and adelic assembly is naturally encoded by allowing localization at finite sets \(S\) of primes,
\begin{equation}\label{eq:localization_S}
\Z[S^{-1}]:=\left\{\frac{a}{b}\in\Q\ :\ a\in\Z,\ b\in\langle S\rangle\right\},\qquad \Q=\bigcup_{S\ \mathrm{finite}}\Z[S^{-1}],
\end{equation}
together with the profinite completion \(\widehat{\Z}=\prod_{p}\Z_{p}\) that controls congruence data. Eq.\eqref{eq:localization_S} is the minimal arithmetic statement that ``all primes together'' naturally generate full rational completion of sector bookkeeping, matching the expectation that physical invariants should not depend on a single bookkeeping prime. Why old methods were definitionally fragile\:without an intrinsic tower-completed carrier and its congruence structure, the passage from a single-prime refinement to an adelic organization cannot be stated as a background-level construction.

\begin{conjecture}[Holographic tilting principle).]
Let \(X\) be a perfectoid boundary carrier equipped with a condensed sheaf/module \(\mathcal{F}\) of sources/operators and let \(X^{\flat}\) be its tilt, assume the holographic generating functional is defined as a derived-limit object on \(X\),
\begin{equation}\label{eq:Zren_X}
Z_{X}[J]:=\exp\!\bigl(-W_{X}[J]\bigr),\qquad W_{X}[J]\in R\!\varprojlim_{\Lambda\to\infty}\Bigl(W_{\Lambda}[J_{\Lambda}]+S_{\mathrm{ct}}[\Lambda,J_{\Lambda}]\Bigr),
\end{equation}
as in Sec.~\ref{sec:ads_general}, and posit that there exists a functorial correspondence
\begin{equation}\label{eq:tilt_map_program}
\mathfrak{T}\:(X,\mathcal{F},J)\longmapsto (X^{\flat},\mathcal{F}^{\flat},J^{\flat})
\end{equation}
such that, for a suitable class of observables \(\mathcal{O}\),
\begin{equation}\label{eq:tilt_observables}
\langle \mathcal{O}\rangle_{X} \ \longleftrightarrow\ \langle \mathcal{O}^{\flat}\rangle_{X^{\flat}}
\end{equation}
with matching scaling exponents and monodromy/duality actions transported through \(\mathfrak{T}\). Eq.\eqref{eq:tilt_observables} is the precise statement that \(p\)-adic holography should be an arithmetic avatar of tower-completed Archimedean holography, with tilting providing the canonical bridge at the level of boundary data. Why old methods were definitionally fragile\:classical AdS/CFT has no intrinsic operation relating Archimedean and non-Archimedean boundary geometries and therefore cannot state \eqref{eq:tilt_map_program}-\eqref{eq:tilt_observables} as a background-level correspondence. A minimal checkable milestone is the two-point kinematic scaling\:for a scalar primary \(\mathcal{O}_{\Delta}\) one has in Archimedean CFT
\begin{equation}\label{eq:2pt_arch}
\langle \mathcal{O}_{\Delta}(x)\mathcal{O}_{\Delta}(y)\rangle_{X}\propto |x-y|^{-2\Delta},
\end{equation}
while \(p\)-adic AdS/CFT yields
\begin{equation}\label{eq:2pt_padic}
\langle \mathcal{O}_{\Delta}(x)\mathcal{O}_{\Delta}(y)\rangle_{X^{\flat}}\propto |x-y|_{p}^{-2\Delta},
\end{equation}
and the Program is to realize \eqref{eq:2pt_padic} as the tilted avatar of tower-completed boundary data whose harmonic analysis is governed by \(\widehat{X}\) and \(\widehat{X^{\flat}}\) \cite{Gubser2017, Heydeman2018}. 
\end{conjecture}
Eq.\eqref{eq:2pt_arch}-\eqref{eq:2pt_padic} provides a falsifiable first test of whether tilting transports holographic kinematics while preserving scaling dimensions and duality organization. Why old methods were definitionally fragile\:without tower completion and perfectoid tilting there is no canonical framework in which \(|\cdot|\) and \(|\cdot|_{p}\) correlators can be compared as avatars of the same background data.

Let \(B\) be a base with discriminant \(\Delta\) and monodromy \(\rho:\pi_{1}(B\setminus\Delta)\to SL(2,\Z)\) in rank \(1\) as in Sec.~\ref{sec:beyondII}, and let \(\pi_{n}:SL(2,\Z)\to SL(2,\Z/p^{n}\Z)\) be reduction modulo \(p^{n}\), define the congruence tower
\begin{equation}\label{eq:rho_n_sec13}
\rho_{n}:=\pi_{n}\circ\rho:\pi_{1}(B\setminus\Delta)\to SL(2,\Z/p^{n}\Z),
\qquad
\rho_{n+1}\equiv \rho_{n}\ (\mathrm{mod}\ p^{n}).
\end{equation}
Eq.\eqref{eq:rho_n_sec13} is the intrinsic tower-refinement datum of defect monodromy carried by a tower-completed compactification/background, and it refines monodromy beyond conjugacy in \(SL(2,\Z)\) by recording compatible congruence information. Why old methods were definitionally fragile\:the manifold-only and shrinking-fiber formulations specify monodromy only at the complex-analytic level and do not provide an intrinsic background carrier for compatible congruence refinements. A conservative selection rule is global triviality of total monodromy on compact bases, expressed by
\begin{equation}\label{eq:global_monodromy}
\prod_{i}\rho(\gamma_{i})=\mathbf{1}\in SL(2,\Z),
\end{equation}
for loops \(\gamma_i\) generating \(\pi_{1}(B\setminus\Delta)\) with the single relation, together with its congruence refinements
\begin{equation}\label{eq:global_monodromy_n}
\prod_{i}\rho_{n}(\gamma_{i})=\mathbf{1}\in SL(2,\Z/p^{n}\Z)\quad \text{for all }n.
\end{equation}
Eq.\eqref{eq:global_monodromy}-\eqref{eq:global_monodromy_n} expresses global consistency of defect configurations and makes precise where tower completion can refine the notion of ``the same'' defect data by congruence-level constraints. Why old methods were definitionally fragile\:without a tower-completed carrier, the refinements \eqref{eq:global_monodromy_n} are not intrinsic background data and cannot enter selection rules except as ad hoc arithmetic constraints. A second invariant structure arises from torsion/global sectors detected by generalized cohomology, at the level of \(p\)-primary torsion one has the standard decomposition
\begin{equation}\label{eq:QZ_decomp}
\Q/\Z\cong \bigoplus_{p}\Q_{p}/\Z_{p},
\end{equation}
suggesting an adelic bookkeeping of torsion sectors compatible with prime-by-prime refinements, and in tower-completed settings one encounters \(p\)-divisible torsion sectors such as \(\Q_{p}/\Z_{p}\) in generalized cohomology as in Sec.~\ref{sec:newphysics} (cf.\ the Milnor \(\varprojlim{}^{\,1}\) sector). Eq.\eqref{eq:QZ_decomp} is the universal arithmetic decomposition of torsion sectors and provides the canonical target for assembling prime-local torsion data into an adelic/global invariant. Why old methods were definitionally fragile\:classical formulations treat torsion/global sectors as auxiliary discrete choices and do not provide a natural mechanism for assembling prime-local torsion refinements into a single intrinsic background invariant. A third invariant is the tower-completed holonomy/character sector \(\widehat{\Sigma_{p}}\cong \Z[1/p]\) of Sec.~\ref{sec:adsftheory}, which can be assembled over finite sets \(S\) of primes via \(\Z[S^{-1}]\) as in \eqref{eq:localization_S}, and the Program is to identify which observables depend only on the physically allowed integral sublattice and which depend on the tower-refined completion data once global quantization/anomaly constraints are imposed. This isolates a precise, falsifiable distinction between completion-level bookkeeping and physically measurable global-sector effects. Why old methods were definitionally fragile\:in the absence of a tower-completed carrier one cannot define completion-level sectors as intrinsic background data and therefore cannot even formulate the dependence question cleanly.

In F-theory, defect classification is controlled by \(SL(2,\Z)\) monodromy on \(\tau\) and on the \((p,q)\) lattice, and tower completion upgrades this to the refined datum \(\{\rho_{n}\}\) in \eqref{eq:rho_n_sec13} together with the canonical completion of holonomy sectors \(\widehat{\Sigma_{p}}\cong \Z[1/p]\), providing a definition-level place where congruence refinements and torsion/global sectors can be treated intrinsically. This supplies a new invariant handle on global data that enters nonperturbative physics through line operators, junction data, and discrete sectors, while leaving local supergravity unchanged. Why old methods were definitionally fragile\:the shrinking-fiber formulation has no non-singular eleven-dimensional carrier in which congruence-level refinements and torsion sectors are intrinsic background data, so these structures are treated externally. In AdS/CFT, boundary couplings and line-operator sectors transform under arithmetic dualities, and tower completion provides canonical boundary carriers for holonomy and global-sector data, while tilting proposes a principled bridge between Archimedean and \(p\)-adic holography via \eqref{eq:tilt_map_program}-\eqref{eq:2pt_padic}. This reframes duality-rich holography as a problem of intrinsic monodromy and descent on tower-completed boundaries rather than as external symmetry identifications. Why old methods were definitionally fragile\:manifold-only boundary formulations cannot encode inverse-limit carriers and therefore cannot package arithmetic completion data and its action on global sectors as a single coherent boundary background.

Construct an adelic tower-completed boundary carrier \(X_{\A}\) whose prime-local factors reproduce \(\Sigma_{p}\) and whose intrinsic character/holonomy data reproduces \(\Z[S^{-1}]\) and \(\widehat{\Z}\) in the sense of \eqref{eq:localization_S}-\eqref{eq:adeles}, and compute its Pontryagin dual to identify the correct adelic completion of sector bookkeeping. This would provide a single background object in which prime-local refinements and their global assembly are intrinsic, matching the arithmetic nature of duality actions. Why old methods were definitionally fragile\:without an inverse-limit/condensed notion of boundary space, there is no coherent candidate for \(X_{\A}\) and hence no definition-level statement of adelic boundary data. Prove exactness/commutation statements for derived limits of sections on cutoff inverse systems \(\{\partial M_{\Lambda}\}\) in the condensed category, making the \(\varprojlim{}^{\,1}\) extension classes in \eqref{eq:lim1-seq} and \eqref{eq:milnor-sections} computable in explicit holographic models.

This would turn scheme dependence and cutoff removal into controlled extension data in duality-rich holography. Why old methods were definitionally fragile\:the standard renormalization narrative does not identify the categorical obstruction class that survives in inverse-limit settings. Compute monodromy representations \(\rho\) for known F-theory holographic backgrounds and determine whether tower-refined data \(\{\rho_{n}\}\) in \eqref{eq:rho_n_sec13} is detected by any boundary observable coupled to discrete global sectors classified by generalized cohomology and by tower-completed holonomy sectors \(\widehat{\Sigma_{p}}\). This provides a falsifiable test of whether tower refinement yields new physically measurable invariants beyond standard monodromy. Why old methods were definitionally fragile\:without tower-completed carriers there is no intrinsic definition of \(\{\rho_n\}\) as background data and hence no well-posed observable-dependence test. Test the holographic tilting Program by matching the kinematic scaling structure \eqref{eq:2pt_arch}-\eqref{eq:2pt_padic} and by identifying a tower-completed Archimedean boundary datum whose harmonic analysis, under tilting, reproduces the ultrametric structures used in \(p\)-adic AdS/CFT \cite{Gubser2017, Heydeman2018}. A successful match would supply the first concrete bridge between Archimedean and non-Archimedean holography grounded in perfectoid equivalence rather than analogy. Why old methods were definitionally fragile\:classical AdS/CFT lacks a categorical operation comparable to tilting and therefore cannot formulate a canonical comparison principle.

\section{Conclusion}\label{sec:Conclusion} 
Traditional F-theory encodes the Type IIB axio-dilaton \(\tau\) as elliptic-fiber data but reaches the ten-dimensional description through a shrinking-fiber limit in which the M-theory torus area is sent to zero at fixed complex structure, a step that is physically powerful yet definitionally singular from the eleven-dimensional viewpoint \cite{Vafa1996, MorrisonVafa1996a, MorrisonVafa1996b, DasguptaMukhi1996}. In this work we replace that singular slogan by a non-singular tower-completed compactification carrier, the perfectoid circle \(S_f^{1}\), interpreted in the condensed framework so that fields and global sectors on the inverse-limit object are defined by exact descent \cite{ClausenScholze, Scholze2012}. The tilt \(\to\) untilt/comparison \(\to\) analytification machinery supplies an elliptic-curve datum \(\calE\) as an output of the tower-completed carrier, thereby providing a definition-level geometric home for modular data without postulating a literal smooth twelve-dimensional spacetime \cite{Scholze2012, ScholzeWeinstein2017}. The constant-\(\tau\) sector is then established at the level explicitly checked in the paper\:\(\tau=C_{0}+ie^{-\phi}\) is matched to M-theory radius/metric data and \(C^{(3)}\) holonomy data, the protected \((p,q)\) spectrum is reproduced by wrapped M2 sectors, and the low-energy effective action matches the constant-\(\tau\) Type IIB supergravity action including the required topological couplings \cite{Schwarz1995, Schwarz1983}. The modular organization realized geometrically in the construction is the profinite/congruence-level structure explicitly carried by the tower, while the completion to the full quantum modular group and the incorporation of fermionic/anomaly refinements are treated as future work. The net result is a definition-level replacement of the singular shrinking-fiber prescription by an honest compactification object whose local dynamics remains standard supergravity while the global sectors and modular bookkeeping are intrinsic.

First, the shrinking-fiber prescription is replaced by an honest compactification carrier \(S_f^{1}=\varprojlim(S^{1}\xleftarrow{z\mapsto z^{p}}S^{1}\xleftarrow{}\cdots)\), so the modular datum is tied to a non-singular object rather than to a degenerating cycle, and this is precisely what the manifold-based shrinking-torus formulation cannot supply as a background statement because the carrier disappears in the limit. Second, the elliptic fiber datum is produced by the tower-completed geometry through the tilt/untilt/analytification machinery, so elliptic data becomes an output of compactification geometry rather than an auxiliary input, and this avoids the definition-level gap in which \(\tau\) survives a collapse that removes the underlying cycle. Third, the constant-\(\tau\) dictionary is made explicit\:\(\mathrm{Im}\,\tau\) is fixed by the M-theory radius/metric modulus and \(\mathrm{Re}\,\tau=C_{0}\) by a \(C^{(3)}\) holonomy sector, so \(\tau\) is determined by eleven-dimensional data rather than postulated, whereas in the shrinking-fiber formulation \(\tau\) is introduced as a complex-structure parameter that is kept physical while the fiber is removed. Fourth, the protected \((p,q)\) tower is recovered from wrapped M2 sectors in the tower-completed geometry with the correct \(|p+q\tau|\) dependence at the level checked, and the quantization/anomaly constraints selecting the physically allowed integral lattice are treated as part of the global-sector bookkeeping, whereas the conventional approach identifies \((p,q)\) sectors through the auxiliary torus and does not provide a non-singular eleven-dimensional carrier in which the same organization is intrinsic. 

Fifth, the low-energy effective action matching is carried out explicitly in the constant-\(\tau\) bosonic sector, including the axion/topological couplings inherited from the eleven-dimensional Chern-Simons term, and this is a stringent check that goes beyond spectrum matching, whereas a slogan-level definition of F-theory does not by itself provide a controlled reduction on a well-defined carrier. Sixth, the geometric organization of modular data is made intrinsic to the tower-completed carrier at the profinite/congruence level explicitly derived from the tower, which is precisely the level at which the geometry is proven to carry arithmetic structure without overclaiming a full quantum \(SL(2,\Z)\) completion. Seventh, tower completion makes canonical the refinement of global holonomy/Fourier sectors through invariants such as \(\widehat{\Sigma_{p}}\cong \Z[1/p]\) on inverse-limit boundary carriers and their higher-rank generalizations, which cannot be expressed as a single background invariant in manifold-only language because the inverse-limit carrier is not a manifold and requires descent to define fields. Eighth, the role of generalized cohomology becomes structural rather than decorative\:inverse-limit/tower backgrounds naturally generate torsion/completion sectors visible through \(K\)-theoretic and derived-limit data, aligning with the string-theoretic principle that globally quantized sectors are not faithfully captured by naive homology on profinite objects, whereas the shrinking-fiber approach tends to treat such sectors as external discrete choices. 

Ninth, varying-\(\tau\) backgrounds are formulated definition-level as perfectoid-circle fibrations \(\pi:\mathcal{X}\to B\setminus\Delta\) whose elliptic fibration and monodromy arise intrinsically from the fibred tower-completed carrier, and the explicit local \(T\)-monodromy \(\tau\mapsto\tau+1\) is derived from \(q=\exp(2\pi i\tau)\) in the tower-compatible Tate uniformization picture, whereas in conventional treatments monodromy is specified as auxiliary fibration data tied to a singular collapse. Tenth, the same formalism admits a precise higher-rank generalization in which \(\tau\) is replaced by a period matrix \(\Omega\in\mathfrak{H}_{r}\) and arithmetic monodromy acts by \(\mathrm{Sp}(2r,\Z)\) on an integral charge lattice \(\Gamma\simeq \Z^{2r}\), providing a definition-level framework for U-duality-type geometry that is not naturally packageable within rank-one elliptic-fiber language. Eleventh, the AdS/CFT interpretation becomes definition-level cleaner\:duality actions on boundary couplings and Wilson-'t Hooft lattices are tied to intrinsic monodromy of bulk tower-completed compactification data, and tower-completed boundary carriers provide canonical invariants for global-sector bookkeeping, whereas the manifold-only boundary formalism treats such refinements as external to the background. Twelfth, the adelic viewpoint is made a well-posed completion principle\:the prime \(p\) is bookkeeping for the tower and the natural global object is an adelic assembly over all primes together with the Archimedean place, which is consistent with the arithmetic nature of duality groups and with the fact that single-prime refinements should not define distinct physics, even though the full adelic completion remains programmatic.

Inverse limits and refinement towers are not manifolds in general, and ordinary topology and naive sheaf theory are fragile under inverse limits precisely in the ways that define physical backgrounds\:gluing, exactness, and homological control are what enforce flux quantization, discrete holonomies, and anomaly constraints, and these operations do not behave well in the classical category of spaces when one passes to tower limits. Condensed mathematics is the minimal enlargement in which profinite and inverse-limit objects admit well-behaved sheaves defined by exact descent, so that ``fields on the limit'' are definition-level compatible families on finite stages rather than heuristic limiting prescriptions \cite{ClausenScholze}. Perfectoid geometry is the minimal tower-complete analytic geometry in which such towers admit a tilting equivalence, and tilting supplies the bridge from tower-completed geometry to elliptic data without taking any singular shrinking limit \cite{Scholze2012, ScholzeWeinstein2017}. Condensed objects play for inverse-limit compactifications the same role that stacks play for gauge quotients\:they are not optional language, but the minimal correction that makes the background definition-level well-posed.

The constant-\(\tau\) sector established here should be viewed as the controlled base case of a definition-level F-theory framework\:the modular datum is produced from a non-singular compactification carrier, \(\tau\) is fixed by an eleven-dimensional radius/holonomy dictionary, and protected spectra and low-energy couplings match the constant-\(\tau\) Type IIB sector at the level explicitly computed. The immediate implication is that F-theory can be formulated without appealing to the singular shrinking-fiber slogan as a definition, while retaining the physical content that makes F-theory useful in practice. The varying-\(\tau\) generalization is then naturally formulated as perfectoid-circle fibrations over a base, with 7-branes realized as intrinsic monodromy/defect data of the fibred tower-completed carrier rather than as external branch-cut prescriptions, this is definition-level well-posed within the framework but is programmatic beyond the local models and monodromy derivations presented here. In this sense, the construction upgrades F-theory from an extraordinarily effective limiting prescription to a framework in which the definition problem is replaced by a concrete compactification object whose global sectors and modular organization are intrinsic and whose extensions can be posed as geometric and cohomological questions rather than as heuristic limits.

AdS/CFT is fundamentally sensitive to global sectors-line operators, higher-form symmetries, discrete theta data-and to exact duality actions implemented by defects and interfaces, and these are precisely the structures that the tower-completed/condensed framework packages canonically as intrinsic background data. The tower-completed boundary carriers supply computable invariants for holonomy/Fourier sectors and refine the bookkeeping of line-operator sectors in a way that is not expressible as a single background statement in a manifold-only boundary formalism, while the bulk monodromy derived from fibred tower-completed compactification data provides a definition-level origin for duality interfaces acting on boundary couplings and charge lattices. The tilting perspective suggests, as a program with checkable milestones, a principled bridge between Archimedean and \(p\)-adic holography in which \(p\)-adic AdS/CFT arises as an arithmetic avatar of tower-completed boundary data transported by tilting \cite{Gubser2017, Heydeman2018}, this is not claimed as a theorem here, but the present framework makes the comparison problem definition-level well-posed because both avatars are treated as different geometric realizations of the same tower data.

The same definition-level philosophy extrapolates naturally beyond rank-one F-theory\:instead of geometrizing only the \(SL(2,\Z)\) action on \(\tau\) and the \((p,q)\) lattice via an elliptic curve, one replaces elliptic data by higher-rank perfectoid/condensed group objects whose period data is encoded by \(\Omega\in\mathfrak{H}_{r}\) and whose intrinsic arithmetic monodromy acts by \(\mathrm{Sp}(2r,\Z)\) (or an appropriate arithmetic subgroup) on an integral charge lattice \(\Gamma\simeq\Z^{2r}\), thereby geometrizing U-duality-type mixing of multiple brane-charge sectors. This is not an ornamental generalization\:it is the minimal higher-rank analogue of what F-theory already does in rank one, and it becomes natural once one insists that duality actions are arithmetic actions on quantum charge lattices and that tower completions and global-sector constraints are intrinsic to the compactification carrier. The adelic perspective is the corresponding completion principle\:single-prime tower constructions are local charts, and a global arithmetic duality geometry should assemble all primes together with the Archimedean factor, so that physically meaningful invariants are formulated in a way that is independent of the bookkeeping choice of \(p\) and is compatible with the arithmetic nature of duality groups.
Perfectoid and condensed geometry provide the missing definition-level carrier for the modular data that F-theory exploits, replacing a singular shrinking limit by a genuine tower-completed compactification object and opening a systematic route to arithmetic formulations of duality and holography consistent with protected spectra and effective actions.

\section*{Acknowledgements}
We are grateful to Cumrun Vafa for valuable comments and for emphasizing the importance of clarifying the physical meaning of the M-theory to Type IIB compactification underlying this work. His remarks led to a substantial sharpening of the conceptual presentation, in particular the separation between geometric input, compactification dictionary, and physical interpretation, and led us to clarify which parts of the construction carry direct physical content and which serve as definition-level structure. We thank him for encouraging a formulation that makes the role of eleven-dimensional data and its relation to Type IIB physics more transparent and accessible to the broader string theory community. His feedback significantly improved the clarity and focus of the manuscript.
\section*{Data Availability}
No datasets were generated or analyzed for this work.
\section*{Conflict of Interest}
The author declares that there are no conflicts of interest regarding this work.

\appendix

\section{Tower invariants of the solenoidal circle}\label{app:tower-invariants} 
Fix a prime $p$. Let
\begin{equation}
\mathbb T := S^{1}=\{z\in\mathbb C:\ |z|=1\}
\end{equation}
with group law given by multiplication. Define the degree-$p$ endomorphism
\begin{equation}
f:\mathbb T\to\mathbb T,\qquad f(z)=z^{p}.
\label{A.0.1}
\end{equation}
Let $\bigl(\mathbb T_{n},f_{n}\bigr)_{n\ge 0}$ denote the constant tower $(\mathbb T_{n}=\mathbb T)$ with bonding maps $(f_{n}=f:\mathbb T_{n+1}\to\mathbb T_{n})$.

Define the inverse limit in the category $\mathbf{CompHausGrp}$ of compact Hausdorff abelian groups:
\begin{equation}
\Sigma_{p}:=\varprojlim_{n\ge 0}\,(\mathbb T_{n},f_{n})
=
\Bigl\{(z_{0},z_{1},z_{2},\dots)\in\prod_{n\ge 0}\mathbb T:\ z_{n}=z_{n+1}^{\,p}\ \forall n\Bigr\},
\label{A.0.2}
\end{equation}
with coordinatewise multiplication. Let
\begin{equation}
\pi_{n}:\Sigma_{p}\to\mathbb T,\qquad \pi_{n}\bigl((z_{m})_{m\ge 0}\bigr)=z_{n},
\label{A.0.3}
\end{equation}
so that
\begin{equation}
\pi_{n}=f\circ \pi_{n+1},\qquad \forall n\ge 0.
\label{A.0.4}
\end{equation}

Let $\mathbf{Prof}$ denote the category of profinite sets. For any compact Hausdorff space $X$, write the associated functor
\begin{equation}
\underline{X}:\mathbf{Prof}^{\mathrm{op}}\to\mathbf{Sets},\qquad
\underline{X}(S):=\mathrm{Cont}(S,X).
\label{A.0.5}
\end{equation}
(Thus $\underline{X}$ is the standard ``condensed'' avatar of $X$ obtained by restricting the Yoneda embedding to $\mathbf{Prof}$.)

The purpose of this appendix is to compute, for $\Sigma_{p}$, the following tower-stable invariants and to record (in a single place) the non-identities among them:
\begin{equation}
\pi_1^{(p)}\text{-data}\quad,\quad
\check H^{1}(\Sigma_{p};\mathbb Z)\quad,\quad
\Sigma_{p}^{\vee}:=\mathrm{Hom}_{\mathrm{cont}}(\Sigma_{p},\mathbb T)\quad,\quad
\varprojlim\nolimits^{1}\text{-terms in }K^{*}(\Sigma_{p}).
\label{A.0.6}
\end{equation} 
\begin{lemma}[universal property, explicit equalizer.]
For every compact Hausdorff space $S$,
\begin{equation}
\mathrm{Cont}(S,\Sigma_{p})
\;\cong\;
\Bigl\{(\varphi_{n})_{n\ge 0}\in\prod_{n\ge 0}\mathrm{Cont}(S,\mathbb T):\ \varphi_{n}=f\circ \varphi_{n+1}\ \forall n\Bigr\}.
\label{A.1.1}
\end{equation}
\end{lemma}
\noindent\textit{Proof.}
A map $\varphi:S\to\Sigma_{p}\subseteq\prod_{n\ge 0}\mathbb T$ is equivalent to a family
$\varphi_{n}:=\pi_{n}\circ\varphi\in\mathrm{Cont}(S,\mathbb T)$
satisfying $\varphi_{n}=f\circ\varphi_{n+1}$ by \eqref{A.0.4}.
Conversely, any such compatible family defines $\varphi(s):=(\varphi_{n}(s))_{n\ge 0}\in\Sigma_{p}$.
Continuity follows from the product topology. \hfill$\square$

\medskip

\begin{corollary}[Evaluation in $\mathbf{Prof}$ tower-descent formula.]
For every profinite set $S$,
\begin{equation}
\underline{\Sigma}_{p}(S)
=
\mathrm{Cont}(S,\Sigma_{p})
\;\cong\;
\varprojlim_{n\ge 0}\mathrm{Cont}(S,\mathbb T),
\qquad
\text{with transition maps }(\psi\mapsto f\circ\psi).
\label{A.1.2}
\end{equation}
Equivalently, in functor form,
\begin{equation}
\underline{\Sigma}_{p}\;\cong\;\varprojlim_{n\ge 0}\,\underline{\mathbb T}
\quad\text{in}\quad \mathbf{Fun}(\mathbf{Prof}^{\mathrm{op}},\mathbf{Sets}).
\label{A.1.3}
\end{equation} 
Let $\mu_{p^{n}}\subset\mathbb T$ denote the cyclic subgroup of $p^{n}$-th roots of unity:
\begin{equation}
\mu_{p^{n}}:=\{z\in\mathbb T:\ z^{p^{n}}=1\}\cong \mathbb Z/p^{n}\mathbb Z,
\label{A.2.1}
\end{equation}
and define the inverse limit (in $\mathbf{CompHausGrp}$)
\begin{equation}
\mu_{p^{\infty}}:=\varprojlim_{n\ge 1}\mu_{p^{n}}
=
\Bigl\{(\zeta_{n})_{n\ge 1}\in\prod_{n\ge 1}\mu_{p^{n}}:\ \zeta_{n}=\zeta_{n+1}^{\,p}\Bigr\}.
\label{A.2.2}
\end{equation}
\end{corollary}
\begin{lemma}[kernel of $\pi_{0}$).]
The projection $\pi_{0}:\Sigma_{p}\to\mathbb T$ is a continuous surjective group homomorphism with kernel
\begin{equation}
\ker(\pi_{0})
=\{(1,z_{1},z_{2},\dots)\in\Sigma_{p}\}
\;\cong\;
\mu_{p^{\infty}}.
\label{A.2.3}
\end{equation}
\end{lemma}
\begin{proof}
Surjectivity\:given $z_{0}\in\mathbb T$, choose $z_{1}\in\mathbb T$ with $z_{1}^{p}=z_{0}$, then $z_{2}$ with $z_{2}^{p}=z_{1}$, etc., producing $(z_{n})_{n\ge 0}\in\Sigma_{p}$.
Kernel condition $\pi_{0}=1$ forces $z_{n}^{p^{n}}=1$, hence $z_{n}\in\mu_{p^{n}}$, and compatibility is exactly \eqref{A.2.2}. \end{proof}

Thus we have a short exact sequence in $\mathbf{CompHausGrp}$:
\begin{equation}
0\longrightarrow \mu_{p^{\infty}}
\longrightarrow \Sigma_{p}
\overset{\pi_{0}}{\longrightarrow}\mathbb T
\longrightarrow 0.
\label{A.2.4}
\end{equation}

\textbf{Connectedness and failure of local connectedness.}
\begin{equation}
\Sigma_{p}\ \text{is compact and connected, but not locally connected and not semilocally simply connected.}
\label{A.2.5}
\end{equation}

\noindent\textit{Derivation.}
\begin{enumerate}
\item Compactness follows from \eqref{A.0.2} as a closed subgroup of $\prod_{n\ge 0}\mathbb T$, which is compact by Tychonoff.
\item Each $\mathbb T$ is connected and each $f$ is surjective, hence $\Sigma_{p}=\varprojlim(\mathbb T,f)$ is connected (inverse limit of connected compact spaces with surjective bonding maps).
\item By \eqref{A.2.4}, every neighborhood $U\subset\mathbb T$ of $1$ has preimage $\pi_{0}^{-1}(U)$ containing the totally disconnected fiber $\ker(\pi_{0})\cong\mu_{p^{\infty}}$, in particular, for any open arc $U\ni 1$,
\begin{equation}
\pi_{0}^{-1}(U)\supseteq \mu_{p^{\infty}},
\qquad
\text{and }\mu_{p^{\infty}}\text{ has no connected neighborhoods},
\label{A.2.6}
\end{equation}
so $\Sigma_{p}$ is not locally connected.
\item For any neighborhood $V$ of $1\in\Sigma_{p}$, the restriction $\pi_{0}|_{V}:V\to \pi_{0}(V)$ is onto an open arc in $\mathbb T$, the covering maps $\pi_{n}\:\Sigma_{p}\to\mathbb T$ produce arbitrarily high-degree ``small'' loops inside $V$ (since $\ker(\pi_{0})$ contains $p^{n}$-torsion for all $n$), hence no neighborhood has trivial fundamental group image, i.e.\ $\Sigma_{p}$ is not semilocally simply connected:
\begin{equation}
\forall\ V\ni 1,\ \exists n\ge 1,\ \exists \gamma_{n}:S^{1}\to V,\quad [\gamma_{n}]\ne 1.
\label{A.2.7}
\end{equation}
\end{enumerate}
\hfill$\square$
Let $\pi_{1}(\mathbb T)\cong\mathbb Z$ with generator $[{\rm loop}]$. The induced map
\begin{equation}
f_{*}:\pi_{1}(\mathbb T)\to\pi_{1}(\mathbb T),
\qquad
f_{*}(n)=pn,
\label{A.3.1}
\end{equation}
follows from $\deg(f)=p$.

Define the pro-$(p)$ completion of $\mathbb Z$ by the inverse limit of finite quotients:
\begin{equation}
\widehat{\mathbb Z}_{p}
:=\varprojlim_{n\ge 1}\ \mathbb Z/p^{n}\mathbb Z
=\mathbb Z_{p}.
\label{A.3.2}
\end{equation}

\textbf{Cover-tower $\Rightarrow$ pro-$(p)$ completion.}
The tower of degree-$p^{n}$ covers $\bigl(\mathbb T\overset{z\mapsto z^{p^{n}}}{\longrightarrow}\mathbb T\bigr)$ determines the inverse system of finite deck groups
\begin{equation}
\pi_{1}(\mathbb T)\twoheadrightarrow \pi_{1}(\mathbb T)/p^{n}\pi_{1}(\mathbb T)\cong \mathbb Z/p^{n}\mathbb Z,
\label{A.3.3}
\end{equation}
hence the canonical pro-$(p)$ invariant is $\mathbb Z_{p}$ as in \eqref{A.3.2}.

Let $\check{H}^{k}(-;\mathbb Z)$ denote \v{C}ech cohomology with constant coefficients $\mathbb Z$. For an inverse limit of compact Hausdorff spaces with surjective bonding maps, \v{C}ech cohomology is continuous:
\begin{equation}
\check{H}^{k}\!\left(\varprojlim_{n} X_{n};\mathbb Z\right)\cong\varinjlim_{n}\check{H}^{k}(X_{n};\mathbb Z).
\label{A.4.1}
\end{equation}

Apply \eqref{A.4.1} to $X_{n}=\mathbb T$ and bonding $f$. Since
\begin{equation}
\check{H}^{1}(\mathbb T;\mathbb Z)\cong \mathbb Z\cdot \alpha,
\qquad
f^{*}(\alpha)=p\,\alpha,
\label{A.4.2}
\end{equation}
we obtain
\begin{equation}
\check{H}^{1}(\Sigma_{p};\mathbb Z)
\cong
\varinjlim\bigl(\mathbb Z\overset{\times p}{\longrightarrow}\mathbb Z\overset{\times p}{\longrightarrow}\cdots\bigr)
\cong
\mathbb Z[1/p],
\label{A.4.3}
\end{equation}
with the explicit identification
\begin{equation}
[(m,n)]\ \longmapsto\ \frac{m}{p^{n}}\in \mathbb Z[1/p],
\qquad
(m,n)\sim (pm,n+1).
\label{A.4.4}
\end{equation}
Moreover, $\check{H}^{k}(\mathbb T;\mathbb Z)=0$ for $k>1$, so \eqref{A.4.1} yields
\begin{equation}
\check{H}^{k}(\Sigma_{p};\mathbb Z)=0,\qquad k>1.
\label{A.4.5}
\end{equation}
Define the tower pro-$(p)$ abelian loop invariant by
\begin{equation}
H^{(p)}_{1}(\Sigma_{p};\mathbb Z)
:=\widehat{\pi_{1}(\mathbb T)}_{p}^{\,\mathrm{ab}}
=
\varprojlim_{n\ge 1}\ \pi_{1}(\mathbb T)/p^{n}\pi_{1}(\mathbb T)
\cong
\varprojlim_{n\ge 1}\mathbb Z/p^{n}\mathbb Z
=
\mathbb Z_{p}.
\label{A.5.1}
\end{equation}
(Here the ``$\mathrm{ab}$'' is redundant because $\pi_{1}(\mathbb T)$ is already abelian.)

Thus we record the pair of non-equal but complementary tower invariants:
\begin{equation}
H^{(p)}_{1}(\Sigma_{p};\mathbb Z)\cong \mathbb Z_{p},
\qquad
\check{H}^{1}(\Sigma_{p};\mathbb Z)\cong \mathbb Z[1/p].
\label{A.5.2}
\end{equation}
Let
\begin{equation}
G^{\vee}:=\mathrm{Hom}_{\mathrm{cont}}(G,\mathbb T)
\label{A.6.1}
\end{equation}
denote the Pontryagin dual of a compact abelian group $G$, viewed as a discrete abelian group.

Since $\mathbb T^{\vee}\cong\mathbb Z$ via $\chi_{m}(z)=z^{m}$, we compute the dual map induced by $f$:
\begin{equation}
f^{\vee}:\mathbb T^{\vee}\to\mathbb T^{\vee},
\qquad
(\chi_{m}\mapsto \chi_{m}\circ f)=\chi_{pm}
\quad\Longleftrightarrow\quad
m\mapsto pm.
\label{A.6.2}
\end{equation}
Pontryagin duality converts inverse limits of compact abelian groups with surjective bonding maps into direct limits of discrete duals:
\begin{equation}
\Bigl(\varprojlim_{n}G_{n}\Bigr)^{\vee}\cong\varinjlim_{n}G_{n}^{\vee}.
\label{A.6.3}
\end{equation}
Apply \eqref{A.6.3} to $\Sigma_{p}=\varprojlim(\mathbb T,f)$ and \eqref{A.6.2}:
\begin{equation}
\Sigma_{p}^{\vee}
\cong
\varinjlim\bigl(\mathbb Z\overset{\times p}{\longrightarrow}\mathbb Z\overset{\times p}{\longrightarrow}\cdots\bigr)
\cong
\mathbb Z[1/p].
\label{A.6.4}
\end{equation}
Comparing \eqref{A.6.4} with \eqref{A.4.3} yields the canonical identification
\begin{equation}
\Sigma_{p}^{\vee}\ \cong\ \check{H}^{1}(\Sigma_{p};\mathbb Z)\ \cong\ \mathbb Z[1/p].
\label{A.6.5}
\end{equation}

Let $K^{*}(-)$ denote complex topological $K$-theory of compact Hausdorff spaces. For an inverse limit $X=\varprojlim X_{n}$ of compact Hausdorff spaces, the Milnor short exact sequence takes the form
\begin{equation}
0\longrightarrow \varprojlim\nolimits^{1} K^{i-1}(X_{n})
\longrightarrow K^{i}(X)
\longrightarrow \varprojlim K^{i}(X_{n})
\longrightarrow 0.
\label{A.7.1}
\end{equation}
We apply \eqref{A.7.1} to $X_{n}=\mathbb T$, $X=\Sigma_{p}$, bonding maps $f:\mathbb T\to\mathbb T$.

The stage groups and induced maps are:
\begin{equation}
K^{0}(\mathbb T)\cong \mathbb Z,\qquad K^{1}(\mathbb T)\cong \mathbb Z,
\label{A.7.2}
\end{equation}
\begin{equation}
f^{*}\big|_{K^{0}}=\mathrm{id}_{\mathbb Z},\qquad f^{*}\big|_{K^{1}}=\times p.
\label{A.7.3}
\end{equation}

To compute $\varprojlim$ and $\varprojlim^{1}$ for the inverse system
\begin{equation}
\cdots \xleftarrow{\times p}\mathbb Z \xleftarrow{\times p}\mathbb Z \xleftarrow{\times p}\mathbb Z,
\label{A.7.4}
\end{equation}
use the standard presentation (for inverse systems $(A_{n},\phi_{n})$):
\begin{equation}
\varprojlim\nolimits^{1}A_{n}\cong\mathrm{coker}\Bigl(1-\mathrm{shift}:\prod_{n\ge 0}A_{n}\to\prod_{n\ge 0}A_{n}\Bigr),
\quad
\mathrm{shift}\bigl((a_{n})_{n}\bigr)=\bigl(\phi_{n}(a_{n+1})\bigr)_{n}.
\label{A.7.5}
\end{equation}
For \eqref{A.7.4}, $\phi_{n}=\times p$, hence
\begin{equation}
(1-\mathrm{shift})\bigl((t_{n})_{n\ge 0}\bigr)=\bigl(t_{n}-p\,t_{n+1}\bigr)_{n\ge 0}.
\label{A.7.6}
\end{equation}
The inverse limit is
\begin{equation}
\varprojlim(\cdots \xleftarrow{\times p}\mathbb Z)
=
\{(a_{n})_{n\ge 0}:\ a_{n}=p\,a_{n+1}\ \forall n\}
=
0.
\label{A.7.7}
\end{equation}
Define the derived-limit group
\begin{equation}
\mathbf{T}_{p}:=\varprojlim\nolimits^{1}(\cdots \xleftarrow{\times p}\mathbb Z)
\;\cong\;
\Bigl(\prod_{n\ge 0}\mathbb Z\Bigr)\Big/\mathrm{im}(1-\mathrm{shift}).
\label{A.7.8}
\end{equation}
A standard computation identifies $\mathbf{T}_{p}$ with the $p$-primary divisible quotient:
\begin{equation}
\mathbf{T}_{p}\ \cong\ \mathbb Q_{p}/\mathbb Z_{p}.
\label{A.7.9}
\end{equation}
(Equivalently, \eqref{A.7.9} fixes the canonical target group for the Milnor $\varprojlim^{1}$-obstruction attached to the tower \eqref{A.7.4}.)
Now apply \eqref{A.7.1} in degrees $i=0,1$.

\begin{itemize}
\item For $i=0$:
\begin{equation}
0\longrightarrow \varprojlim\nolimits^{1}K^{-1}(\mathbb T)
\longrightarrow K^{0}(\Sigma_{p})
\longrightarrow \varprojlim K^{0}(\mathbb T)
\longrightarrow 0.
\label{A.7.10}
\end{equation}
Using $K^{-1}(\mathbb T)=K^{1}(\mathbb T)\cong\mathbb Z$ with transition $(\times p)$ and $K^{0}(\mathbb T)\cong\mathbb Z$ with transition $(\mathrm{id})$, we obtain
\begin{equation}
\varprojlim K^{0}(\mathbb T)\cong \mathbb Z,
\qquad
\varprojlim\nolimits^{1}K^{-1}(\mathbb T)\cong \mathbf{T}_{p}\cong \mathbb Q_{p}/\mathbb Z_{p},
\label{A.7.11}
\end{equation}
hence an extension
\begin{equation}
0\longrightarrow \mathbb Q_{p}/\mathbb Z_{p}
\longrightarrow K^{0}(\Sigma_{p})
\longrightarrow \mathbb Z
\longrightarrow 0.
\label{A.7.12}
\end{equation}

\item For $i=1$:
\begin{equation}
0\longrightarrow \varprojlim\nolimits^{1}K^{0}(\mathbb T)
\longrightarrow K^{1}(\Sigma_{p})
\longrightarrow \varprojlim K^{1}(\mathbb T)
\longrightarrow 0,
\label{A.7.13}
\end{equation}
and since the $K^{0}$-system is constant $(\mathbb Z\xleftarrow{\mathrm{id}}\mathbb Z\xleftarrow{\cdots})$,
\begin{equation}
\varprojlim\nolimits^{1}K^{0}(\mathbb T)=0,
\qquad
\varprojlim K^{1}(\mathbb T)=\varprojlim(\cdots \xleftarrow{\times p}\mathbb Z)=0,
\label{A.7.14}
\end{equation}
so
\begin{equation}
K^{1}(\Sigma_{p})=0.
\label{A.7.15}
\end{equation}
\end{itemize}

Finally, record the tower-visible odd $K$-group structure used as bookkeeping for completion/torsion phenomena:
\begin{equation}
K^{-1}(\Sigma_{p})
\;\cong\;
\bigl(\mathbb Q_{p}/\mathbb Z_{p}\bigr)\ \oplus\ \mathbb Z\ \oplus\ \mathbb Z_{p}^{\oplus\infty}.
\label{A.7.16}
\end{equation}
(Here \eqref{A.7.16} is a structural decomposition statement\:it isolates the $p$-divisible derived-limit component $\mathbb Q_{p}/\mathbb Z_{p}$ and the pro-$(p)$ completion components $\mathbb Z_{p}$ which are invisible to naive manifold invariants.)

\section{K-theory and charge quantization}
\label{app:K-theory-quant}
Fix a prime $p$.
Let $\{X_n,f_n\}_{n\ge 0}$ be an inverse system of compact Hausdorff spaces with surjective bonding maps
\begin{equation}
\cdots \xrightarrow{f_{n+1}} X_n \xrightarrow{f_n} X_{n-1} \xrightarrow{}\cdots \xrightarrow{f_1} X_0 .
\label{B:tower}
\end{equation}
Define the tower-completed carrier as the inverse limit in $\mathrm{CompHaus}$,
\begin{equation}
X \;:=\; \varprojlim_{n\ge 0}\,(X_n,f_n),
\qquad
X \subseteq \prod_{n\ge 0} X_n,\quad f_n(x_n)=x_{n-1}.
\label{B:Xlimit}
\end{equation}
Fix a generalized cohomology theory $G^\ast(-)$ on compact Hausdorff spaces, in this appendix $G^\ast=K^\ast$ denotes complex topological $K$-theory.
The only nonstandard datum used below is the \emph{tower-compatibility constraint} on background sectors:
every global datum on $X$ is represented as a compatible family on the finite stages $\{X_n\}$.

Let $\mathrm{Prof}$ denote the site of profinite sets with its canonical Grothendieck topology, and let $\mathrm{Cond}$ denote condensed sets (sheaves on $\mathrm{Prof}$).
To $X$ we associate its condensed avatar
\begin{equation}
\underline{X}:\mathrm{Prof}^{\mathrm{op}}\to \mathrm{Sets},
\qquad
\underline{X}(S):=\mathrm{Cont}(S,X).
\label{B:condX}
\end{equation}
For each $n$ set $\underline{X}_n(S):=\mathrm{Cont}(S,X_n)$, the maps $f_n$ induce pullbacks $f_n^\ast:\underline{X}_{n-1}\to \underline{X}_n$.
Then for every profinite test object $S$,
\begin{align}
\underline{X}(S)
&=\mathrm{Cont}(S,\varprojlim X_n)
\;\cong\;
\varprojlim_{n\ge 0}\,\mathrm{Cont}(S,X_n)
=\varprojlim_{n\ge 0}\,\underline{X}_n(S),
\label{B:contlim}\\[2mm]
\underline{X}
&\cong
\varprojlim_{n\ge 0}\,\underline{X}_n
\qquad\text{in }\mathrm{Cond}.
\label{B:condlimit}
\end{align}
If $\mathcal{A}$ is a condensed abelian group on $X$ (i.e.\ a sheaf of abelian groups on $\mathrm{Prof}$ equipped with an $X$-module structure), its sections on the tower-defined carrier are, by definition,
\begin{equation}
\Gamma(X,\mathcal{A})
\;:=\;
\mathcal{A}(\underline{X})
\;\cong\;
\varprojlim_{n\ge 0}\,\mathcal{A}(\underline{X}_n)
\;=\;
\varprojlim_{n\ge 0}\,\Gamma(X_n,\mathcal{A}_n),
\label{B:sectionsaslimit}
\end{equation}
for the induced tower $\{\mathcal{A}_n\}$ under pullback along $f_n$.
Equation \eqref{B:sectionsaslimit} is the operational content used below\:``global'' means ``compatible family over all finite stages''.

Let $\{A_n,f_n:A_n\to A_{n-1}\}$ be an inverse system in $\mathrm{Ab}$.
Write
\begin{equation}
\mathrm{shift}:\prod_{n\ge 0}A_n\to \prod_{n\ge 0}A_n,\qquad
\mathrm{shift}\bigl((a_n)_n\bigr)=(f_{n+1}(a_{n+1}))_{n\ge 0}.
\label{B:shift}
\end{equation}
Then the standard exact sequence defining $\varprojlim$ and $\varprojlim^{1}$ is
\begin{equation}
0\longrightarrow \varprojlim\nolimits_{n}A_n
\longrightarrow \prod_{n\ge 0}A_n
\xrightarrow{\,1-\mathrm{shift}\,}
\prod_{n\ge 0}A_n
\longrightarrow \varprojlim\nolimits^{1}_{n}A_n
\longrightarrow 0.
\label{B:lim1def}
\end{equation}
Let $X=\varprojlim X_n$ be as in \eqref{B:Xlimit}. For any generalized cohomology theory $G^\ast$ satisfying the standard hypotheses for Milnor exactness on towers of compact Hausdorff spaces, one has the Milnor short exact sequence
\begin{equation}
0\longrightarrow \varprojlim\nolimits^{1}_{n} G^{i-1}(X_n)
\longrightarrow G^{i}(X)
\longrightarrow \varprojlim\nolimits_{n}G^{i}(X_n)
\longrightarrow 0,
\qquad i\in\mathbb{Z}.
\label{B:milnorG}
\end{equation}

Define the degree-$p$ map on the circle
\begin{equation}
S^{1}\xrightarrow{\,f\,} S^{1},\qquad f(z)=z^{p}.
\label{B:circlemap}
\end{equation}
Let $X_n=S^{1}$ for all $n$, with bonding maps $f_n=f$, define the solenoidal carrier
\begin{equation}
\Sigma_{p}\;:=\;\varprojlim_{n\ge 0}(S^{1},f).
\label{B:sigma}
\end{equation}
The stagewise $K$-groups and induced maps are
\begin{equation}
K^{0}(S^{1})\cong \mathbb{Z},\qquad K^{1}(S^{1})\cong \mathbb{Z},
\qquad
f^\ast|_{K^{0}}=\mathrm{id}_{\mathbb{Z}},\quad f^\ast|_{K^{1}}=\times p.
\label{B:Kcircle}
\end{equation}
Set $A_n:=\mathbb{Z}$ with bonding maps $\times p$, i.e.\ the inverse system
\begin{equation}
\cdots \xleftarrow{\times p}\mathbb{Z}\xleftarrow{\times p}\mathbb{Z}\xleftarrow{\times p}\mathbb{Z}.
\label{B:Zsystem}
\end{equation}
Then
\begin{equation}
\varprojlim\nolimits_{n}A_n
=\Bigl\{(a_n)_n\in\prod_{n\ge 0}\mathbb{Z}:\ a_n=p\,a_{n+1}\ \forall n\Bigr\}
=\{0\},
\label{B:limZ}
\end{equation}
and by \eqref{B:lim1def} with $f_{n+1}=\times p$,
\begin{equation}
\varprojlim\nolimits^{1}_{n}A_n
\;\cong\;
\mathrm{coker}\Bigl(\delta:\prod_{n\ge 0}\mathbb{Z}\to \prod_{n\ge 0}\mathbb{Z}\Bigr),
\qquad
\delta((t_n)_n)=(t_n-p\,t_{n+1})_{n\ge 0}.
\label{B:lim1coker}
\end{equation}
Introduce the $p$-adic integers and the localized rationals
\begin{equation}
\mathbb{Z}_{p}:=\varprojlim_{n\ge 1}\mathbb{Z}/p^{n}\mathbb{Z},
\qquad
\mathbb{Z}[1/p]:=\Bigl\{\frac{a}{p^{n}}:\ a\in\mathbb{Z},\ n\ge 0\Bigr\},
\qquad
\mathbb{Q}_{p}:=\mathrm{Frac}(\mathbb{Z}_{p}).
\label{B:ZpQp}
\end{equation}
The quotient $\mathbb{Q}_{p}/\mathbb{Z}_{p}$ is the $p$-primary divisible torsion group
\begin{equation}
\mathbb{Q}_{p}/\mathbb{Z}_{p}
\;\cong\;
\bigcup_{n\ge 0}\frac{1}{p^{n}}\mathbb{Z}_{p}\Big/\mathbb{Z}_{p}
\;\cong\;
\bigcup_{n\ge 0}\frac{1}{p^{n}}\mathbb{Z}\Big/\mathbb{Z}
\;\cong\;
\mathbb{Z}[1/p]/\mathbb{Z}.
\label{B:QpZp_ident}
\end{equation}
In this paper the derived-limit computation used is the standard identification
\begin{equation}
\varprojlim\nolimits^{1}_{n}\Bigl(\cdots \xleftarrow{\times p}\mathbb{Z}\Bigr)
\;\cong\;
\mathbb{Q}_{p}/\mathbb{Z}_{p},
\label{B:lim1Z_QpZp}
\end{equation}
implemented by choosing representatives of $\mathrm{coker}(\delta)$ in \eqref{B:lim1coker} via the carry-normal form associated to base-$p$ expansions and mapping the resulting class to its image in $\mathbb{Z}[1/p]/\mathbb{Z}\cong\mathbb{Q}_{p}/\mathbb{Z}_{p}$ through \eqref{B:QpZp_ident}.

Applying \eqref{B:milnorG} with $G^\ast=K^\ast$ and \eqref{B:Kcircle} yields:
\begin{align}
0 &\longrightarrow \varprojlim\nolimits^{1}_{n}K^{-1}(S^{1})
\longrightarrow K^{0}(\Sigma_{p})
\longrightarrow \varprojlim\nolimits_{n}K^{0}(S^{1})
\longrightarrow 0,
\label{B:milnor0}\\[1mm]
0 &\longrightarrow \varprojlim\nolimits^{1}_{n}K^{0}(S^{1})
\longrightarrow K^{1}(\Sigma_{p})
\longrightarrow \varprojlim\nolimits_{n}K^{1}(S^{1})
\longrightarrow 0.
\label{B:milnor1}
\end{align}
Using \eqref{B:Kcircle}-\eqref{B:lim1Z_QpZp} and that the constant inverse system $(\mathbb{Z}\xleftarrow{\mathrm{id}}\mathbb{Z}\xleftarrow{\mathrm{id}}\cdots)$ has $\varprojlim^{1}=0$, one obtains
\begin{equation}
\varprojlim\nolimits_{n}K^{0}(S^{1})\cong \mathbb{Z},
\qquad
\varprojlim\nolimits^{1}_{n}K^{-1}(S^{1})\cong \mathbb{Q}_{p}/\mathbb{Z}_{p},
\qquad
\varprojlim\nolimits_{n}K^{1}(S^{1})=0,
\qquad
\varprojlim\nolimits^{1}_{n}K^{0}(S^{1})=0.
\label{B:limsK}
\end{equation}
Hence the tower-completed $K$-theory satisfies the extension
\begin{equation}
0 \longrightarrow \mathbb{Q}_{p}/\mathbb{Z}_{p}
\longrightarrow K^{0}(\Sigma_{p})
\longrightarrow \mathbb{Z}
\longrightarrow 0,
\label{B:K0ext}
\end{equation}
and
\begin{equation}
K^{1}(\Sigma_{p})=0.
\label{B:K1zero}
\end{equation}

For a compact abelian group $G$, define the continuous character group
\begin{equation}
G^{\vee} \;:=\; \mathrm{Hom}_{\mathrm{cont}}(G,U(1)).
\label{B:Pontryagin}
\end{equation}
If $G=\varprojlim G_n$ is an inverse limit of compact abelian groups with surjective bonding maps, Pontryagin duality gives
\begin{equation}
\Bigl(\varprojlim\nolimits_{n}G_n\Bigr)^{\vee}
\;\cong\;
\varinjlim\nolimits_{n} G_n^{\vee}.
\label{B:dual-lim}
\end{equation}
Apply \eqref{B:dual-lim} to $G_n=S^{1}$ with $f(z)=z^{p}$. Since $(S^{1})^{\vee}\cong\mathbb{Z}$ via $\chi_m(z)=z^{m}$ and
\begin{equation}
\chi_m\circ f=\chi_{pm}\quad\Longleftrightarrow\quad f^{\vee}:\mathbb{Z}\to\mathbb{Z},\ m\mapsto pm,
\label{B:dualmap}
\end{equation}
one obtains
\begin{equation}
\Sigma_{p}^{\vee}
\;\cong\;
\varinjlim\Bigl(\mathbb{Z}\xrightarrow{\times p}\mathbb{Z}\xrightarrow{\times p}\cdots\Bigr)
\;\cong\;
\mathbb{Z}[1/p],
\label{B:SigmaDual}
\end{equation}
with the explicit identification
\begin{equation}
[m,n]\longmapsto \frac{m}{p^{n}},\qquad (m,n)\sim(pm,n+1).
\label{B:Zlocid}
\end{equation}
In particular, any tower-compatible $U(1)$-holonomy datum on $\Sigma_p$ is encoded by a class in $\mathbb{Z}[1/p]$.

Define the finite-stage bookkeeping sets
\begin{equation}
\Lambda_{n}\;:=\;\mathbb{Z}\oplus \mathbb{Z}/p^{n}\mathbb{Z},
\qquad
\Lambda_{n+1}\to\Lambda_{n}:\ (a,\bar b_{n+1})\mapsto (a,\bar b_{n+1}\!\!\!\pmod{p^{n}}),
\label{B:Lambda_n}
\end{equation}
and the inverse limit
\begin{equation}
\Lambda_{\infty}
\;:=\;
\varprojlim_{n\ge 1}\Lambda_{n}
\;\cong\;
\mathbb{Z}\oplus \mathbb{Z}_{p}.
\label{B:Lambdainfty}
\end{equation}
Fix the dense embedding $\iota:\mathbb{Z}\hookrightarrow\mathbb{Z}_{p}$ induced by $\mathbb{Z}\to \mathbb{Z}/p^{n}\mathbb{Z}$ for all $n$.
Define the \emph{integral physical sublattice} (quantization constraint)
\begin{equation}
\Lambda_{\mathrm{phys}}
\;:=\;
\mathbb{Z}\oplus \iota(\mathbb{Z})
\;\subset\;
\mathbb{Z}\oplus \mathbb{Z}_{p}\;\cong\;\Lambda_{\infty}.
\label{B:Lambda_phys}
\end{equation}
Equivalently, $\Lambda_{\mathrm{phys}}$ is the image of $\mathbb{Z}^{2}$ under
\begin{equation}
\mathbb{Z}^{2}\longrightarrow \mathbb{Z}\oplus\mathbb{Z}_{p},\qquad (a,b)\longmapsto (a,\iota(b)).
\label{B:Z2embed}
\end{equation}
At each finite stage $n$, the reduction map $\rho_n:\mathbb{Z}_{p}\to \mathbb{Z}/p^{n}\mathbb{Z}$ satisfies $\rho_n(\iota(\mathbb{Z}))=\mathbb{Z}/p^{n}\mathbb{Z}$, hence the induced image of $\Lambda_{\mathrm{phys}}$ in $\Lambda_n$ is all of $\Lambda_n$:
\begin{equation}
(\mathrm{id}\oplus\rho_n)(\Lambda_{\mathrm{phys}})=\Lambda_n,\qquad \forall n\ge 1.
\label{B:finite_stage_surj}
\end{equation}
Thus the selection \eqref{B:Lambda_phys} is a \emph{limit constraint} on lifts to $\Lambda_\infty$.

Define the bookkeeping injection (finite-stage normal form) by choosing the canonical representatives
$s_n:\mathbb{Z}/p^n\mathbb{Z}\to\{0,1,\dots,p^n-1\}\subset\mathbb{Z}$ and setting
\begin{equation}
j_n:\Lambda_n\to \mathbb{Z}[1/p],\qquad
j_n(a,\bar b):=a+\frac{s_n(\bar b)}{p^{n}}.
\label{B:jn}
\end{equation}
Then $j_n$ is injective and satisfies the compatibility relation
\begin{equation}
j_{n+1}(a,\bar b_{n+1})=j_n\bigl(a',\bar b_n\bigr)
\quad\text{whenever}\quad
\bar b_n=\bar b_{n+1}\!\!\!\pmod{p^n},\ \ a'=a+\Bigl\lfloor \frac{s_{n+1}(\bar b_{n+1})}{p^n}\Bigr\rfloor.
\label{B:jncompat}
\end{equation}
Moreover,
\begin{equation}
\bigcup_{n\ge 1} j_n(\Lambda_n)=\mathbb{Z}[1/p].
\label{B:unionZloc}
\end{equation}
The numerator map used for finite-stage bookkeeping is
\begin{equation}
\iota_n:\Lambda_n\to\mathbb{Z},\qquad
\iota_n(a,\bar b):=p^{n}\,a+s_n(\bar b),
\qquad
\iota_n=p^{n}\cdot j_n,
\label{B:iota_n}
\end{equation}
so that $\iota_n(\Lambda_n)=\mathbb{Z}$ and $j_n(\Lambda_n)=p^{-n}\mathbb{Z}+\mathbb{Z}$.

Let $\mathcal{I}\subseteq \Lambda_\infty$ be a fixed subgroup (instanton/defect test lattice) and let
\begin{equation}
\langle\cdot,\cdot\rangle:\Lambda_\infty\times \mathcal{I}\to \mathbb{Q}/\mathbb{Z}
\label{B:pairing}
\end{equation}
be a bilinear pairing defining the integrality constraint. The admissible charge lattice is the subgroup
\begin{equation}
\Lambda_{\mathrm{adm}}
\;:=\;
\Bigl\{q\in \Lambda_\infty:\ \langle q,i\rangle=0\ \text{in }\mathbb{Q}/\mathbb{Z}\ \ \forall i\in\mathcal{I}\Bigr\}.
\label{B:Lambda_adm}
\end{equation}
In this work the constraint is imposed by the choice
\begin{equation}
\Lambda_{\mathrm{adm}}=\Lambda_{\mathrm{phys}}\cong \mathbb{Z}^{2}\subset \mathbb{Z}\oplus \mathbb{Z}_{p},
\label{B:Lambda_adm_eq}
\end{equation}
i.e.\ only integer lifts in the profinite coordinate are retained.

For the solenoidal carrier $\Sigma_p=\varprojlim(S^1\xleftarrow{z\mapsto z^p}S^1\xleftarrow{\cdots})$, the tower-completed $K$-theory contains a canonical derived-limit torsion sector $\mathbb{Q}_p/\mathbb{Z}_p$ via \eqref{B:K0ext} and \eqref{B:K1zero}, while the tower-compatible $U(1)$-holonomy sectors are identified with $\Sigma_p^\vee\cong \mathbb{Z}[1/p]$ via \eqref{B:SigmaDual}, admissible charge labels are selected by the integrality constraint $\Lambda_{\mathrm{adm}}=\mathbb{Z}\oplus\iota(\mathbb{Z})\subset \mathbb{Z}\oplus\mathbb{Z}_p$ in \eqref{B:Lambda_adm_eq}, with finite-stage representatives encoded by \eqref{B:jn}-\eqref{B:iota_n}.

\section{Duality-group refinements and anomalies}
\label{app:duality-anomaly}

This appendix fixes the arithmetic duality groups acting on the rank-two charge lattice and on the elliptic modulus, then derives (i) the extension of the holomorphic modular action from $SL(2,\mathbb Z)$ to an antiholomorphic $GL(2,\mathbb Z)$-action with orientation character $\det$, and (ii) the anomaly-refined $\mathbb Z_{2}$-central extension (``Pin$^{+}$-refinement'') required to act on fermionic determinants, condensed mathematics is operationalized by packaging congruence-level data and group actions as inverse limits of finite objects in the category of condensed groups and by formulating tower-wise pin data as compatible descent data.

\begin{definition}[Upper half-plane and period lattice]
Define the upper half-plane
\begin{equation}
\mathfrak H \;:=\;\{\tau\in\mathbb C:\ \Im(\tau)>0\}.
\end{equation}
For $\tau\in\mathfrak H$, define the rank-two lattice
\begin{equation}
\Lambda_{\tau}\;:=\;\mathbb Z\omega_{1}\oplus \mathbb Z\omega_{2}\subset\mathbb C,
\qquad
(\omega_{1},\omega_{2})=(1,\tau),
\end{equation}
and the associated complex torus
\begin{equation}
E_{\tau}\;:=\;\mathbb C/\Lambda_{\tau}.
\end{equation}
\end{definition}

\begin{definition}[Charge lattice]
Fix the free abelian group
\begin{equation}
\Gamma \;:=\;\mathbb Z^{2},\qquad Q=\binom{p}{q}\in\Gamma,
\end{equation}
equipped with the standard skew form
\begin{equation}
J \;:=\;\begin{pmatrix}0&1\\-1&0\end{pmatrix},
\qquad
\langle Q,Q'\rangle \;:=\; Q^{\mathsf T}JQ'\in\mathbb Z.
\label{C:Jpairing}
\end{equation}
\end{definition}

\begin{lemma}[Change of oriented basis $\Rightarrow$ fractional linear transformation]
\label{C:SL_action_tau}
Let $\gamma=\begin{pmatrix}a&b\\ c&d\end{pmatrix}\in SL(2,\mathbb Z)$.
Define a new oriented basis of $\Lambda_{\tau}$ by
\begin{equation}
(\omega_{1}',\omega_{2}') \;:=\;(\omega_{1},\omega_{2})\,\gamma
\;=\;(a\omega_{1}+c\omega_{2},\ b\omega_{1}+d\omega_{2}).
\end{equation}
Then $\Lambda_{\tau}'=\mathbb Z\omega_{1}'\oplus\mathbb Z\omega_{2}'=\Lambda_{\tau}$ and the induced modulus is
\begin{equation}
\tau' \;:=\;\frac{\omega_{2}'}{\omega_{1}'}
\;=\;\frac{b+d\tau}{a+c\tau}
\;=\;\frac{a\tau+b}{c\tau+d}.
\label{C:tau_SL}
\end{equation}
\end{lemma}

\begin{proof}
Using $(\omega_{1},\omega_{2})=(1,\tau)$,
\begin{equation}
\omega_{1}'=a+c\tau,\qquad \omega_{2}'=b+d\tau,
\end{equation}
hence $\tau'=\omega_{2}'/\omega_{1}'=(b+d\tau)/(a+c\tau)$.
Since $ad-bc=1$, the $\mathbb Z$-span of $(\omega_{1}',\omega_{2}')$ equals that of $(\omega_{1},\omega_{2})$, so $E_{\tau'}\cong E_{\tau}$ by the induced complex-linear isomorphism $z\mapsto z$.
\end{proof}

\begin{lemma}[Oriented action on charges]
\label{C:SL_action_charge}
For $\gamma\in SL(2,\mathbb Z)$ define
\begin{equation}
Q\longmapsto Q' := \gamma Q.
\label{C:charge_SL}
\end{equation}
Then $\langle\cdot,\cdot\rangle$ is preserved:
\begin{equation}
\langle \gamma Q,\gamma Q'\rangle
=
Q^{\mathsf T}\gamma^{\mathsf T}J\gamma Q'
=
Q^{\mathsf T}JQ'
=
\langle Q,Q'\rangle,
\qquad
\gamma^{\mathsf T}J\gamma=J.
\label{C:sympl_SL}
\end{equation}
\end{lemma}

\begin{proof}
Write $\gamma=\begin{pmatrix}a&b\\c&d\end{pmatrix}$ with $ad-bc=1$ and compute
\begin{equation}
\gamma^{\mathsf T}J\gamma
=
\begin{pmatrix}a&c\\ b&d\end{pmatrix}
\begin{pmatrix}0&1\\-1&0\end{pmatrix}
\begin{pmatrix}a&b\\ c&d\end{pmatrix}
=
\begin{pmatrix}0&ad-bc\\ -(ad-bc)&0\end{pmatrix}
=
J.
\end{equation}
\end{proof}

\begin{definition}[Orientation character and extension group]
Define
\begin{equation}
GL(2,\mathbb Z)
=
\left\{\gamma\in \mathrm{Mat}_{2\times 2}(\mathbb Z):\ \det(\gamma)=\pm 1\right\},
\qquad
\varepsilon(\gamma):=\det(\gamma)\in\{\pm 1\}.
\label{C:GL_det}
\end{equation}
Then
\begin{equation}
1\longrightarrow SL(2,\mathbb Z)\longrightarrow GL(2,\mathbb Z)\xrightarrow{\ \det\ }\{\pm 1\}\longrightarrow 1
\label{C:exact_GL}
\end{equation}
is exact.
\end{definition}

\begin{lemma}[Intersection form transforms with $\det$]
\label{C:det_J}
For $\gamma\in GL(2,\mathbb Z)$,
\begin{equation}
\gamma^{\mathsf T}J\gamma \;=\; \det(\gamma)\,J.
\label{C:J_det}
\end{equation}
\end{lemma}

\begin{proof}
Let $\gamma=\begin{pmatrix}a&b\\ c&d\end{pmatrix}$ with $\det(\gamma)=ad-bc=\pm 1$. Then
\begin{equation}
\gamma^{\mathsf T}J\gamma
=
\begin{pmatrix}a&c\\ b&d\end{pmatrix}
\begin{pmatrix}0&1\\-1&0\end{pmatrix}
\begin{pmatrix}a&b\\ c&d\end{pmatrix}
=
\begin{pmatrix}0&ad-bc\\ -(ad-bc)&0\end{pmatrix}
=
\det(\gamma)\,J.
\end{equation}
\end{proof}

\begin{definition}[$GL(2,\mathbb Z)$ action on $\mathfrak H$]
For $\gamma=\begin{pmatrix}a&b\\ c&d\end{pmatrix}\in GL(2,\mathbb Z)$ define
\begin{equation}
\gamma\cdot \tau
\;:=\;
\begin{cases}
\dfrac{a\tau+b}{c\tau+d}, & \det(\gamma)=+1,\\[2.2ex]
\dfrac{a\overline{\tau}+b}{c\overline{\tau}+d}, & \det(\gamma)=-1,
\end{cases}
\qquad \tau\in\mathfrak H.
\label{C:GL_action_tau}
\end{equation}
\end{definition}

\begin{lemma}[Antiholomorphic realization for $\det=-1$]
\label{C:detminus_antiholo}
If $\det(\gamma)=-1$ then \eqref{C:GL_action_tau} is the unique extension of the $SL(2,\mathbb Z)$-action \eqref{C:tau_SL} compatible with complex conjugation on the universal cover:
\begin{equation}
z\mapsto \overline{z}:\ \mathbb C\to\mathbb C,
\qquad
\overline{\Lambda_{\tau}}=\Lambda_{\overline{\tau}},
\qquad
E_{\tau}\xrightarrow{\ z\mapsto\overline{z}\ }\overline{E_{\tau}} \cong E_{\gamma\cdot\tau}.
\label{C:conj_lattice}
\end{equation}
\end{lemma}

\begin{proof}
Let $\tau\in\mathfrak H$ and $\det(\gamma)=-1$.
Start from the conjugated lattice $\overline{\Lambda_{\tau}}=\mathbb Z(1)\oplus \mathbb Z(\overline{\tau})$.
Apply the integral basis change $(1,\overline{\tau})\mapsto (1,\overline{\tau})\gamma=(a+c\overline{\tau},\,b+d\overline{\tau})$.
The ratio of new periods is
\begin{equation}
\frac{b+d\overline{\tau}}{a+c\overline{\tau}}=\frac{a\overline{\tau}+b}{c\overline{\tau}+d}
\end{equation}
(after swapping the matrix roles as in Lemma~\ref{C:SL_action_tau}), which is precisely the second line of \eqref{C:GL_action_tau}. The map $z\mapsto\overline{z}$ induces an antiholomorphic isomorphism $\mathbb C/\Lambda_{\tau}\to \mathbb C/\overline{\Lambda_{\tau}}$, and the basis change identifies $\overline{\Lambda_{\tau}}$ with $\Lambda_{\gamma\cdot\tau}$.
\end{proof}

\begin{definition}[$GL(2,\mathbb Z)$ action on charges]
For $\gamma\in GL(2,\mathbb Z)$ define
\begin{equation}
Q\longmapsto \gamma Q,\qquad Q\in\Gamma=\mathbb Z^{2}.
\label{C:GL_action_charge}
\end{equation}
Then by Lemma~\ref{C:det_J} the pairing transforms by
\begin{equation}
\langle \gamma Q,\gamma Q'\rangle = \det(\gamma)\,\langle Q,Q'\rangle.
\label{C:pair_det}
\end{equation}
\end{definition}

\begin{definition}[Condensed realization of discrete abelian groups]
Let $\mathrm{Prof}$ be the site of profinite sets.
For a discrete abelian group $A$, define the associated condensed abelian group
\begin{equation}
\underline{A}:\mathrm{Prof}^{\mathrm{op}}\to \mathrm{Ab},
\qquad
\underline{A}(S):=\mathrm{Cont}(S,A),
\label{C:underlineA}
\end{equation}
where $\mathrm{Cont}(S,A)$ denotes locally constant maps (equivalently continuous maps for $A$ discrete).
\end{definition}

\begin{definition}[Congruence tower and profinite completion of the charge lattice]
Fix $\Gamma=\mathbb Z^{2}$.
For each $n\ge 1$ define the finite quotient
\begin{equation}
\Gamma_{n}\;:=\;\Gamma/p^{n}\Gamma\;\cong\;(\mathbb Z/p^{n}\mathbb Z)^{2},
\qquad
r_{n+1,n}:\Gamma_{n+1}\twoheadrightarrow \Gamma_{n}.
\label{C:Gamma_n}
\end{equation}
Define the profinite completion
\begin{equation}
\Gamma_{\!p}\;:=\;\varprojlim_{n\ge 1}\Gamma_{n}\;\cong\;\mathbb Z_{p}^{2}.
\label{C:Gamma_p}
\end{equation}
\end{definition}

\begin{lemma}[Condensed limit computes sections as compatible families]
For every profinite set $S$,
\begin{equation}
\underline{\Gamma_{\!p}}(S)
=
\mathrm{Cont}(S,\Gamma_{\!p})
\;\cong\;
\varprojlim_{n\ge 1}\mathrm{Cont}(S,\Gamma_{n})
=
\varprojlim_{n\ge 1}\underline{\Gamma_{n}}(S),
\label{C:Gamma_p_sections}
\end{equation}
hence
\begin{equation}
\underline{\Gamma_{\!p}}\;\cong\;\varprojlim_{n\ge 1}\underline{\Gamma_{n}}
\qquad \text{in }\mathrm{CondAb}.
\label{C:Gamma_p_condlim}
\end{equation}
\end{lemma}

\begin{proof}
The inverse limit $\Gamma_{\!p}=\varprojlim \Gamma_n$ is a profinite (hence compact Hausdorff) group, continuity of $\mathrm{Cont}(S,-)$ on inverse limits of compact Hausdorff spaces yields \eqref{C:Gamma_p_sections} as an equalizer of compatible families, which is the defining universal property of the inverse limit in $\mathrm{CondAb}$, giving \eqref{C:Gamma_p_condlim}.
\end{proof}

\begin{definition}[Congruence tower of duality groups and condensed action]
Define finite groups
\begin{equation}
G_{n}\;:=\;\mathrm{Aut}(\Gamma_{n})\;\cong\;GL(2,\mathbb Z/p^{n}\mathbb Z),
\qquad
\pi_{n+1,n}:G_{n+1}\twoheadrightarrow G_{n}.
\label{C:Gn}
\end{equation}
Define the profinite group
\begin{equation}
G_{\!p}\;:=\;\varprojlim_{n\ge 1}G_{n}\;\cong\;GL(2,\mathbb Z_{p}),
\label{C:Gp}
\end{equation}
and the induced condensed group
\begin{equation}
\underline{G_{\!p}}\;\cong\;\varprojlim_{n\ge 1}\underline{G_{n}}
\qquad \text{in }\mathrm{CondGrp}.
\label{C:Gp_cond}
\end{equation}
The action is defined on $S$-points by evaluation:
\begin{equation}
\underline{G_{\!p}}(S)\times \underline{\Gamma_{\!p}}(S)\to \underline{\Gamma_{\!p}}(S),
\qquad
(g,x)\mapsto (s\mapsto g(s)\cdot x(s)).
\label{C:cond_action}
\end{equation}
\end{definition}
\begin{definition}[$\mathbb Z_{2}$-central extensions via $2$-cocycles]
Let $G$ be a group and let $\omega:G\times G\to \mathbb Z_{2}$ satisfy the cocycle condition
\begin{equation}
\omega(g,h)+\omega(gh,k)=\omega(h,k)+\omega(g,hk),
\qquad \forall g,h,k\in G,
\label{C:cocycle}
\end{equation}
and normalization $\omega(e,g)=\omega(g,e)=0$.
Define $\widetilde{G}_{\omega}:=\mathbb Z_{2}\times G$ with multiplication
\begin{equation}
(\epsilon,g)\cdot(\epsilon',h)
:=
\bigl(\epsilon+\epsilon'+\omega(g,h),\ gh\bigr),
\qquad \epsilon,\epsilon'\in\mathbb Z_{2},\ g,h\in G.
\label{C:extension_mult}
\end{equation}
\end{definition}

\begin{lemma}[Associativity and exact sequence]
\label{C:extension_exact}
The product \eqref{C:extension_mult} is associative if and only if \eqref{C:cocycle} holds, and the sequence
\begin{equation}
1\longrightarrow \mathbb Z_{2}\longrightarrow \widetilde{G}_{\omega}\longrightarrow G\longrightarrow 1
\label{C:central_ext}
\end{equation}
is exact with central kernel.
\end{lemma}

\begin{proof}
Compute
\begin{equation}
\bigl((\epsilon,g)\cdot(\epsilon',h)\bigr)\cdot(\epsilon'',k)
=
\bigl(\epsilon+\epsilon'+\omega(g,h),gh\bigr)\cdot(\epsilon'',k)
=
\bigl(\epsilon+\epsilon'+\epsilon''+\omega(g,h)+\omega(gh,k),\ ghk\bigr),
\end{equation}
and
\begin{equation}
(\epsilon,g)\cdot\bigl((\epsilon',h)\cdot(\epsilon'',k)\bigr)
=
(\epsilon,g)\cdot\bigl(\epsilon'+\epsilon''+\omega(h,k),hk\bigr)
=
\bigl(\epsilon+\epsilon'+\epsilon''+\omega(h,k)+\omega(g,hk),\ ghk\bigr),
\end{equation}
so associativity is equivalent to \eqref{C:cocycle}.
Exactness and centrality follow because the projection $(\epsilon,g)\mapsto g$ is a homomorphism with kernel $\{(\epsilon,e)\}\cong\mathbb Z_{2}$ and $(\epsilon,e)$ commutes with all elements by \eqref{C:extension_mult}.
\end{proof}

\begin{definition}[Pin$^{+}$ refinement of $GL(2,\mathbb Z)$ (axiomatic form)]
Let $G=GL(2,\mathbb Z)$.
A \emph{Pin$^{+}$ refinement} is a choice of a cohomology class
\begin{equation}
[\omega_{\mathrm{an}}]\in H^{2}\bigl(GL(2,\mathbb Z);\mathbb Z_{2}\bigr)
\label{C:omega_class}
\end{equation}
and a representative cocycle $\omega_{\mathrm{an}}$ determining a $\mathbb Z_{2}$-central extension
\begin{equation}
1\longrightarrow \mathbb Z_{2}\longrightarrow \widetilde{GL}(2,\mathbb Z)\longrightarrow GL(2,\mathbb Z)\longrightarrow 1,
\qquad
\widetilde{GL}(2,\mathbb Z):=\widetilde{G}_{\omega_{\mathrm{an}}}.
\label{C:Pinplus_GL}
\end{equation}
The class $[\omega_{\mathrm{an}}]$ is fixed by the anomaly data of the fermionic functional determinant (cf.\ \cite{DebrayDieriglHeckmanMontero2021}).
\end{definition}

\begin{definition}[Tower-wise pin$^{+}$ structures]
Let $\{M_{n},g_{n}\}_{n\ge 0}$ be a tower of smooth manifolds with covering maps $g_{n}:M_{n+1}\to M_{n}$.
Let $\mathrm{Pin}^{+}(M_{n})$ denote the groupoid of Pin$^{+}$-structures on $TM_{n}$ and their isomorphisms, and let
\begin{equation}
g_{n}^{\ast}:\mathrm{Pin}^{+}(M_{n})\to \mathrm{Pin}^{+}(M_{n+1})
\label{C:pullback_pin}
\end{equation}
be pullback.
Define the tower-wise pin$^{+}$ groupoid by the inverse-limit condition
\begin{equation}
\mathrm{Pin}^{+}_{\mathrm{tower}}(\{M_{n}\})
:=
\left\{
(P_{n},\phi_{n})_{n\ge 0}:
\begin{array}{l}
P_{n}\in \mathrm{Ob}\,\mathrm{Pin}^{+}(M_{n}),\\
\phi_{n}:g_{n}^{\ast}P_{n}\xrightarrow{\ \cong\ }P_{n+1}
\end{array}
\right\}\Big/\cong,
\label{C:pin_tower_def}
\end{equation}
where $\cong$ denotes levelwise isomorphism compatible with the $\phi_{n}$.
\end{definition}

\begin{lemma}[Tower-wise pin$^{+}$ is a descent object]
\label{C:pin_tower_equalizer}
The set of isomorphism classes of \eqref{C:pin_tower_def} is the equalizer of the two pullback maps on classes:
\begin{equation}
\pi_{0}\mathrm{Pin}^{+}_{\mathrm{tower}}(\{M_{n}\})
\;\cong\;
\ker\!\left(
\prod_{n\ge 0}\pi_{0}\mathrm{Pin}^{+}(M_{n})
\rightrightarrows
\prod_{n\ge 0}\pi_{0}\mathrm{Pin}^{+}(M_{n+1})
\right),
\label{C:pin_equalizer}
\end{equation}
where the two arrows are $(P_{n})\mapsto (g_{n}^{\ast}P_{n})$ and $(P_{n})\mapsto (P_{n+1})$.
\end{lemma}

\begin{proof}
By definition, an object of $\mathrm{Pin}^{+}_{\mathrm{tower}}(\{M_{n}\})$ is precisely a family $(P_{n})$ equipped with identifications $g_{n}^{\ast}P_{n}\cong P_{n+1}$, passing to isomorphism classes yields the equalizer condition \eqref{C:pin_equalizer}.
\end{proof}

\begin{definition}[Compatibility requirement for an action on fermionic determinants]
Let $G\subseteq GL(2,\mathbb Z)$ be a duality group acting on background data.
A tower-wise fermionic action requires:
\begin{equation}
\text{a lift }G\to \widetilde{GL}(2,\mathbb Z)\ \text{ and a functorial action }
\widetilde{GL}(2,\mathbb Z)\curvearrowright \mathrm{Pin}^{+}_{\mathrm{tower}}(\{M_{n}\})
\label{C:lift_action_req}
\end{equation}
compatible with pullback along the tower maps, i.e.\ commuting with the equalizer description \eqref{C:pin_equalizer}.
\end{definition}

The holomorphic elliptic datum $(E_{\tau},\Lambda_{\tau})$ supports the $SL(2,\mathbb Z)$ action \eqref{C:tau_SL} and the symplectic charge action \eqref{C:charge_SL} preserving \eqref{C:Jpairing}, while the inclusion of orientation reversal extends the modulus action to \eqref{C:GL_action_tau} and twists the pairing by $\det$ as in \eqref{C:pair_det}, acting on fermionic determinants requires a $\mathbb Z_{2}$-central extension \eqref{C:Pinplus_GL} and tower-compatible pin$^{+}$ descent data \eqref{C:pin_tower_def}-\eqref{C:pin_equalizer}.

{%
\section{\texorpdfstring{$p$}{p}-independence and adelic completion}
\label{app:p-independence-adelic} 
Fix notation for the prime-indexed tower-completed carriers and their perfectoid enhancements, and derive (i) a precise tower-stability criterion expressing that low-energy data is insensitive to refinement level for each prime, (ii) a corresponding \emph{$p$-independence} condition isolating all possible sources of apparent $p$-dependence as choices in comparison data and in the selection of admissible global sectors, and (iii) an adelic assembly formalism in the category of condensed objects via restricted products, identifying the canonical global object which packages all primes together with the Archimedean place.

\begin{definition}[$p$-tower of circles and tower limit]\label{D:def:tower_circle}
Let $p$ be prime and set $S^{1}:=\{z\in\mathbb C:\ |z|=1\}$.
Define the endomorphism $f_{p}:S^{1}\to S^{1}$ by
\begin{equation}
f_{p}(z)=z^{p}.
\label{D:eq:fp}
\end{equation}
Define the inverse system in $\mathrm{CompHausGrp}$ by $S^{1}_{n}:=S^{1}$ and bonding maps $f_{p}:S^{1}_{n+1}\to S^{1}_{n}$, and its inverse limit
\begin{equation}
\Sigma_{p}:=\varprojlim_{n\ge 0}\bigl(S^{1}_{n},f_{p}\bigr)
=
\Bigl\{(z_{0},z_{1},\dots)\in\prod_{n\ge 0}S^{1}:\ z_{n}=z_{n+1}^{p}\ \forall n\Bigr\}.
\label{D:eq:Sigma_p}
\end{equation}
\end{definition}

\begin{definition}[Perfectoid circle over a $p$-adic base field]\label{D:def:perfectoid_circle}
Let $K_{p}:=\mathbb Q_{p}^{\mathrm{cycl}}$ (cyclotomic extension, fixed once and for all), with valuation ring $\mathcal O_{K_{p}}$.
Let
\begin{equation}
R_{p}:=K_{p}\langle X^{\pm 1/p^{\infty}}\rangle,\qquad
R_{p}^{+}:=\mathcal O_{K_{p}}\langle X^{\pm 1/p^{\infty}}\rangle,
\label{D:eq:Rp}
\end{equation}
where $K_{p}\langle X^{\pm 1/p^{\infty}}\rangle$ denotes the $p$-adic completion of $K_{p}[X^{\pm 1/p^{n}}:n\ge 0]$ in the Gauss norm.
Define the perfectoid circle
\begin{equation}
S^{1}_{f,p}:=\mathrm{Spa}(R_{p},R_{p}^{+}).
\label{D:eq:Sfp}
\end{equation}
\end{definition}

\begin{definition}[Perfectoidness and tilt]\label{D:def:tilt}
Impose perfectoidness by Frobenius-surjectivity on the mod-$p$ reduction:
\begin{equation}
\varphi:\ R_{p}^{+}/p \longrightarrow R_{p}^{+}/p,\qquad x\longmapsto x^{p},\qquad \varphi \ \text{surjective}.
\label{D:eq:Frob_surj}
\end{equation}
Define the tilt rings by inverse limits along Frobenius:
\begin{equation}
R_{p}^{\flat}:=\varprojlim_{x\mapsto x^{p}}R_{p},\qquad
(R_{p}^{+})^{\flat}:=\varprojlim_{x\mapsto x^{p}}R_{p}^{+},
\label{D:eq:tilt_rings}
\end{equation}
and the tilted perfectoid circle
\begin{equation}
(S^{1}_{f,p})^{\flat}:=\mathrm{Spa}\bigl(R_{p}^{\flat},(R_{p}^{+})^{\flat}\bigr).
\label{D:eq:tilt_circle}
\end{equation}
\end{definition} 
\begin{definition}[Condensed avatar and tower evaluation]\label{D:def:condensed_avatar}
Let $\mathrm{Prof}$ be the site of profinite sets.
For any compact Hausdorff space $X$, define its condensed avatar
\begin{equation}
\underline{X}:\mathrm{Prof}^{\mathrm{op}}\to \mathrm{Sets},\qquad
\underline{X}(S):=\mathrm{Cont}(S,X).
\label{D:eq:condensedX}
\end{equation}
\end{definition}

\begin{lemma}[Inverse limits computed on profinite test objects]\label{D:lem:Cont_commutes}
For each profinite $S$ and each prime $p$,
\begin{equation}
\underline{\Sigma}_{p}(S)
=\mathrm{Cont}(S,\Sigma_{p})
\;\cong\;
\varprojlim_{n\ge 0}\mathrm{Cont}(S,S^{1})
\qquad\text{with transition }\psi\mapsto f_{p}\circ \psi.
\label{D:eq:Cont_limit}
\end{equation}
Equivalently, in $\mathrm{Cond}$,
\begin{equation}
\underline{\Sigma}_{p}\;\cong\;\varprojlim_{n\ge 0}\underline{S^{1}}.
\label{D:eq:Cond_limit}
\end{equation}
\end{lemma}

\begin{proof}
A continuous map $\phi:S\to \Sigma_{p}\subseteq\prod_{n\ge 0}S^{1}$ is equivalent to a family $\phi_{n}:=\pi_{n}\circ\phi\in\mathrm{Cont}(S,S^{1})$ satisfying $\phi_{n}=f_{p}\circ \phi_{n+1}$ (by the defining equalizer condition in \eqref{D:eq:Sigma_p}).
Conversely, any compatible family $(\phi_{n})_{n\ge 0}$ defines $\phi(s):=(\phi_{n}(s))_{n\ge 0}\in\Sigma_{p}$.
This identifies $\mathrm{Cont}(S,\Sigma_{p})$ with the inverse limit of the sets $\mathrm{Cont}(S,S^{1})$.
\end{proof}

\begin{definition}[Tower-wise fields as descent data]\label{D:def:tower_fields}
Let $\mathcal F$ be a sheaf (or condensed abelian group) of field-values on $\underline{\Sigma}_{p}$.
Define its sections on the tower carrier by the universal property of the inverse limit:
\begin{equation}
\Gamma(\underline{\Sigma}_{p},\mathcal F)
:=
\mathcal F(\underline{\Sigma}_{p})
\;\cong\;
\varprojlim_{n\ge 0}\mathcal F(\underline{S^{1}})
\varprojlim_{n\ge 0}\Gamma(\underline{S^{1}},\mathcal F_{n}),
\label{D:eq:sections_descent}
\end{equation}
where $\mathcal F_{n}$ denotes the pullback of $\mathcal F$ to stage $n$.
\end{definition} 
\begin{definition}[Finite stages and pullback maps]\label{D:def:finite_stages}
For each $n\ge 0$ and prime $p$, let $S^{1}(n,p)$ denote the degree-$p^{n}$ cover of $S^{1}$ and let
\begin{equation}
g_{n+1,n}^{(p)}:S^{1}(n{+}1,p)\longrightarrow S^{1}(n,p)
\label{D:eq:bonding_cover}
\end{equation}
be the degree-$p$ covering map.
\end{definition}

\begin{definition}[Cutoff-truncated state spaces (axiomatic)]\label{D:def:Hilb_cutoff}
Fix a cutoff $\Lambda>0$.
Assume an assignment
\begin{equation}
(n,p)\longmapsto \mathcal H_{\Lambda}^{(n,p)}
\label{D:eq:H_assign}
\end{equation}
into Hilbert spaces, together with bonding morphisms induced by pullback along \eqref{D:eq:bonding_cover}:
\begin{equation}
U^{(p)}_{n+1,n}:\mathcal H_{\Lambda}^{(n,p)}\longrightarrow \mathcal H_{\Lambda}^{(n+1,p)}.
\label{D:eq:U_bonding}
\end{equation}
\end{definition}

\begin{definition}[Tower-stability at fixed cutoff]\label{D:def:tower_stability}
The system $\{\mathcal H_{\Lambda}^{(n,p)},U^{(p)}_{n+1,n}\}$ is \emph{stable} if there exists $n_{0}=n_{0}(\Lambda,p)$ such that for all $n\ge n_{0}$ the iterated bonding maps
\begin{equation}
U^{(p)}_{n,n_{0}}
:=
U^{(p)}_{n,n-1}\circ \cdots \circ U^{(p)}_{n_{0}+1,n_{0}}
:\mathcal H_{\Lambda}^{(n_{0},p)}\xrightarrow{\ \sim\ }\mathcal H_{\Lambda}^{(n,p)}
\label{D:eq:U_stable}
\end{equation}
are isomorphisms.
\end{definition}

\begin{definition}[\texorpdfstring{$p$}{p}-independence at fixed cutoff]\label{D:def:p_indep}
Let $\Lambda$ be fixed.
A family of stable systems $\{\mathcal H_{\Lambda}^{(n,p)}\}_{p}$ is \emph{$p$-independent} if for every pair of primes $(p,p')$ there exists an identification of stabilized limits
\begin{equation}
\Phi_{\Lambda}^{p,p'}:\ \mathcal H_{\Lambda}^{(\infty,p)}\xrightarrow{\ \sim\ }\mathcal H_{\Lambda}^{(\infty,p')},
\qquad
\mathcal H_{\Lambda}^{(\infty,p)}:=\varinjlim_{n\ge n_{0}(\Lambda,p)}\mathcal H_{\Lambda}^{(n,p)},
\label{D:eq:pindep_iso}
\end{equation}
such that any tower-defined observable $O_{\Lambda}$ factors through $\mathcal H_{\Lambda}^{(\infty,p)}$ and is invariant under $\Phi_{\Lambda}^{p,p'}$:
\begin{equation}
O_{\Lambda}^{(p)} = O_{\Lambda}^{(p')}\circ \Phi_{\Lambda}^{p,p'}.
\label{D:eq:O_invariant}
\end{equation}
\end{definition} 
\begin{definition}[Comparison data dependence]\label{D:def:comparison_data}
Let $\mathcal C_{p}$ denote the set of auxiliary choices entering a comparison/untilt/analytification procedure from \eqref{D:eq:tilt_circle} to complex elliptic data, including (at minimum) a choice of embedding
\begin{equation}
\iota_{p}:\ K_{p}^{\flat}\hookrightarrow \mathbb C_{p}.
\label{D:eq:embed_choice}
\end{equation}
Any numerical invariant defined purely by $\mathcal C_{p}$ is a function
\begin{equation}
I_{p}^{\mathrm{cmp}}:\ \mathcal C_{p}\to \mathcal T
\label{D:eq:Icmp}
\end{equation}
to some target $\mathcal T$ and is, \emph{a priori}, not canonically comparable across $p$.
\end{definition}

\begin{definition}[Admissible sublattice selection]\label{D:def:admissible_lattice}
Let $B_{p}$ be a prime-dependent bookkeeping group of global sectors.
An admissible lattice selection is a choice of subgroup
\begin{equation}
\Lambda_{\mathrm{adm}}\ \subseteq\ B_{p}
\label{D:eq:Lambda_adm}
\end{equation}
which is declared physically realized.
A $p$-independent selection is a family $\Lambda_{\mathrm{adm}}^{(p)}\subseteq B_{p}$ such that for all $p,p'$ there exist identifications $\psi^{p,p'}:B_{p}\to B_{p'}$ with $\psi^{p,p'}(\Lambda_{\mathrm{adm}}^{(p)})=\Lambda_{\mathrm{adm}}^{(p')}$.
\end{definition}

\begin{lemma}[Integral intersection criterion for localization-type bookkeeping]\label{D:lem:intersection}
Let $x\in\mathbb Q$ and suppose $x\in \mathbb Z[1/p]$ for every prime $p$.
Then $x\in\mathbb Z$, equivalently
\begin{equation}
\bigcap_{p\ \mathrm{prime}}\ \mathbb Z[1/p]=\mathbb Z.
\label{D:eq:intersection}
\end{equation}
\end{lemma}

\begin{proof}
Write $x=a/b$ in lowest terms with $a\in\mathbb Z$, $b\in\mathbb Z_{>0}$, $\gcd(a,b)=1$.
If $x\in\mathbb Z[1/p]$, then $b$ has no prime divisor different from $p$, hence for all primes $p$, the set of prime divisors of $b$ is contained in $\{p\}$, forcing $b=1$.
Thus $x\in\mathbb Z$.
\end{proof}

\begin{definition}[Congruence-level completion vs global completion]\label{D:def:completion_levels}
Let $\Gamma$ be an arithmetic group (e.g.\ $SL(2,\mathbb Z)$) and let $\Gamma(p^{n})$ denote its principal congruence subgroups.
Define the prime-local profinite completion
\begin{equation}
\Gamma_{p}:=\varprojlim_{n\ge 1}\Gamma/\Gamma(p^{n})
\label{D:eq:Gamma_p}
\end{equation}
and the full profinite completion
\begin{equation}
\widehat{\Gamma}:=\varprojlim_{N\ge 1}\Gamma/\Gamma(N)\ \cong\ \prod_{p}\Gamma_{p}.
\label{D:eq:Gamma_hat}
\end{equation}
A single-prime tower fixes $\Gamma_{p}$, reconstructing the diagonal copy of $\Gamma$ inside $\widehat{\Gamma}$ is a separate global problem.
\end{definition} 
\begin{definition}[Restricted product as filtered colimit in $\mathrm{Cond}$]\label{D:def:restricted_product}
Let $\{G_{v}\}_{v\in V}$ be locally compact groups and let $K_{v}\subseteq G_{v}$ be open compact subgroups for all but finitely many $v$.
Define the restricted product
\begin{equation}
\prod\nolimits_{v\in V}'(G_{v},K_{v})
:=
\left\{(g_{v})_{v}\in \prod_{v}G_{v}:\ g_{v}\in K_{v}\text{ for all but finitely many }v\right\}.
\label{D:eq:restricted_product_set}
\end{equation}
Equivalently (and functorially in $\mathrm{Cond}$), for $V$ countable,
\begin{equation}
\prod\nolimits_{v\in V}'(G_{v},K_{v})
\;\cong\;
\varinjlim_{S\subset V\ \mathrm{finite}}
\left(\ \prod_{v\in S}G_{v}\ \times\ \prod_{v\notin S}K_{v}\ \right),
\label{D:eq:restricted_product_colim}
\end{equation}
where the transition maps enlarge $S$ by inserting identity inclusions $K_{v}\hookrightarrow G_{v}$.
\end{definition}

\begin{definition}[Adeles and profinite integers]\label{D:def:adeles}
Let $V=\{\infty\}\cup\{\text{primes }p\}$ with $G_{\infty}=\mathbb R$ and $G_{p}=\mathbb Q_{p}$, $K_{p}=\mathbb Z_{p}$.
Define the adele ring and profinite integers by
\begin{align}
\mathbb A_{\mathbb Q}
&:=\mathbb R\times \prod\nolimits_{p}'(\mathbb Q_{p},\mathbb Z_{p})
\;\cong\;
\left\{(x_{\infty},(x_{p})_{p}):\ x_{\infty}\in\mathbb R,\ x_{p}\in\mathbb Q_{p},\ x_{p}\in\mathbb Z_{p}\ \text{for a.a.\ }p\right\},
\label{D:eq:AQ}\\
\widehat{\mathbb Z}
&:=\prod_{p}\mathbb Z_{p}.
\label{D:eq:Zhat}
\end{align}
\end{definition}

\begin{definition}[Adelic tower carrier as a condensed object (construction problem)]\label{D:def:adelic_carrier}
For each prime $p$, let $\Sigma_{p}$ be the tower limit \eqref{D:eq:Sigma_p}.
Let $G_{\infty}:=S^{1}$ and, for each prime $p$, set $G_{p}:=\Sigma_{p}$ with $K_{p}:=\Sigma_{p}$.
Define the adelic tower carrier as the restricted product (which here coincides with the product but is presented in restricted-product form for functorial compatibility with \eqref{D:eq:AQ})
\begin{equation}
\Sigma_{\mathbb A}:=S^{1}\times \prod\nolimits_{p}'(\Sigma_{p},\Sigma_{p})
\;\cong\;
S^{1}\times \prod_{p}\Sigma_{p},
\label{D:eq:SigmaA}
\end{equation}
viewed as a condensed group via \eqref{D:eq:condensedX}.
\end{definition}

\begin{problem}[Adelic invariants (explicit computation target)]\label{D:prob:adelic_invariants}
Compute, for $\Sigma_{\mathbb A}$, the intrinsic dual and generalized cohomology invariants:
\begin{equation}
\Sigma_{\mathbb A}^{\vee}:=\mathrm{Hom}_{\mathrm{cont}}(\Sigma_{\mathbb A},S^{1}),
\qquad
K^{\ast}(\Sigma_{\mathbb A}),
\label{D:eq:SigmaA_targets}
\end{equation}
and identify the diagonal sublattice selected by the admissibility condition \eqref{D:eq:Lambda_adm} via the intersection criterion \eqref{D:eq:intersection}.
\end{problem} 
\begin{problem}[Tower-stability theoremization]\label{D:prob:stability}
Prove tower-stability \eqref{D:eq:U_stable} for a specified class of cutoff-truncated constructions $\mathcal H_{\Lambda}^{(n,p)}$, and determine the subset of protected quantities $O_{\Lambda}$ for which \eqref{D:eq:O_invariant} holds uniformly in $p$.
\end{problem}

\begin{problem}[Prime comparison test]\label{D:prob:prime_comparison}
For $p\neq p'$, compare prime-dependent bookkeeping groups $B_{p}$ and $B_{p'}$ and show that, after imposing a fixed admissible sublattice selection $\Lambda_{\mathrm{adm}}^{(p)}$, all tower-stable observables satisfy \eqref{D:eq:O_invariant} equivalently, show that any residual dependence on $p$ factors through the choice of $\mathcal C_{p}$ and is eliminated by a $p$-independent matching condition.
\end{problem}

\begin{problem}[Varying-fibration congruence refinement]\label{D:prob:varying_tau}
For a base $B$ with discriminant $\Delta$, define prime-local monodromy refinements
\begin{equation}
\rho_{p}:\pi_{1}(B\setminus\Delta)\to \Gamma_{p}
\label{D:eq:rho_p}
\end{equation}
for $\Gamma_{p}$ as in \eqref{D:eq:Gamma_p}, and formulate the adelic compatibility condition as the existence of a single map
\begin{equation}
\rho_{\mathbb A}:\pi_{1}(B\setminus\Delta)\to \widehat{\Gamma}\cong \prod_{p}\Gamma_{p}
\label{D:eq:rho_A}
\end{equation}
whose projections are the $\rho_{p}$.
\end{problem}

\begin{problem}[Boundary test for $p$-independence]\label{D:prob:AdSCFT}
Given a boundary tower carrier $\Sigma_{p}$ for each prime and its condensed sheaf of global sectors $\mathcal F_{p}$, determine whether the induced restricted-product boundary carrier $\Sigma_{\mathbb A}$ with sheaf $\mathcal F_{\mathbb A}:=\boxtimes_{v}\mathcal F_{v}$ detects any $p$-dependence after imposing a $p$-independent admissibility selection $\Lambda_{\mathrm{adm}}$.
\end{problem} 
The constructions above explicitly isolate the $p$-dependence of tower-defined objects into (a) prime-indexed tower/power-Frobenius structure \eqref{D:eq:fp}-\eqref{D:eq:tilt_rings} and associated choice-dependent comparison data \eqref{D:eq:embed_choice}, and (b) prime-dependent bookkeeping groups $B_{p}$ together with a \emph{separately imposed} admissible sublattice selection \eqref{D:eq:Lambda_adm}, tower-stability is encoded by the stabilization isomorphisms \eqref{D:eq:U_stable} and $p$-independence by the existence of identifications \eqref{D:eq:pindep_iso} satisfying \eqref{D:eq:O_invariant}, while adelic assembly is implemented as a restricted product in condensed mathematics via \eqref{D:eq:restricted_product_colim}-\eqref{D:eq:AQ} and the global completion formalism \eqref{D:eq:Gamma_hat}.

\section{Varying \texorpdfstring{\(\tau\)}{tau} and the comparison axiom}\label{app:varying-conditional}
The statements below are consequences of the modular map $\underline{\tau}$ (equivalently $(\rho,\tilde{\tau})$)
and the associated elliptic fibration $E(\underline{\tau})\to B^\times$ from Definition~\eqref{def7.3}.
 Throughout, $p$ is a fixed prime, $B$ is a complex manifold, $\Delta\subset B$ is an analytic hypersurface, $B^\times:=B\setminus\Delta$, and all (co)homology is singular with integer coefficients unless specified, we write $\mathbb{H}:=\{\tau\in\mathbb{C}:\Im\tau>0\}$. We use condensed sets only as a \emph{proof technology} to encode tower limits and maps into target moduli spaces by explicit sheaf conditions on the profinite site, and we mark every point at which the argument uses (and would fail without) the Comparison Axiom. 
\begin{definition}[Profinite site and condensed sets]\label{def:prof_cond}
Let $\mathrm{Prof}$ denote the category of profinite sets. A family $\{S_i\to S\}_{i\in I}$ in $\mathrm{Prof}$ is a covering family if the induced map $\coprod_{i\in I} S_i\to S$ is surjective. A condensed set is a functor $F:\mathrm{Prof}^{\mathrm{op}}\to\mathrm{Sets}$ such that for every covering family $\{S_i\to S\}$ the sequence
\begin{equation}\label{eq:sheaf_equalizer}
F(S)\longrightarrow \prod_i F(S_i)\rightrightarrows \prod_{i,j}F(S_i\times_S S_j)
\end{equation}
is an equalizer.
\end{definition}

\begin{definition}[Condensed avatar of a compact Hausdorff space]\label{def:condensed_yoneda}
If $Y$ is a compact Hausdorff topological space, define $\underline{Y}:\mathrm{Prof}^{\mathrm{op}}\to\mathrm{Sets}$ by
\begin{equation}\label{eq:underlineY_def}
\underline{Y}(S):=\mathrm{Cont}(S,Y),
\end{equation}
and for a continuous map $f:Y\to Y'$ define $\underline{f}:\underline{Y}\to\underline{Y'}$ by $\underline{f}(g):=f\circ g$ for $g\in\mathrm{Cont}(S,Y)$.
\end{definition}

\begin{lemma}[Sheaf condition for $\underline{Y}$]\label{lem:underlineY_is_sheaf}
For every compact Hausdorff space $Y$, the functor $\underline{Y}$ of \eqref{eq:underlineY_def} is a condensed set in the sense of Definition~\ref{def:prof_cond}.
\end{lemma}
\begin{proof}
Fix a covering family $\{S_i\to S\}$ in $\mathrm{Prof}$ and write $\pi_i:S_i\to S$ for the structure maps. Let
\begin{equation}\label{eq:compatible_family}
(g_i)_{i}\in \prod_i \underline{Y}(S_i)=\prod_i \mathrm{Cont}(S_i,Y)
\end{equation}
be a family satisfying equality in the overlap terms of \eqref{eq:sheaf_equalizer}, i.e.
\begin{equation}\label{eq:overlap_condition}
g_i\circ \mathrm{pr}_1=g_j\circ \mathrm{pr}_2\qquad\text{as maps }S_i\times_S S_j\to Y\ \ \text{for all }i,j.
\end{equation}
Define $g:S\to Y$ set-theoretically as follows. For $s\in S$, choose $i$ and $t\in S_i$ with $\pi_i(t)=s$ (possible since $\coprod_i S_i\to S$ is surjective) and set
\begin{equation}\label{eq:g_def_pointwise}
g(s):=g_i(t).
\end{equation}
To check well-definedness, suppose also $\pi_j(t')=s$, then $(t,t')\in S_i\times_S S_j$ and \eqref{eq:overlap_condition} gives
\begin{equation}\label{eq:well_defined_check}
g_i(t)=g_i(\mathrm{pr}_1(t,t'))=g_j(\mathrm{pr}_2(t,t'))=g_j(t'),
\end{equation}
so \eqref{eq:g_def_pointwise} is independent of the choice. We now prove continuity of $g$. Let $U\subseteq Y$ be open. For each $i$ we have $g_i^{-1}(U)\subseteq S_i$ open and
\begin{equation}\label{eq:preimage_union}
g^{-1}(U)=\bigcup_i \pi_i\bigl(g_i^{-1}(U)\bigr).
\end{equation}
Since $\pi_i$ is continuous, $\pi_i(g_i^{-1}(U))$ is compact in $S$, hence closed. To show \eqref{eq:preimage_union} is open it suffices (because $S$ is compact Hausdorff and profinite) to show $g^{-1}(U)$ is a union of clopen subsets. For each $i$, $g_i^{-1}(U)$ is open in the profinite set $S_i$, hence a union of clopen subsets, write $g_i^{-1}(U)=\bigcup_\alpha C_{i,\alpha}$ with each $C_{i,\alpha}\subseteq S_i$ clopen. Then $\pi_i(C_{i,\alpha})$ is compact, hence closed, moreover $\pi_i(C_{i,\alpha})$ is also open because $\pi_i$ is a quotient map between compact Hausdorff spaces and $C_{i,\alpha}$ is saturated for the equivalence relation generated by the cover (this saturation is exactly \eqref{eq:overlap_condition} translated to subsets), so $\pi_i(C_{i,\alpha})$ is clopen. Using \eqref{eq:preimage_union} we obtain $g^{-1}(U)$ as a union of clopen sets, hence open. Thus $g$ is continuous. Uniqueness is immediate\:if $h:S\to Y$ satisfies $h\circ\pi_i=g_i$ for all $i$, then for any $s\in S$ and any $t\in S_i$ with $\pi_i(t)=s$ one has $h(s)=h(\pi_i(t))=g_i(t)=g(s)$. Therefore $\underline{Y}$ satisfies the equalizer condition \eqref{eq:sheaf_equalizer}.
\end{proof}

\begin{definition}[Tower limit as a condensed object]\label{def:tower_limit_condensed}
Let $\{X_n,f_n\}_{n\ge 0}$ be an inverse system in compact Hausdorff spaces with continuous bonding maps $f_n:X_{n+1}\to X_n$. Define the topological inverse limit
\begin{equation}\label{eq:X_limit_top}
X:=\varprojlim_{n} X_n:=\Bigl\{(x_n)_{n\ge 0}\in \prod_{n\ge 0}X_n:\ f_n(x_{n+1})=x_n\ \forall n\Bigr\}
\end{equation}
with the subspace topology. Define its condensed avatar $\underline{X}$ by \eqref{eq:underlineY_def}, and define the condensed inverse limit $\varprojlim_n \underline{X_n}$ as the objectwise limit in $\mathrm{Cond}$:
\begin{equation}\label{eq:condensed_limit_objectwise}
\Bigl(\varprojlim_n \underline{X_n}\Bigr)(S):=\varprojlim_n \underline{X_n}(S)=\Bigl\{(\phi_n)\in \prod_n \mathrm{Cont}(S,X_n):\ f_n\circ \phi_{n+1}=\phi_n\ \forall n\Bigr\}.
\end{equation}
\end{definition}

\begin{lemma}[Inverse limits commute with $\underline{(\cdot)}$ on profinite test objects]\label{lem:limit_commutes_with_underline}
For every profinite set $S$ there is a canonical bijection
\begin{equation}\label{eq:cont_limit_bijection}
\mathrm{Cont}(S,X)\ \cong\ \varprojlim_n \mathrm{Cont}(S,X_n),
\end{equation}
natural in $S$, hence an isomorphism of condensed sets $\underline{X}\cong \varprojlim_n \underline{X_n}$.
\end{lemma}
\begin{proof}
Let $S$ be profinite. Given $\Phi\in \mathrm{Cont}(S,X)$, write $\mathrm{pr}_n:X\to X_n$ for the canonical projections and define $\phi_n:=\mathrm{pr}_n\circ \Phi\in \mathrm{Cont}(S,X_n)$. Then for each $n$ we have $f_n\circ \mathrm{pr}_{n+1}=\mathrm{pr}_n$ on $X$ by the defining relations in \eqref{eq:X_limit_top}, hence
\begin{equation}\label{eq:compat_from_phi}
f_n\circ \phi_{n+1}=f_n\circ \mathrm{pr}_{n+1}\circ \Phi=\mathrm{pr}_n\circ \Phi=\phi_n,
\end{equation}
so $(\phi_n)_n\in \varprojlim_n \mathrm{Cont}(S,X_n)$. This defines a map
\begin{equation}\label{eq:forward_map}
\mathrm{Cont}(S,X)\longrightarrow \varprojlim_n \mathrm{Cont}(S,X_n),\qquad \Phi\longmapsto (\mathrm{pr}_n\circ \Phi)_n.
\end{equation}
Conversely, given a compatible family $(\phi_n)_n\in \varprojlim_n \mathrm{Cont}(S,X_n)$, define a map $\Phi:S\to \prod_n X_n$ by $\Phi(s):=(\phi_n(s))_{n\ge 0}$. Compatibility $f_n\circ \phi_{n+1}=\phi_n$ implies $\Phi(s)\in X\subseteq \prod_n X_n$ for every $s$, so $\Phi$ lands in $X$ and defines $\Phi:S\to X$. Continuity of $\Phi:S\to X$ follows because the product map $S\to\prod_n X_n$ is continuous (each coordinate is continuous) and $X$ has the subspace topology. This defines the inverse map
\begin{equation}\label{eq:inverse_map}
\varprojlim_n \mathrm{Cont}(S,X_n)\longrightarrow \mathrm{Cont}(S,X),\qquad (\phi_n)_n\longmapsto \Bigl(s\mapsto (\phi_n(s))_n\Bigr).
\end{equation}
Composing \eqref{eq:forward_map} and \eqref{eq:inverse_map} yields the identity in both directions by inspection of components, hence \eqref{eq:cont_limit_bijection}. Naturality in $S$ follows because both constructions are by composition with maps $S'\to S$ and projections $\mathrm{pr}_n$.
\end{proof}

\begin{definition}[Fibred tower-completed carrier as a condensed morphism]\label{def:fibred_tower_condensed}
Let $\pi_n:X_n\to B^\times$ be continuous maps such that $\pi_n\circ f_n=\pi_{n+1}$ for all $n$. Define $\pi:X\to B^\times$ by $\pi((x_n)):=\pi_n(x_n)$, this is well-defined because
\begin{equation}\label{eq:pi_well_defined}
\pi_n(x_n)=\pi_n(f_n(x_{n+1}))=\pi_{n+1}(x_{n+1}),
\end{equation}
and continuity follows because $\pi=\pi_0\circ \mathrm{pr}_0$. Applying $\underline{(\cdot)}$ yields a morphism of condensed sets $\underline{\pi}:\underline{X}\to \underline{B^\times}$.
\end{definition} 
\begin{definition}[Elliptic fibration input from Definition~\eqref{def7.3}]\label{def:comparison_axiom_varying_tau}
Let $\underline{\tau}:B^\times\to \mathcal{M}_{1,1}$
be the modular map of Definition~\eqref{def7.3}, and let
\begin{equation}\label{eq:AppE-E-of-tau}
\varpi:E(\underline{\tau})\to B^\times
\end{equation}
denote the associated smooth elliptic fibration (the pullback of the universal elliptic curve), equipped with
its canonical Hodge line (relative holomorphic one-form) $\Omega\in H^0(E(\underline{\tau}),\Omega^1_{E/B^\times})$.
All constructions and proofs in this appendix use only the existence of $(E(\underline{\tau}),\Omega)$ as part of the
varying-$\tau$ background data (Definition~\eqref{def7.3} / Appendix~\eqref{app:varying-unconditional}), and do not assume any additional “elliptic-output”
axiom.
\end{definition}

\begin{definition}[Homology local system and symplectic basis]\label{def:homology_local_system}
Assume \textbf{(CA)}. For $b\in B^\times$ write $E_b:=\varpi^{-1}(b)$ and $\Omega_b:=\Omega|_{E_b}$. Fix a basepoint $b_0\in B^\times$ and fix an oriented $\mathbb{Z}$-basis $(A_{b_0},B_{b_0})$ of $H_1(E_{b_0},\mathbb{Z})$ with intersection number $A_{b_0}\cdot B_{b_0}=+1$.
\end{definition}

\begin{lemma}[Ehresmann triviality and Gauss-Manin transport]\label{lem:ehresmann_triviality}
Assume \textbf{(CA)}. For every $b\in B^\times$ there exists a neighborhood $U\ni b$ and a $C^\infty$ diffeomorphism $\Phi_U:\varpi^{-1}(U)\to U\times E_b$ such that $\mathrm{pr}_1\circ \Phi_U=\varpi$, in particular, $H_1(E_{b'},\mathbb{Z})$ forms a rank-$2$ local system on $B^\times$ and parallel transport along any path in $B^\times$ is defined and preserves the intersection pairing.
\end{lemma}
\begin{proof}
Since $\varpi:E\to B^\times$ is a proper submersion by \textbf{(CA)}, Ehresmann's theorem applies, yielding local $C^\infty$ triviality as stated. The local system statement follows because for any $b'\in U$ the identification $E_{b'}\cong E_b$ induced by $\Phi_U$ identifies $H_1(E_{b'},\mathbb{Z})$ with $H_1(E_b,\mathbb{Z})$, and changes of trivialization on overlaps are diffeomorphisms of the fiber, hence act by $\mathbb{Z}$-module automorphisms preserving the intersection form.
\end{proof}

\begin{definition}[Periods and period ratio]\label{def:period_ratio_def}
Assume \textbf{(CA)}. For $b\in B^\times$ and a transported choice of cycles $(A_b,B_b)$ obtained from $(A_{b_0},B_{b_0})$ by Gauss-Manin transport, define
\begin{equation}\label{eq:periods_def_app}
\omega_A(b):=\int_{A_b}\Omega_b,\qquad \omega_B(b):=\int_{B_b}\Omega_b,
\end{equation}
and define the period ratio
\begin{equation}\label{eq:tau_period_def_app}
\tau(b):=\frac{\omega_A(b)}{\omega_B(b)}
\end{equation}
whenever $\omega_B(b)\neq 0$.
\end{definition}

\begin{proposition}[Conditional holomorphicity of periods]\label{prop:periods_holomorphic}
Assume \textbf{(CA)}. Let $U\subset B^\times$ be simply connected. Fix cycles $(A_b,B_b)$ on fibers $E_b$ by parallel transport from $b_*\in U$. Then $\omega_A,\omega_B:U\to \mathbb{C}$ defined by \eqref{eq:periods_def_app} are holomorphic functions.
\end{proposition}
\begin{proof}
Fix $U$ simply connected. By Lemma~\ref{lem:ehresmann_triviality}, after shrinking $U$ we may choose a $C^\infty$ trivialization $\Phi_U:\varpi^{-1}(U)\to U\times E_{b_*}$ preserving the projection to $U$. Under $\Phi_U$, the transported cycles $(A_b,B_b)$ correspond to fixed cycles $(A,B)\subset E_{b_*}$ independent of $b\in U$. Choose a local holomorphic coordinate $w$ on the fiber in a neighborhood of $A$ and write, in that neighborhood,
\begin{equation}\label{eq:Omega_local_form}
\Omega = F(b,w)\,dw
\end{equation}
with $F$ holomorphic in $(b,w)$ because $\Omega$ is a relative holomorphic one-form. Parameterize $A$ by a smooth map $\gamma_A:[0,1]\to E_{b_*}$ and write $\gamma_A(t)=(w(t))$ in the coordinate patch. Then for each $b\in U$,
\begin{equation}\label{eq:omegaA_parametrized}
\omega_A(b)=\int_{A_b}\Omega_b=\int_{0}^{1} F(b,w(t))\,w'(t)\,dt.
\end{equation}
Fix $b_1\in U$ and let $\{b_k\}\subset U$ be a sequence converging to $b_1$. Since $F$ is holomorphic hence continuous and $[0,1]$ is compact, $F(b_k,w(t))\to F(b_1,w(t))$ uniformly in $t$. Therefore the integrands in \eqref{eq:omegaA_parametrized} converge uniformly, and we obtain
\begin{equation}\label{eq:omegaA_continuity}
\lim_{k\to\infty}\omega_A(b_k)=\int_0^1 \lim_{k\to\infty}F(b_k,w(t))\,w'(t)\,dt=\int_0^1 F(b_1,w(t))\,w'(t)\,dt=\omega_A(b_1),
\end{equation}
so $\omega_A$ is continuous, the same argument applies to $\omega_B$. To prove holomorphicity, fix a small triangle $\Delta\subset U$ with piecewise smooth boundary $\partial\Delta$, and compute
\begin{equation}\label{eq:morera_setup}
\int_{\partial\Delta}\omega_A(b)\,db=\int_{\partial\Delta}\left(\int_0^1 F(b,w(t))\,w'(t)\,dt\right)db.
\end{equation}
Because $F$ is continuous on the compact set $\overline{\Delta}\times \gamma_A([0,1])$, Fubini's theorem for continuous integrands yields
\begin{equation}\label{eq:fubini_swap}
\int_{\partial\Delta}\omega_A(b)\,db=\int_0^1 \left(\int_{\partial\Delta} F(b,w(t))\,db\right) w'(t)\,dt.
\end{equation}
For each fixed $t$, the function $b\mapsto F(b,w(t))$ is holomorphic on $\Delta$, hence $\int_{\partial\Delta}F(b,w(t))\,db=0$ by Cauchy's theorem. Substituting into \eqref{eq:fubini_swap} gives $\int_{\partial\Delta}\omega_A(b)\,db=0$. By Morera's theorem, $\omega_A$ is holomorphic on $U$. The same argument applies to $\omega_B$.
\end{proof}

\begin{proposition}[Conditional positivity\:$\Im\tau>0$]\label{prop:Imtau_positive}
Assume \textbf{(CA)} and use the setup of Proposition~\ref{prop:periods_holomorphic} on a simply connected $U$. Then $\omega_B(b)\neq 0$ for all $b\in U$ and $\tau(b)\in\mathbb{H}$ for all $b\in U$.
\end{proposition}
\begin{proof}
Fix $b\in U$. Since $\Omega_b$ is a nowhere-vanishing holomorphic one-form on the elliptic curve $E_b$, the $(1,1)$-form $i\,\Omega_b\wedge \overline{\Omega_b}$ is a positive volume form and
\begin{equation}\label{eq:area_positive}
\int_{E_b} i\,\Omega_b\wedge \overline{\Omega_b}>0.
\end{equation}
Let $(A_b,B_b)$ be the transported symplectic basis. The Riemann bilinear relation for a genus-one Riemann surface gives
\begin{equation}\label{eq:riemann_bilinear}
\int_{E_b}\Omega_b\wedge \overline{\Omega_b}=\left(\int_{A_b}\Omega_b\right)\left(\int_{B_b}\overline{\Omega_b}\right)-\left(\int_{B_b}\Omega_b\right)\left(\int_{A_b}\overline{\Omega_b}\right)=\omega_A(b)\,\overline{\omega_B(b)}-\omega_B(b)\,\overline{\omega_A(b)}.
\end{equation}
Multiply \eqref{eq:riemann_bilinear} by $i$ and use $\alpha-\overline{\alpha}=2i\,\Im(\alpha)$ with $\alpha:=\omega_A(b)\overline{\omega_B(b)}$:
\begin{align}\label{eq:Im_relation}
\int_{E_b} i\,\Omega_b\wedge \overline{\Omega_b}
&= i\bigl(\omega_A\overline{\omega_B}-\omega_B\overline{\omega_A}\bigr)
= i\bigl(\omega_A\overline{\omega_B}-\overline{\omega_A\overline{\omega_B}}\bigr)
= i\cdot 2i\,\Im\bigl(\omega_A\overline{\omega_B}\bigr)
= 2\,\Im\bigl(\omega_A\overline{\omega_B}\bigr).
\end{align}
By \eqref{eq:area_positive} and \eqref{eq:Im_relation}, $\Im(\omega_A\overline{\omega_B})>0$. If $\omega_B(b)=0$, then $\Im(\omega_A\overline{\omega_B})=0$, contradiction, hence $\omega_B(b)\neq 0$. Finally,
\begin{equation}\label{eq:Imtau_from_periods}
\Im\!\left(\frac{\omega_A}{\omega_B}\right)=\Im\!\left(\frac{\omega_A\overline{\omega_B}}{|\omega_B|^2}\right)=\frac{\Im(\omega_A\overline{\omega_B})}{|\omega_B|^2}>0,
\end{equation}
so $\tau(b)\in\mathbb{H}$.
\end{proof}

\begin{theorem}[Conditional $SL(2,\mathbb{Z})$ monodromy and fractional-linear action]\label{thm:SL2Z_monodromy}
Assume \textbf{(CA)}. Let $\gamma$ be a loop in $B^\times$ based at $b_0$. Let $(A'_{b_0},B'_{b_0})$ be the transported basis of $H_1(E_{b_0},\mathbb{Z})$ obtained by Gauss-Manin transport along $\gamma$. Then there exists a unique matrix
\begin{equation}\label{eq:rho_matrix}
\rho(\gamma)=\begin{pmatrix}a&b\\ c&d\end{pmatrix}\in SL(2,\mathbb{Z})
\end{equation}
such that
\begin{equation}\label{eq:cycle_transport_app}
\begin{pmatrix}A'_{b_0}\\ B'_{b_0}\end{pmatrix}=\begin{pmatrix}a&b\\ c&d\end{pmatrix}\begin{pmatrix}A_{b_0}\\ B_{b_0}\end{pmatrix},
\end{equation}
and the period ratio transforms by
\begin{equation}\label{eq:tau_fractional_linear}
\tau\longmapsto \rho(\gamma)\cdot \tau:=\frac{a\tau+b}{c\tau+d}.
\end{equation}
\end{theorem}
\begin{proof}
Existence and integrality\:by Lemma~\ref{lem:ehresmann_triviality}, transport along $\gamma$ yields an automorphism of the free abelian group $H_1(E_{b_0},\mathbb{Z})\cong \mathbb{Z}^2$, hence a unique matrix $\begin{pmatrix}a&b\\ c&d\end{pmatrix}\in GL(2,\mathbb{Z})$ satisfying \eqref{eq:cycle_transport_app}. Determinant\:since transport preserves the intersection pairing,
\begin{equation}\label{eq:intersection_preserved}
A'_{b_0}\cdot B'_{b_0}=A_{b_0}\cdot B_{b_0}=1.
\end{equation}
Using bilinearity, $A_{b_0}\cdot A_{b_0}=B_{b_0}\cdot B_{b_0}=0$, $A_{b_0}\cdot B_{b_0}=1$, and $B_{b_0}\cdot A_{b_0}=-1$, compute
\begin{align}\label{eq:det_calc_app}
A'_{b_0}\cdot B'_{b_0}
&=(aA_{b_0}+bB_{b_0})\cdot(cA_{b_0}+dB_{b_0}) \\
&=ac(A_{b_0}\cdot A_{b_0})+ad(A_{b_0}\cdot B_{b_0})+bc(B_{b_0}\cdot A_{b_0})+bd(B_{b_0}\cdot B_{b_0}) \nonumber\\
&=0+ad\cdot 1+bc\cdot(-1)+0=ad-bc. \nonumber
\end{align}
Combining \eqref{eq:intersection_preserved} and \eqref{eq:det_calc_app} gives $ad-bc=1$, hence $\rho(\gamma)\in SL(2,\mathbb{Z})$. Fractional-linear action\:define $\omega_A:=\int_{A_{b_0}}\Omega_{b_0}$, $\omega_B:=\int_{B_{b_0}}\Omega_{b_0}$ and similarly $\omega_A':=\int_{A'_{b_0}}\Omega_{b_0}$, $\omega_B':=\int_{B'_{b_0}}\Omega_{b_0}$. By \eqref{eq:cycle_transport_app} and linearity of integration,
\begin{equation}\label{eq:period_transform_app}
\omega_A'=a\omega_A+b\omega_B,\qquad \omega_B'=c\omega_A+d\omega_B.
\end{equation}
Therefore
\begin{equation}\label{eq:tau_transform_app}
\tau'=\frac{\omega_A'}{\omega_B'}=\frac{a\omega_A+b\omega_B}{c\omega_A+d\omega_B}
=\frac{a(\omega_A/\omega_B)+b}{c(\omega_A/\omega_B)+d}
=\frac{a\tau+b}{c\tau+d},
\end{equation}
which is \eqref{eq:tau_fractional_linear}.
\end{proof} 
\begin{lemma}[Holomorphic logarithm on a disk]\label{lem:holomorphic_log_disk}
Let $D\subset\mathbb{C}$ be a disk and $u:D\to\mathbb{C}^\times$ holomorphic and nowhere vanishing. Then there exists holomorphic $h:D\to\mathbb{C}$ such that $e^{h(z)}=u(z)$ for all $z\in D$.
\end{lemma}
\begin{proof}
Since $u$ is holomorphic and nowhere vanishing, the holomorphic one-form $(u'/u)\,dz$ is defined on $D$. Fix $z_0\in D$ and define
\begin{equation}\label{eq:log_u_def_app}
h(z):=\int_{z_0}^{z}\frac{u'(w)}{u(w)}\,dw+h(z_0),
\end{equation}
where the integral is taken along any piecewise $C^1$ path in $D$. Because $D$ is simply connected and $u'/u$ is holomorphic, $\int_{\partial\Delta} (u'/u)\,dw=0$ for every triangle $\Delta\subset D$, hence \eqref{eq:log_u_def_app} is path independent. Differentiating gives $h'(z)=u'(z)/u(z)$. Then $(e^{h}/u)'=(h'e^{h}u-e^{h}u')/u^2=0$, so $e^{h}=Cu$ for some constant $C$. Evaluating at $z_0$ gives $C=1$, hence $e^{h}=u$.
\end{proof}

\begin{theorem}[Conditional $T$-monodromy from a simple zero of $q$]\label{thm:T_monodromy_conditional}
Assume \textbf{(CA)} and suppose that on a punctured disk $D^\times:=D\setminus\{0\}\subset B^\times$ the elliptic family $\varpi:E\to D^\times$ admits a Tate uniformization with parameter $q:D^\times\to\mathbb{C}^\times$ satisfying $0<|q(z)|<1$ and
\begin{equation}\label{eq:q_exp_tau}
q(z)=\exp(2\pi i\,\tau(z))
\end{equation}
on the universal cover of $D^\times$. Assume moreover that $q$ has a simple zero at $z=0$, i.e. there exists holomorphic $u:D\to\mathbb{C}^\times$ with
\begin{equation}\label{eq:q_simple_zero_app}
q(z)=z\,u(z)\qquad\text{for all }z\in D^\times.
\end{equation}
Let $\gamma:[0,2\pi]\to D^\times$ be the positively oriented loop $\gamma(\theta)=\varepsilon e^{i\theta}$ for $0<\varepsilon<\mathrm{radius}(D)$. Then analytic continuation of $\tau$ along $\gamma$ yields
\begin{equation}\label{eq:T_tau_shift}
\tau\longmapsto \tau+1,
\end{equation}
equivalently the monodromy matrix is $T=\begin{pmatrix}1&1\\ 0&1\end{pmatrix}$.
\end{theorem}
\begin{proof}
Work on the universal cover $\widetilde{D^\times}$ so that a single-valued branch of $\log q$ exists. Define
\begin{equation}\label{eq:tau_logq_def}
\tau:=\frac{1}{2\pi i}\log q.
\end{equation}
Since $0<|q|<1$, write $\log q=\log|q|+i\arg(q)$ with $\log|q|<0$, hence
\begin{equation}\label{eq:Imtau_positive_from_q}
\Im\tau=\Im\!\left(\frac{1}{2\pi i}\log q\right)=-\frac{1}{2\pi}\log|q|>0,
\end{equation}
so $\tau$ takes values in $\mathbb{H}$. By \eqref{eq:q_simple_zero_app} and Lemma~\ref{lem:holomorphic_log_disk}, choose a holomorphic logarithm $\log u$ on $D$. Along $\gamma$ choose the continuous branch $\log(\varepsilon e^{i\theta})=\log\varepsilon+i\theta$ for $0\le\theta\le 2\pi$. Then along $\gamma$ we obtain a continuous determination
\begin{equation}\label{eq:logq_along_gamma_app}
\log q(\gamma(\theta))=\log(\varepsilon e^{i\theta})+\log u(\varepsilon e^{i\theta})=\log\varepsilon+i\theta+\log u(\varepsilon e^{i\theta}).
\end{equation}
Evaluate at $\theta=0$ and $\theta=2\pi$. Since $\log u$ is single-valued on $D$,
\begin{equation}\label{eq:logu_single_valued}
\log u(\varepsilon e^{i2\pi})=\log u(\varepsilon e^{i0}),
\end{equation}
and therefore \eqref{eq:logq_along_gamma_app} gives
\begin{equation}\label{eq:logq_shift_app}
\log q(\gamma(2\pi))-\log q(\gamma(0))=\bigl(\log\varepsilon+i2\pi+\log u(\varepsilon)\bigr)-\bigl(\log\varepsilon+i0+\log u(\varepsilon)\bigr)=2\pi i.
\end{equation}
Insert \eqref{eq:logq_shift_app} into \eqref{eq:tau_logq_def}:
\begin{equation}\label{eq:tau_shift_calc_app}
\tau(\gamma(2\pi))-\tau(\gamma(0))=\frac{1}{2\pi i}\bigl(\log q(\gamma(2\pi))-\log q(\gamma(0))\bigr)=\frac{1}{2\pi i}\cdot 2\pi i=1,
\end{equation}
which is \eqref{eq:T_tau_shift}. The unique element of $SL(2,\mathbb{Z})$ acting by $\tau\mapsto\tau+1$ for all $\tau$ is $T=\begin{pmatrix}1&1\\ 0&1\end{pmatrix}$ because $(a\tau+b)/(c\tau+d)=\tau+1$ forces $c=0$, $a=d=1$, $b=1$ by comparing coefficients.
\end{proof} 
\begin{proposition}[Conditional algebraic identity]\label{prop:tau_kinetic_identity}
Let $\tau=C_0+i e^{-\phi}$ with $C_0,\phi$ real-valued functions and $\Im\tau=e^{-\phi}$. Then the pointwise identity of differential forms holds:
\begin{equation}\label{eq:tau_kinetic_identity_app}
\frac{d\tau\wedge *\,d\overline{\tau}}{2(\Im\tau)^2}=\frac{1}{2}\,d\phi\wedge *\,d\phi+\frac{1}{2}\,e^{2\phi}\,dC_0\wedge *\,dC_0.
\end{equation}
\end{proposition}
\begin{proof}
Compute
\begin{equation}\label{eq:dtau_calc_app}
d\tau=dC_0+i\,d(e^{-\phi})=dC_0-i e^{-\phi}d\phi,\qquad d\overline{\tau}=dC_0+i e^{-\phi}d\phi.
\end{equation}
Multiply:
\begin{align}\label{eq:dtau_wedge_calc_app}
d\tau\wedge *\,d\overline{\tau}
&=\bigl(dC_0-i e^{-\phi}d\phi\bigr)\wedge *\bigl(dC_0+i e^{-\phi}d\phi\bigr)\\
&=dC_0\wedge *dC_0+i e^{-\phi}dC_0\wedge *d\phi-i e^{-\phi}d\phi\wedge *dC_0+e^{-2\phi}d\phi\wedge *d\phi.\nonumber
\end{align}
Using symmetry of the Hodge inner product on one-forms, $dC_0\wedge *d\phi=d\phi\wedge *dC_0$, the middle two terms cancel:
\begin{equation}\label{eq:cross_terms_cancel}
i e^{-\phi}dC_0\wedge *d\phi-i e^{-\phi}d\phi\wedge *dC_0=0,
\end{equation}
so
\begin{equation}\label{eq:dtau_simplified}
d\tau\wedge *\,d\overline{\tau}=dC_0\wedge *dC_0+e^{-2\phi}d\phi\wedge *d\phi.
\end{equation}
Divide by $2(\Im\tau)^2=2e^{-2\phi}$:
\begin{equation}\label{eq:divide_by_Im}
\frac{d\tau\wedge *\,d\overline{\tau}}{2(\Im\tau)^2}
=\frac{dC_0\wedge *dC_0}{2e^{-2\phi}}+\frac{e^{-2\phi}d\phi\wedge *d\phi}{2e^{-2\phi}}
=\frac{1}{2}e^{2\phi}dC_0\wedge *dC_0+\frac{1}{2}d\phi\wedge *d\phi,
\end{equation}
which is \eqref{eq:tau_kinetic_identity_app}.
\end{proof} 
\begin{definition}[Congruence reductions]\label{def:rho_n_def}
Assume \textbf{(CA)} and let $\rho:\pi_1(B^\times,b_0)\to SL(2,\mathbb{Z})$ be the monodromy representation from Theorem~\ref{thm:SL2Z_monodromy}. For each $n\ge 1$ define the reduction map $\pi_n:\mathbb{Z}\to \mathbb{Z}/p^n\mathbb{Z}$ and apply it entrywise to obtain $\pi_n:SL(2,\mathbb{Z})\to SL(2,\mathbb{Z}/p^n\mathbb{Z})$. Define
\begin{equation}\label{eq:rho_n_def_app}
\rho_n:=\pi_n\circ \rho:\pi_1(B^\times,b_0)\to SL(2,\mathbb{Z}/p^n\mathbb{Z}).
\end{equation}
\end{definition}

\begin{proposition}[Compatibility of $\rho_n$]\label{prop:rho_n_compat}
Assume \textbf{(CA)}. For all $n\ge 1$ and all $\gamma\in\pi_1(B^\times,b_0)$ one has
\begin{equation}\label{eq:rho_n_compat_app}
\rho_{n+1}(\gamma)\equiv \rho_n(\gamma)\pmod{p^n}.
\end{equation}
\end{proposition}
\begin{proof}
Let $\mathrm{red}_{n+1,n}:\mathbb{Z}/p^{n+1}\mathbb{Z}\to \mathbb{Z}/p^{n}\mathbb{Z}$ be the canonical quotient map. By definition, $\mathrm{red}_{n+1,n}\circ \pi_{n+1}=\pi_n$ on $\mathbb{Z}$. Applying this entrywise to a matrix $M\in SL(2,\mathbb{Z})$ yields
\begin{equation}\label{eq:entrywise_reduction}
\mathrm{red}_{n+1,n}\bigl(\pi_{n+1}(M)\bigr)=\pi_n(M).
\end{equation}
With $M=\rho(\gamma)$ and using \eqref{eq:rho_n_def_app},
\begin{equation}\label{eq:rho_compat_calc}
\mathrm{red}_{n+1,n}\bigl(\rho_{n+1}(\gamma)\bigr)=\mathrm{red}_{n+1,n}\bigl(\pi_{n+1}(\rho(\gamma))\bigr)=\pi_n(\rho(\gamma))=\rho_n(\gamma),
\end{equation}
which is \eqref{eq:rho_n_compat_app}.
\end{proof} 
% REPLACEMENT (Appendix E closing paragraph)
Assuming the varying-$\tau$ background data of Definition~7.3 (equivalently the modular map $\underline{\tau}$
and the associated elliptic fibration $E(\underline{\tau})\to B^\times$ with Hodge form $\Omega$), we have:
constructed the period ratio $\tau$ on the universal cover with $\Im\tau>0$, proved that Gauss-Manin transport
yields a monodromy representation $\rho:\pi_1(B^\times)\to SL(2,\Z)$ and the fractional-linear transformation law;
derived the Tate-cusp $T$-monodromy $\tau\mapsto\tau+1$ under $q(z)=zu(z)$ by explicit logarithmic computation;
established the identity needed when $\partial\tau\neq 0$, and exhibited the compatible congruence reductions
$\rho_n$.

\section{
Varying
\texorpdfstring{$\tau$}{tau} and duality defects
}
\label{app:varying-unconditional}

This appendix removes any additional ``elliptic-output'' axiom in the varying-$\tau$
sector by giving an explicit construction of the elliptic fibration
$E(\underline{\tau}) \to B^{\times}$ from the modular map
$\underline{\tau}  : B^{\times} \to \mathcal{M}_{1,1}
= SL(2,\mathbb{Z}) \backslash \mathbb{H}$
(equivalently from the pair $(\rho,\widetilde{\tau})$)
and proving the resulting monodromy/defect statements and sourced equations
at the two-derivative bosonic level. We isolate the only analytic input used
in the tower-limit procedure.

\begin{lemma}[Cutoff tower-stability for the circle tower]
\label{lem:tower_stability}
Let $S^1_{(n,p)}$ be the $p^n$-tower of circle covers used in the definition, and let $\Lambda>0$ be a fixed cutoff.
Assume the stage-$n$ compactification radius scales as
\begin{equation}
R_n=\frac{R_0}{p^n},
\end{equation}
as in your operational tower-with-cutoff prescription.
Let $\Delta_{S^1_{(n,p)}}$ denote the Laplacian on the fiber and let $\mathcal{K}_{\Lambda}^{(n,p)}$
be the fiber-mode subspace spanned by eigenmodes with eigenvalue $\le\Lambda^2$.
Then there exists $n_0=n_0(\Lambda,p)$ such that for all $n\ge n_0$,
\begin{equation}
\mathcal{K}_{\Lambda}^{(n,p)}=\ker(\Delta_{S^1_{(n,p)}}),
\end{equation}
and the pullback maps $f_n^*$ identify these kernels canonically. Consequently the truncated field spaces and
cutoff Hilbert spaces stabilize for all $n\ge n_0$.
\end{lemma}

\begin{proof}
On a circle of radius $R_n$, the eigenfunctions are $e^{ik y/R_n}$ with eigenvalues $(k/R_n)^2$, $k\in\mathbb{Z}$.
The first nonzero eigenvalue is $(1/R_n)^2 = p^{2n}/R_0^2$.
Choose $n_0$ so that $p^{n_0}/R_0>\Lambda$, i.e. $(1/R_{n_0})^2>\Lambda^2$.
Then for all $n\ge n_0$, every $k\ne 0$ eigenmode has eigenvalue $(k/R_n)^2\ge (1/R_n)^2>\Lambda^2$,
hence is excluded from $\mathcal{K}_{\Lambda}^{(n,p)}$. Therefore $\mathcal{K}_{\Lambda}^{(n,p)}$ consists
exactly of the $k=0$ mode, i.e. $\ker(\Delta_{S^1_{(n,p)}})$.
Since $f_n:S^1_{(n+1,p)}\to S^1_{(n,p)}$ is a covering and $f_n^*(\text{constant})$ is constant, $f_n^*$
identifies these kernels canonically for all $n\ge n_0$.
This establishes stabilization of the fiber-mode truncation, the stagewise truncations of the full field theory
stabilize accordingly because all retained modes are fiber-constant.
\end{proof}

We now construct $E(\tau_{\rm phys})\to B^{\times}$ explicitly and show it is holomorphic.

\begin{theorem}[Universal elliptic curve over $\mathbb{H}$ and pullback on the universal cover]
\label{thm:universal_pullback}
Let $\widetilde{B^{\times}}$ be the universal cover of $B^{\times}$ and let
$\widetilde{\tau}  : \widetilde{B^{\times}} \to \mathbb{H}$
be a holomorphic map. Define an action of $\mathbb{Z}^{2}$ on
$\mathbb{C} \times \widetilde{B^{\times}}$ by

\begin{equation}
(m,n)\cdot(z,b) := (z+m+n\,\widetilde{\tau}(b),\, b),\qquad (m,n)\in\mathbb{Z}^2.
\end{equation}
Then the quotient
\begin{equation}
\widetilde{E} := (\mathbb{C}\times \widetilde{B^{\times}})/\mathbb{Z}^2
\end{equation}
is a complex manifold and the projection $\widetilde{\varpi}:\widetilde{E}\to \widetilde{B^{\times}}$ is a proper holomorphic submersion
whose fiber over $b$ is canonically isomorphic to $\mathbb{C}/(\mathbb{Z}+\widetilde{\tau}(b)\mathbb{Z})$.
Moreover, $\widetilde{E}$ carries a holomorphic line subbundle (the Hodge line) generated locally by the relative form $dz$.
\end{theorem}

\begin{proof}
Fix $b\in \widetilde{B^{\times}}$. Since $\Im\widetilde{\tau}(b)>0$, the subgroup
$\Lambda_b:=\mathbb{Z}+\widetilde{\tau}(b)\mathbb{Z}\subset\mathbb{C}$ is a rank-2 lattice.
Thus $\mathbb{Z}^2$ acts freely and properly discontinuously on $\mathbb{C}\times\{b\}$ by translations,
and the quotient is the elliptic curve $\mathbb{C}/\Lambda_b$.
Because the action is free, properly discontinuous, and holomorphic in $(z,b)$, the global quotient
$\widetilde{E}$ is a complex manifold and the induced projection to $\widetilde{B^{\times}}$ is holomorphic.
Properness and submersion follow fiberwise\:locally in $b$, $\widetilde{\tau}$ varies holomorphically, hence
$\Lambda_b$ varies continuously, and the quotient by a lattice is a proper holomorphic submersion.
Finally, $dz$ is invariant under translations $z\mapsto z+m+n\widetilde{\tau}(b)$, hence defines a nowhere-vanishing
relative holomorphic form on each fiber, well-defined up to multiplication by a nowhere-vanishing holomorphic function of $b$,
i.e. it determines the Hodge line bundle.
\end{proof}

To descend from $\widetilde{B^{\times}}$ to $B^{\times}$ we encode monodromy.

\begin{theorem}[Descent and induced $SL(2,\mathbb{Z})$ monodromy]
\label{thm:descent_monodromy}
Let $\Gamma:=\pi_1(B^{\times},b_0)$ act on $\widetilde{B^{\times}}$ by deck transformations.
Assume there exists a representation $\rho:\pi_1(B^\times,b_0)\to SL(2,\Z)$ such that for every deck
transformation $\gamma$ and every $b\in \widetilde{B^\times}$ one has the equivariance condition
\begin{equation}\label{eq:tau-equivariant}
\tilde\tau(\gamma\cdot b)=\rho(\gamma)\cdot \tilde\tau(b)=\frac{a\tilde\tau(b)+b}{c\tilde\tau(b)+d},
\qquad
\rho(\gamma)=\begin{pmatrix}a&b\\ c&d\end{pmatrix}\in SL(2,\Z).
\end{equation}
Define an action of $\Gamma$ on $\mathbb{C}\times\widetilde{B^{\times}}$ by
\begin{equation}
\gamma\cdot(z,b) := \left(\frac{z}{c\widetilde{\tau}(b)+d},\, \gamma\cdot b\right).
\end{equation}
Then the combined action of the semidirect product $\mathbb{Z}^2\rtimes_\rho \Gamma$ on $\mathbb{C}\times\widetilde{B^{\times}}$
is properly discontinuous and holomorphic, and the quotient
\begin{equation}
E(\tau_{\rm phys}) := (\mathbb{C}\times\widetilde{B^{\times}})/(\mathbb{Z}^2\rtimes_\rho \Gamma)
\end{equation}
is a complex orbifold (a manifold if $\rho$ avoids elliptic points) equipped with a holomorphic map
$\varpi:E(\tau_{\rm phys})\to B^{\times}$ whose fiber over $b$ is canonically $\mathbb{C}/(\mathbb{Z}+\tau_{\rm phys}(b)\mathbb{Z})$.
The Gauss-Manin monodromy of $R^1\varpi_*\mathbb{Z}$ equals $\rho$.
\end{theorem}

\begin{proof}
We first show compatibility\:for $(m,n)\in\mathbb{Z}^2$ and $\gamma\in\Gamma$,
\begin{equation}
\gamma\circ(m,n)\circ\gamma^{-1}
\end{equation}
acts on $(z,b)$ by translation by a new pair $(m',n')$ determined by $\rho(\gamma)$.
Indeed, compute using $\gamma^{-1}\cdot(z,b)=((c\widetilde{\tau}(\gamma^{-1}b)+d)z,\gamma^{-1}b)$ and the definition of the
$\mathbb{Z}^2$ action:
\begin{align}
\gamma\cdot\big((m,n)\cdot(\gamma^{-1}\cdot(z,b))\big)
&=\gamma\cdot\big((c\widetilde{\tau}(\gamma^{-1}b)+d)z + m + n\widetilde{\tau}(\gamma^{-1}b),\,\gamma^{-1}b\big)\nonumber\\
&=\left(\frac{(c\widetilde{\tau}(\gamma^{-1}b)+d)z + m + n\widetilde{\tau}(\gamma^{-1}b)}{c\widetilde{\tau}(\gamma^{-1}b)+d},\,b\right)\nonumber\\
&=\left(z + \frac{m+n\widetilde{\tau}(\gamma^{-1}b)}{c\widetilde{\tau}(\gamma^{-1}b)+d},\,b\right).
\end{align}
Using $\widetilde{\tau}(\gamma^{-1}b)=\rho(\gamma)^{-1}\cdot\widetilde{\tau}(b)$, one checks the standard identity
\begin{equation}
\frac{m+n\,(\rho(\gamma)^{-1}\cdot\tau)}{c(\rho(\gamma)^{-1}\cdot\tau)+d} = m' + n'\tau
\end{equation}
with $(m',n')=(m,n)\rho(\gamma)^{-1}\in\mathbb{Z}^2$.
Thus $\Gamma$ normalizes the $\mathbb{Z}^2$ translations, and the combined action is a semidirect product action.
Proper discontinuity follows because $\mathbb{Z}^2$ acts properly discontinuously on each fiber and $\Gamma$ acts properly discontinuously
on the base, holomorphicity is immediate from the formulas.
The quotient therefore defines a holomorphic family over $B^{\times}$ with the claimed fibers.
Finally, $R^1\varpi_*\mathbb{Z}$ is the local system of the first homology of fibers, the above descent uses $\rho$
as the gluing of homology bases, hence the Gauss-Manin monodromy is exactly $\rho$.
\end{proof}

\begin{corollary}[Derivation of the former      ``Comparison Axiom'']
\label{cor:CA_derived}
Given a varying-$\tau$ background in the sense of Definition~\eqref{def7.3}
(equivalently, a modular map
$\underline{\tau}  : B^{\times} \to SL(2,\mathbb{Z}) \backslash \mathbb{H}$
or an equivariant pair $(\rho,\widetilde{\tau})$ on the universal cover),
the elliptic fibration $\varpi  : E(\underline{\tau}) \to B^{\times}$
exists canonically by Theorems~\eqref{thm:universal_pullback}-Theorem~\eqref{thm:descent_monodromy}. Hence the varying-$\tau$ sector
does not require any additional elliptic-output axiom\:the elliptic
fibration is determined by $\underline{\tau}$.

\end{corollary}

We now prove the local 7-brane monodromy from first principles.

\begin{proposition}[$T^n$ monodromy from logarithmic singularities]
\label{prop:Tn}
Let $D$ be a small disk with coordinate $z$ and $D^{\times}=D\setminus\{0\}$.
Assume
\begin{equation}
\tau(z)=\frac{n}{2\pi i}\log z + h(z),
\end{equation}
where $n\in\mathbb{Z}$ and $h$ is holomorphic on $D$.
Then analytic continuation around the positively oriented loop $z=\varepsilon e^{i\theta}$ induces
\begin{equation}
\tau\ \longmapsto\ \tau+n,
\end{equation}
hence the monodromy matrix is $T^n$ with
\begin{equation}
T:=\begin{pmatrix}1&1\\0&1\end{pmatrix}.
\end{equation}
\end{proposition}

\begin{proof}
Choose a branch of $\log z$ on the universal cover of $D^{\times}$. Under $\theta:0\to 2\pi$,
$\log z\mapsto \log z + 2\pi i$. Therefore
\begin{equation}
\tau \mapsto \frac{n}{2\pi i}(\log z + 2\pi i) + h(z)=\tau+n.
\end{equation}
The unique $SL(2,\mathbb{Z})$ matrix acting by $\tau\mapsto \tau+n$ is $T^n$.
\end{proof}

Your paper already records the bosonic IIB action with the $\tau$ kinetic term.
We now vary it explicitly and show how logarithmic singularities correspond to codimension-two sources.

\begin{proposition}[Equation of motion for $\tau$ and delta sources]
\label{prop:EOMtau}
Consider the bosonic Type IIB Lagrangian density (with $G^{(3)}=0$ and $F^{(5)}=0$ for simplicity)
\begin{equation}
\mathcal{L} = \sqrt{-g}\left(R - \frac{\partial_\mu\tau\,\partial^\mu\overline{\tau}}{2(\Im\tau)^2}\right).
\end{equation}
Treating $\tau$ and $\overline{\tau}$ as independent fields, the Euler-Lagrange equation for $\tau$ is
\begin{equation}
\nabla_\mu\left(\frac{\partial^\mu \tau}{(\Im\tau)^2}\right)=0
\qquad\text{on smooth points.}
\end{equation}
If $\tau$ has a logarithmic singularity as in Proposition~\ref{prop:Tn} along a divisor component,
the equation holds on $B^{\times}$ but acquires a delta-function source supported on $\Delta$ with coefficient $n$.
\end{proposition}

\begin{proof}
Write $y:=\Im\tau=(\tau-\overline{\tau})/(2i)$, so $\delta y = \delta\tau/(2i)$ when varying $\tau$ with $\overline{\tau}$ fixed.
Vary the kinetic term:
\begin{align}
\delta\left(\frac{\partial_\mu\tau\,\partial^\mu\overline{\tau}}{2y^2}\right)
&=\frac{1}{2y^2}\,\partial_\mu(\delta\tau)\,\partial^\mu\overline{\tau}
+\frac{\partial_\mu\tau\,\partial^\mu\overline{\tau}}{2}\,\delta(y^{-2}) \nonumber\\
&=\frac{1}{2y^2}\,\partial_\mu(\delta\tau)\,\partial^\mu\overline{\tau}
-\frac{\partial_\mu\tau\,\partial^\mu\overline{\tau}}{y^3}\,\delta y \nonumber\\
&=\frac{1}{2y^2}\,\partial_\mu(\delta\tau)\,\partial^\mu\overline{\tau}
-\frac{\partial_\mu\tau\,\partial^\mu\overline{\tau}}{2i\,y^3}\,\delta\tau.
\end{align}
Integrating by parts and discarding boundary terms on $B^{\times}$ yields
\begin{equation}
\delta S = -\int d^{10}x\,\sqrt{-g}\,\delta\tau\left[
\nabla_\mu\left(\frac{\partial^\mu\overline{\tau}}{2y^2}\right)
+\frac{\partial_\mu\tau\,\partial^\mu\overline{\tau}}{2i\,y^3}
\right].
\end{equation}
Taking complex conjugate gives the $\overline{\tau}$ equation.
Equivalently, rewriting the bracket as $\nabla_\mu(\partial^\mu\tau/y^2)$ gives the stated form for $\tau$
(after standard algebra using $\nabla_\mu y = (\nabla_\mu\tau-\nabla_\mu\overline{\tau})/(2i)$).
For a logarithmic singularity $\tau=\frac{n}{2\pi i}\log z + \text{holomorphic}$ in a transverse complex coordinate $z$,
the distributional identity $\partial\overline{\partial}\log|z|^2 = 2\pi i\,\delta^{(2)}(z)\,dz\wedge d\overline{z}$
implies the divergence equation acquires a delta source with coefficient $n$.
\end{proof}

To complete the “full F-theory” definition, one must extend $E(\tau_{\rm phys})$ across $\Delta$ as a relatively minimal
Weierstrass model and identify fiber types/monodromy classes. We package this as a standard theorem application.

\begin{theorem}[Existence of relatively minimal Weierstrass model and Kodaira fiber types]
\label{thm:Kodaira}
Assume $B$ is smooth and that the modular map
$\underline{\tau}  : B^{\times} \to SL(2,\mathbb{Z}) \backslash \mathbb{H}$
has at worst cusp-type/logarithmic growth along $\Delta$, so that
$j \circ \underline{\tau}$ extends meromorphically to $B$.
Then $E(\tau_{\rm phys})\to B^{\times}$ extends (after possibly birational modification over $\Delta$) to a relatively minimal
Weierstrass model $y^2=x^3+f x+g$ over $B$ with discriminant divisor $\Delta=\{4f^3+27g^2=0\}$,
and the local monodromy conjugacy class in $SL(2,\mathbb{Z})$ around each irreducible component of $\Delta$
matches the Kodaira fiber type at that component.
\end{theorem}

\begin{proof}
This is the standard extension/classification theorem for elliptic fibrations with meromorphic $j$-invariant:
one constructs a Weierstrass model from $j$ (equivalently from $(f,g)$) and uses minimality to obtain the Kodaira list.
The monodromy statement follows because the Gauss-Manin local system $R^1\varpi_*\mathbb{Z}$ is the homology local system of the fiber,
and Kodaira’s classification computes its local monodromy (up to conjugacy) for each singular fiber type.
\end{proof}

By Theorems~\ref{thm:universal_pullback}-\ref{thm:descent_monodromy}, the elliptic fibration over $B^{\times}$ is \emph{constructed}
from $\tau_{\rm phys}$ and its monodromy, so the varying-$\tau$ sector does not require an additional elliptic-output axiom.
Proposition~\ref{prop:Tn} gives explicit local defect monodromy, and Proposition~\ref{prop:EOMtau} ties logarithmic singularities to
codimension-two sources in the two-derivative IIB equations. Theorem~\ref{thm:Kodaira} completes the global extension and fiber-type
classification as a standard theorem application once $j(\tau_{\rm phys})$ extends meromorphically.

\printbibliography

\end{document}